\documentclass[final, onefignum,onetabnum]{siamonline190516}

\usepackage{lipsum}
\usepackage{amsfonts}
\usepackage{graphicx}
\usepackage{epstopdf}
\usepackage{algorithmic}
\ifpdf
  \DeclareGraphicsExtensions{.eps,.pdf,.png,.jpg}
\else
  \DeclareGraphicsExtensions{.eps}
\fi
\usepackage{wrapfig}

\makeatletter
\newcommand*{\addFileDependency}[1]{
  \typeout{(#1)}
  \@addtofilelist{#1}
  \IfFileExists{#1}{}{\typeout{No file #1.}}
}
\makeatother

\newcommand{\SP}{\mathsf{SP}^k X}
\newcommand{\LB}{\mathcal{LB}}
\newcommand{\x}{\mathbf{x}}
\newcommand{\y}{\mathbf{y}}
\newcommand{\supp}{\mathsf{supp}}
\newcommand{\s}{\mathbf{s}}
\newcommand{\Fre}{Fr\'{e}chet\xspace}
\usepackage[utf8]{inputenc}
\usepackage{xspace}
\usepackage{url}
\usepackage{dsfont}
\usepackage{graphicx}
\usepackage{subcaption}
\usepackage{placeins}
\usepackage{tikz}
\usepackage{pgfplots}
\usepackage{subcaption}
\newtheorem{Theorem}{Theorem}
\newtheorem{Proposition}[Theorem]{Proposition}
\newtheorem{Remark}[Theorem]{Remark}
\newtheorem{Observation}[Theorem]{Observation}
\newtheorem{Corollary}[Theorem]{Corollary}
\newtheorem{Lemma}[Theorem]{Lemma}
\newtheorem{Example}[Theorem]{Example}

\theoremstyle{definition}
\newtheorem{Definition}[Theorem]{Definition}
\usepackage{color}

\usepackage{enumitem}
\setlist[enumerate]{leftmargin=.5in}
\setlist[itemize]{leftmargin=.5in}

\newsiamremark{remark}{Remark}
\newsiamremark{hypothesis}{Hypothesis}
\crefname{hypothesis}{Hypothesis}{Hypotheses}
\newsiamthm{claim}{Claim}

\headers{Geometric averages of partitioned datasets}{T.~Needham and T.~Weighill}

\title{Geometric averages of partitioned datasets\thanks{\funding{The second author was supported by NSF grant OIA-1937095.}}}

\author{Tom Needham\thanks{Florida State University 
  (\email{tneedham@fsu.edu}).}
\and Thomas Weighill\thanks{The University of North Carolina at Greensboro 
  (\email{t\_weighill@uncg.edu}).}}

\usepackage{amsopn}

\ifpdf
\hypersetup{
  pdftitle={Geometric averages of partitioned datasets},
  pdfauthor={T. Needham and T. Weighill}
}
\fi

\begin{document}

\maketitle

\begin{abstract}
  We introduce a method for jointly registering ensembles of partitioned datasets in a way which is both geometrically coherent and partition-aware. Once such a registration has been defined, one can group partition blocks across datasets in order to extract summary statistics, generalizing the commonly used order statistics for scalar-valued data. By modeling a partitioned dataset as an unordered $k$-tuple of points in a Wasserstein space, we are able to draw from techniques in optimal transport. More generally, our method is developed using the formalism of local Fr\'{e}chet means in symmetric products of metric spaces. We establish basic theory in this general setting, including Alexandrov curvature bounds and a verifiable characterization of local means. Our method is demonstrated on ensembles of political redistricting plans to extract and visualize basic properties of the space of plans for a particular state, using North Carolina as our main example. 
\end{abstract}



\section{Introduction}

Clustering of data is a fundamental task in unsupervised machine learning which searches for a partitioning of a given dataset which is  optimal with respect to a given objective. This paper introduces statistical methods for the study of ensembles of datasets, such that each dataset comes with a predefined partition into $k$ subsets, say, via some clustering algorithm. In particular, we consider the problem of computing the mean (or \emph{barycenter}) of such an ensemble---this yields a mean of the underlying datasets overlaid with a mean partitioning. This framework has general applications to comparison of clusterings, for instance allowing one to quantify the stability of a clustering algorithm with respect to perturbations of an underlying dataset or to fuse distributed data which has been pre-clustered; see~\cite{wagner2007comparing} for more examples. Our primary motivation is an application to \emph{political redistricting}, where each dataset in the ensemble consists of a \emph{districting plan}, or a partitioning of a given geographical region into districts. Once the mean of an ensemble of partitioned datasets has been computed, the partition blocks of each dataset can be assigned labels $\{1,2,\ldots,k\}$ by registering to the mean, allowing for partition-aware statistical analysis of the ensemble. This generalizes the classical idea of order statistics of an ensemble of scalar-valued datasets (see Example \ref{orderstats} for details). In our redistricting application, this registration allows for direct comparison of demographic and political statistics of districts across different plans.

A $k$-partitioned dataset can be modeled as an unordered $k$-tuple of distributions on the data space; in other words, an unordered $k$-tuple of points in the associated Wasserstein space. With a view toward a general theory, we develop our approach in the context of unordered $k$-tuples of points in an arbitrary metric space $X$. This space of $k$-tuples is referred to as the symmetric product $\SP$ (also called the \emph{sample space} in \cite{harms2020geometry}), and is simply the quotient of the space of ordered $k$-tuples $X^k$ by the order-permuting action of the symmetric group $S_k$. The main examples we are interested in are when $X$ is $\mathbb{R}^n$, Wasserstein space $W_2(\mathbb{R}^n)$, a manifold $M$, or $\mathsf{SP}^m Y$ where $Y$ is one of these spaces (that is, we consider $\mathsf{SP}^k \mathsf{SP}^m Y)$. Our goal is then to study theoretical and computational aspects of the computation of means or barycenters of subsets $S \subset \SP$. We show that under mild assumptions, $\SP$ has curvature unbounded from above (Theorem \ref{thm:curvature}), so that general existence and uniqueness results for barycenters do not directly apply. We can nonetheless characterize local barycenters of subsets $S \subseteq \SP$ (Theorem \ref{stationary}) and we prove that, for many spaces of interest, the labeling of the points in $S$ given by a best matching to a local barycenter is unique (Corollary \ref{onepointcases}). This allows us to implement an algorithm (Algorithm \ref{algorithm}) for computing or approximating local barycenters in $\SP$. 

As was mentioned above, our target application in this paper is \emph{political redistricting}: the process of dividing up a territory into pieces for the purpose of electing representatives. For example, in the United States every state is divided up into a number of Congressional districts roughly proportional to its population, with one member of the U.S.~House of Representatives being elected from each of these districts. Applying our theory and Algorithm \ref{algorithm} produces a new method for visualizing and analyzing large ensembles of computer-generated redistricting plans. The analysis of large ensembles of redistricting plans has recently become prominent in research and litigation surrounding redistricting and gerrymandering. Using our method, we are able to label the districts in thousands of computer-generated redistricting plans in a coherent way and then examine the political and geographic features of the districts assigned to a given label. We demonstrate the value in this approach by comparing Congressional district-level election outcomes for two elections in North Carolina. We also analyze enacted and proposed plans within our framework in a way that complements recent work on quantifying gerrymandering in North Carolina and which answers a clear need for ``local analysis'' \cite{mattingly} of proposed maps.

The outline of the paper is as follows. In Section \ref{sec:basics} we define the symmetric product space and establish some of its geometric properties. In Section \ref{sec:barycenters} we develop the theory of $p$-barycenters in symmetric product spaces. This theory is used to formulate an algorithm for computing local $p$-barycenters. Section \ref{sec:examples} gives some simple examples of this algorithm in practice. Section \ref{sec:redistrict} contains the application to redistricting ensembles. We conclude this introductory section with a survey of related work.

\paragraph*{Optimal transport}
Optimal transport (OT) problems consist of finding the best way to transport a source distribution to a target distribution within a metric space. This problem was first posed by Monge in the eighteenth century \cite{monge1781memoire} and was reformulated in the 1940s by Kantorovich \cite{kantorovich1942translocation} as a linear program, leading to significant progress and interest. Today, OT-based methods are applied in fields including statistics, machine learning, computer graphics and economics---see general references \cite{peyre2019computational,villani2003topics} for details of theoretical and computational aspects of OT.
Optimal transport connects with this paper in two ways. Firstly, Wasserstein space (the space of distributions endowed with an optimal transport distance between them) provides a key example of a space $X$ for which we want to study $\SP$. Secondly, $\SP$ itself can be identified with a subset of Wasserstein space over $X$ (Proposition \ref{prop:isometric_embedding}), so our work can be considered as finding barycenters in a subset of Wasserstein space over a complicated underlying space. Barycenters in Wasserstein space have already been applied in areas such as texture analysis \cite{rabin2011wasserstein}, shape interpolation \cite{solomon2015convolutional} and color transfer \cite{ferradans2014regularized}. Wasserstein barycenters were introduced and studied for Euclidean spaces in \cite{agueh2011barycenters}, sparking a surge of interest in their theory, as well as methods for computing or approximating them \cite{chewi2020gradient,claici2018stochastic,cuturi2014fast,puccetti2020computation, yang2021fast}. Most relevant for the present paper is the theory developed for the manifold setting~\cite{kim2017wasserstein} and the Euclidean discrete case~\cite{anderes2016discrete}, as well as the exact and regularized algorithms in   \cite{cuturi2014fast}. As mentioned above, the work in this paper can be formulated as finding barycenters in some subset of Wasserstein space, but we require more complicated underlying spaces (such as another Wasserstein space) and the use of local barycenters, which forces us to develop new theory specific to these spaces that is not currently found in the literature. A related thread is the theory and computation of barycenters of sets of persistence diagrams~\cite{chowdhury2019geodesics,turner2014frechet}, which treats barycenters in a metric space with curvature unbounded from above.

\paragraph*{Redistricting ensembles} A key question in redistricting research is to determine whether a proposed or enacted redistricting plan is a gerrymander or not---that is, was some agenda other than the basic redistricting requirements of the state driving the line-drawing? A prominent approach in both research and litigation is to compare the plan to a large ensemble of alternatives generated by an algorithm that takes into account some or all of the redistricting criteria for the particular state and level of government \cite{chikina2017assessing,herschlag2017evaluating,Mattingly2018, deford2019recom, deford2019redistricting, bangia2017redistricting, chen2013unintentional, chen2015cutting}. These ensembles are designed to represent the intractably large set of possible alternative plans, and are necessarily very geographically diverse. Typically, these ensembles are analyzed (and compared against the plan being evaluated) at the level of summary statistics -- for example, the number of Republican seats won under historical vote data. Two recent papers have also analyzed spatial characteristics of redistricting ensembles: via graph optimal transport \cite{abrishami2020geometry} and topological data analysis \cite{duchinneedham21}. Our method allows us to combine a geometric perspective in line with these two papers with a classical summary statistics approach.

\paragraph*{Geometry of symmetric product spaces}
The main theoretical object of study in this paper is the $k$-fold symmetric product $\SP$. The recent paper \cite{harms2020geometry} studies the geometry of this space with the $W_p$ metrics defined in the next section. In particular, they show that $\SP$ is a stratified space and that a Fr\'{e}chet mean (with respect to the metric on $X$) of a subset $S \subseteq X$ of size $k$ is a projection of the corresponding point of $\SP$ onto its lowest dimensional stratum. Our work, by contrast, studies barycenters of subsets of $\SP$ relative to the $W_p$ metric on $\SP$. We should also mention the notion of unordered configuration space (see, e.g.,~\cite{ghrist2010configuration}), which is the proper subspace of $\SP$ consisting of points with $k$ distinct entries, and which sometimes appears in applications to robotics.

\section{Symmetric products and the $W_p$ distance}\label{sec:basics}

In this section, we formally define the symmetric product metric and establish some of its basic properties.

\subsection{The $W_p$ metric}\label{sec:symmprod}

We use the notation $\langle k \rangle = \{1,\ldots,k\}$ and denote the $i^{th}$ coordinate of an (ordered) tuple $\mathbf{v}$ by $\mathbf{v}_i$. For context, we recall the definition of Wasserstein distance. 

\begin{Definition}
Let $X$ be a metric space on which every finite Borel measure is a Radon measure. Given two Borel probability measures $\alpha$ and $\beta$, the \emph{$p$-Wasserstein distance} is defined by
\[
W_p(\alpha, \beta) = \left( \inf_{\gamma \in \mathcal{U}(\alpha,\beta)} \int d(x,y)^p d\gamma(x,y) \right)^{1/p},
\]
where $\mathcal{U}(\alpha,\beta)$ is the set of Borel measures on $X^2$ with marginals $\alpha$ and $\beta$. 
\end{Definition}

If $\alpha = \sum_{i=1}^m \mathbf{a}_i \delta_{x_i}$ and $\beta = \sum_{i=1}^n \mathbf{b}_i \delta_{y_i}$ are discrete, then we can equivalently write
\[
W_p(\alpha, \beta) = \left( \min_{P \in U(\mathbf{a}, \mathbf{b})} \sum_{i,j} d(x_i, y_j)^p P_{ij} \right)^{1/p},
\]
where $P$ ranges over the set $U(\mathbf{a},\mathbf{b})$ of $n\times m$ matrices such that $P \mathds{1} = \mathbf{a}$ and $P^T \mathds{1} = \mathbf{b}$, with $\mathds{1}$ denoting the column vector of the appropriate size with all entries equal to one. Throughout this paper, we will denote by $W_p(X)$ the \emph{$p$-Wasserstein space over $X$} -- that is, the set of Borel probability measures on $X$ with finite $p^{th}$ moment endowed with the $p$-Wasserstein metric. 

We now introduce the main theoretical context for studying unordered data. 

\begin{Definition}
Let $X$ be a set. The symmetric group $S_k$ of bijections $\phi:\langle k \rangle \to \langle k \rangle$ acts on the product set $X^k$ by permuting entries of ordered $k$-tuples. We denote the action of a bijection $\pi \in S_k$ on $\x \in X^k$ by $\pi \x$; this action is given explicitly by the formula $(\pi \x)_i = \x_{\pi(i)}$. The \emph{$k$-fold symmetric product of $X$} is the quotient of $X^k$ by the action of the symmetric group $S_k$, denoted
\[
\SP := X^k/S_k.
\]
We denote equivalence under this $S^k$-action by $\x \sim \x'$ and we denote the equivalence class of $\x \in X^k$ by $[\x]$.
\end{Definition}

If $(X,d)$ is a metric space then  $\SP$ comes with a natural family of metrics.

\begin{Definition}
Let $(X,d)$ be a metric space and let $p \geq 1$. Define the \emph{$p$-Wasserstein distance} on $\SP$ as follows. Let $\mathbf{x} = (x_1,\ldots,x_k)$ and $\mathbf{y} = (y_1,\ldots,y_k)$. Then

\begin{equation}\label{Wp_metric}
    W_p([\x], [\y])^p  := \min_{\pi} \frac{1}{k}\sum_{i=1}^k d(x_i, y_{\pi(i)})^p 
\end{equation}
where the minimum ranges over all bijections $\pi: \langle k \rangle \to \langle k \rangle$. We call a bijection realizing this minimum an \emph{optimal matching from $\x$ to $\y$}.
\end{Definition}

Let $d_{\ell^p}$ denote the \emph{$\ell^p$-metric} on $X^k$, given by
\[
d_{\ell^p}(\x, \y)^p = \sum_{i=1}^{k} d(\x_i,\y_i)^p.
\]
Using the fact that $S_k$ acts by isometries on $d_{\ell^p}$, we have 
\[
W_p([\x], [\y]) = \min_{\x' \sim \x,\ \y'\sim \y} \frac{1}{k} d_{\ell^p}(\x', \y') = \min_{\y'\sim \y} \frac{1}{k} d_{\ell^p}(\x, \y').
\]
This relation makes it easy to check that $W_p$ is indeed a metric, and that it induces the quotient topology on $\SP$ when $X^k$ is endowed with the $\ell^p$ metric.

We next give a precise relationship between the Wasserstein $p$-metric on $\SP$ and the classical Wasserstein metric on $W_p(X)$. The result follows easily from the Birkhoff-von Neumann theorem; see \cite[Lemma 4.7]{harms2020geometry} for details.

\begin{Proposition}\label{prop:isometric_embedding}
The map
\begin{equation}\label{eqn:isometric_embedding}
\begin{aligned}
    \iota: \SP &\to W_p(X), \quad [\x] &\mapsto \frac{1}{k}\sum_{i = 1}^k \delta_{\x_i}
\end{aligned}
\end{equation}
is an isometric embedding of $(\SP,W_p)$ into Wasserstein space $W_p(X)$.
\end{Proposition}

\subsection{Geodesics and curvature in $\SP$}\label{sec:geodesics_and_curvature}

In this section we derive some results about the basic geometry of the metric space $(\SP,W_p)$. Along the way, we recall basic notions of metric geometry, following~\cite{bridson2013metric,burago2001course}.

Let $(X,d)$ be a metric space. A \emph{constant speed geodesic} in $X$ is a map $[a,b] \to X$ of some interval which is an isometric embedding up to a  multiplicative constant. The space $X$ is a \emph{geodesic space} if any two points $x,y \in X$ can be joined by a constant-speed geodesic path $\gamma:[0,1] \to X$ of length $d(x,y)$---that is, a path satisfying
\[
d(\gamma(s),\gamma(t)) = d(x,y) \cdot |s-t|, \qquad \gamma(0) = x, \; \gamma(1) = y.
\]
Such a path is called a \emph{minimal geodesic joining $x$ to $y$}. A geodesic $\gamma:[a,b] \to X$ \emph{branches at $t_0 \in (a,b)$} if there exists another geodesic $\bar{\gamma}:[a,b] \to X$ such that $\gamma|_{[a,t_0]} = \bar{\gamma}|_{[a,t_0]}$ but $\gamma \neq \bar{\gamma}$ on some interval $(t_0,t_0 + \varepsilon)$. If $X$ has no branching geodesics, we say that $X$ is \emph{non-branching}.

\begin{Proposition}\label{prop:geometric_facts}
Let $(X,d)$ be a geodesic metric space and endow $\SP$ with the $W_p$ metric.
\begin{enumerate}
    \item\label{item:isometric_embedding} The map 
                \begin{equation}\label{eqn:isometric_embedding_X}
                    \rho: X \to \SP, \qquad 
                    x \mapsto (x,\ldots,x)
                \end{equation}
                is an isometric embedding of $(X,d)$ into $(\SP,W_p)$.
    \item\label{item:geodesics} Let $\pi$ be an optimal matching of $\x,\y \in \SP$ and let
        $\gamma_i:[0,1] \to X$ be a minimal geodesic between $x_i$ and $y_{\pi(i)}$ in $X$. Define
        \[
        \gamma:[0,1] \to \SP, \qquad t \mapsto [\gamma_1(t),\ldots,\gamma_k(t)].
        \]
        Then $\gamma$ is a minimal geodesic in $(\SP,W_p)$ joining $[\x]$ to $[\y]$ for any $p \in [1,\infty]$. In particular, $\SP$ is a geodesic space. 
\end{enumerate}
\end{Proposition}

\begin{Remark}
The map $X \to W_2(X)$ taking $x$ to the Dirac measure $\delta_x$
is well known to be an isometric embedding \cite[Proposition 2.10]{sturm2006geometry}. This map factors as the composition $\iota \circ \rho$, where $\rho:X \to \SP$ is the isometric embedding defined above in \eqref{eqn:isometric_embedding_X} and $\iota:\SP \to W_2(X)$ is the isometric embedding defined in \eqref{eqn:isometric_embedding}. The image of $\rho$ is the lowest dimensional stratum (referred to as the 1-skeleton) in the stratified space structure of $\SP$ described in \cite{harms2020geometry}.
\end{Remark}

\begin{proof}
Point \ref{item:isometric_embedding} follows by a simple computation: for $x,x' \in X$, we have
\[
W_p(\rho(x),\rho(x'))^p = \min_{\pi} \frac{1}{k} \sum_{i=1}^k d\left(\rho(x)_i,\rho(x')_{\pi(i)}\right)^p = \frac{1}{k} \sum_{i=1}^k d(x,x')^p = d(x,x')^p.
\]

To prove point \ref{item:geodesics}, let $D_i := d(x_i,y_{\pi(i)})$ and $D := W_p([\x],[\y])$. By standard arguments, it suffices to show that for all $0 \leq s \leq t \leq 1$,
\[
W_p(\gamma(s),\gamma(t)) \leq (t-s) D
\]
(see, e.g., \cite[Lemma 1.3]{chowdhury2018explicit}). This is straightforward: for $p \in [1,\infty)$ we have
\begin{align*}
    W_p(\gamma(s),\gamma(t))^p &= \min_{\widetilde{\pi}} \sum_{i=1}^k d(\gamma_i(s),\gamma_{\widetilde{\pi}(i)}(t))^p  \\
    &\leq \sum_{i=1}^k d(\gamma_i(s),\gamma_i(t))^p = \sum_{i=1}^k (t-s)^p D_i^p   = (t-s)^p D^p,
\end{align*}
and the $p = \infty$ case is similar.
\end{proof}

The rest of this section deals with curvature of metric spaces, in the sense of Alexandrov. For $\kappa \in \mathbb{R}$, let $M_\kappa$ denote the $2$-dimensional space form of constant curvature $\kappa$, with metric denoted $d_\kappa$ and metric diameter denoted $D_\kappa$. For three points $x,y,z$ in a metric space $(X,d)$, one can always find points $\bar{x},\bar{y},\bar{z} \in M_\kappa$ with $d(x,y) = d_\kappa(\bar{x},\bar{y})$, $d(y,z) = d_\kappa(\bar{y},\bar{z})$ and $d(z,x) = d_\kappa(\bar{z},\bar{x})$. We say that $\bar{x},\bar{y},\bar{z}$ give a \emph{comparison triangle} for $x,y,z$. A geodesic metric space $(X,d)$ is said to have \emph{curvature bounded below} (respectively, \emph{above}) by $\kappa$ if for any three points $x,y,z$ such that $d(x,y) + d(y,z) + d(z,x) \leq 2D_\kappa$ and for any minimal geodesic $\gamma:[0,1] \to X$ joining $x$ to $y$, it holds that $d(z,\gamma(t)) \geq d_\kappa(\bar{z},\bar{\gamma}(t))$ (respectively, $d(z,\gamma(t)) \leq d_\kappa(\bar{z},\bar{\gamma}(t))$)
for all $t \in [0,1]$, where $\bar{x},\bar{y},\bar{z}$ give a comparison triangle for $x,y,z$ in $M_\kappa$ and $\bar{\gamma}:[0,1] \to M_\kappa$ is a minimizing geodesic joining $\bar{x}$ to $\bar{y}$. 

A geodesic metric space with curvature bounded below by zero is called an \emph{Alexandrov space with nonnegative curvature}. In this case, the relevant inequality can be expressed as
   \begin{equation}\label{eqn:alexandrov_condition}
    d(z,\gamma(t))^2 \geq (1-t)d(z,x)^2 + t d(z,y)^2 - t(1-t)d(x,y)^2
    \end{equation}
for all $x,y,z$, where $\gamma$ is a minimal geodesic joining $x$ to $y$ \cite[Section 2.1]{ohta2012barycenters}. A geodesic metric space with curvature bounded above by $\kappa$ is called a \emph{$\mathrm{CAT}(\kappa)$ space}. We deal below with Alexandrov spaces with nonnegative curvature. This category includes most spaces of interest from our data analysis perspective, including Euclidean spaces, complete Riemannian manifolds whose sectional curvature is not everywhere negative and Wasserstein spaces $W_2(X)$, where $X$ is itself an Alexandrov space of nonnegative curvature \cite[Section 2.1]{ohta2012barycenters}.

We now state the main result of this section, which describes Alexandrov curvature bounds for symmetric products. We restrict our attention to the $W_2$ metric in this setting---this is sensible, since even the standard $\ell^p$ space $(\mathbb{R}^n,d_{\ell^p})$ is an Alexandrov space with nonnegative curvature if and only if $p=2$ (in an $\ell^p$-space, one can show that the Alexandrov inequality \eqref{eqn:alexandrov_condition} implies the Parallelogram Law). By similar reasoning, $(\mathbb{R}^k,d_{\ell^p})$ is $\mathrm{CAT}(0)$ if $p=2$ and is otherwise not $\mathrm{CAT}(\kappa)$ for any $\kappa$---see also \cite[Proposition II.1.14]{bridson2013metric}.

\begin{Theorem}\label{thm:curvature}
Let $(X,d)$ be a geodesic metric space and endow $\SP$ with the $W_2$ metric.
\begin{enumerate}
    \item\label{item:alexandrov} $X$ is an Alexandrov space with nonnegative curvature if and only if $\SP$ is.
    \item\label{item:CAT} If $X$ is not a one point space, a 1-manifold or a 1-manifold with boundary then $\SP$ is not $\mathrm{CAT}(\kappa)$ for any $\kappa$.
\end{enumerate}
\end{Theorem}

Theoretical results and computational tools regarding barycenters in $\mathrm{CAT}(\kappa)$ spaces exist in the literature~\cite{yokota2016convex}; for example, barycenters in $\mathrm{CAT}(0)$ spaces are unique~\cite{sturm2003probability}. This theorem indicates that these methods cannot be generally applied to the symmetric product spaces of interest, motivating the new theory developed in Section \ref{sec:barycenters}. Results which are similar to point \ref{item:alexandrov} are established for Wasserstein spaces $W_2(X)$ in \cite[Proposition 2.10]{sturm2006geometry} and \cite[Theorem A.8]{lott2009ricci}. It is well known that Wasserstein spaces do not inherit upper curvature bounds---see, e.g. \cite[Remark 2.10]{bertrand2012geometric}, which shows that if $X$ is $\mathrm{CAT}(0)$ then $W_2(X)$ is not, unless $X$ is a isometric to an interval. A result in a similar spirit to point \ref{item:CAT} is proved for the space of persistence diagrams (in the context of topological data analysis) in \cite[Proposition 2.4]{turner2014frechet}.

\begin{proof}[Proof of Theorem \ref{thm:curvature}]

If $(\SP,W_2)$ has nonnegative curvature then so does $(X,d)$, since $X$ embeds isometrically in $\SP$ via the map $\rho$ defined in \eqref{eqn:isometric_embedding_X}. Now suppose that $X$ has nonnegative curvature. Then so does $(X,d_{\ell^2})$ \cite[Proposition 4.1]{burago1992ad}. 
The quotient $\SP = X^k/S_k$ of the nonnegatively curved space $(X^k,d_{\ell^2})$ by the isometric action of the finite group $S_k$ is therefore nonnegatively curved by \cite[Corollary 4.6]{burago1992ad}. This completes the proof of  \ref{item:alexandrov}.

To prove point \ref{item:CAT}, it suffices to show that $\SP$ contains arbitrarily close points joined by distinct geodesics \cite[Proposition II.1.4]{bridson2013metric}. Moreover, it suffices to prove the claim for $k=2$.

First suppose that $X$ contains a branching geodesic. Consider the point configurations $\x = (x_1,x_2)$ and $\y = (y_1,y_2)$ lying on the branching geodesic near the branch point, as shown in the lefthand side of Figure \ref{fig:proof_schematics}, where $\varepsilon > 0$ is arbitrarily small. Then both of the possible matchings of $\x$ and $\y$ are optimal with respect to $W_2$. Let $\gamma_i$ (respectively, $\gamma_i'$) be the geodesic from $x_i$ to $y_i$ (respectively, $x_i$ to $y_{\pi(i)}$, where $\pi$ is the non-identity element of $S_2$) whose image is contained in the branching geodesic. Then $\gamma := [\gamma_1,\gamma_2]$ and $\gamma' := [\gamma_1',\gamma_2']$ are both minimizing geodesics in $\mathsf{SP}^2 X$, by Proposition \ref{prop:geometric_facts}, point \ref{item:geodesics}. These geodesics are distinct---for example, $\gamma(2\varepsilon/3) \neq \gamma'(2\varepsilon/3)$. Since $\varepsilon > 0$ was arbitrary, $\mathsf{SP}^2 X$ is not $\mathrm{CAT}(\kappa)$ for any $\kappa$. 

\begin{figure}
    \centering
    \includegraphics[width = 0.10\textwidth]{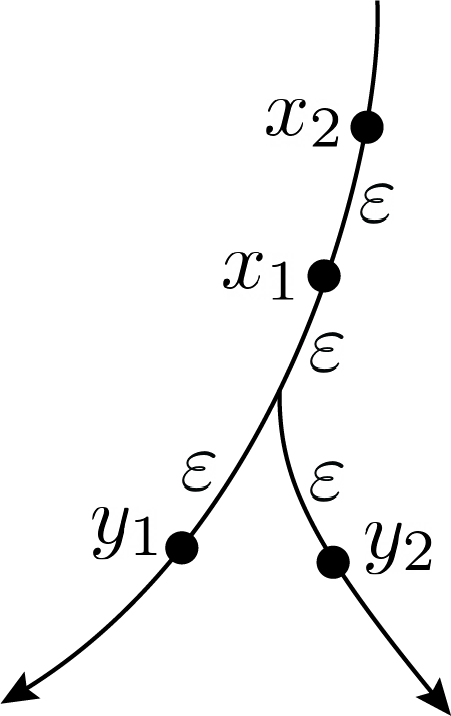}
    \qquad \qquad \includegraphics[width = 0.28\textwidth]{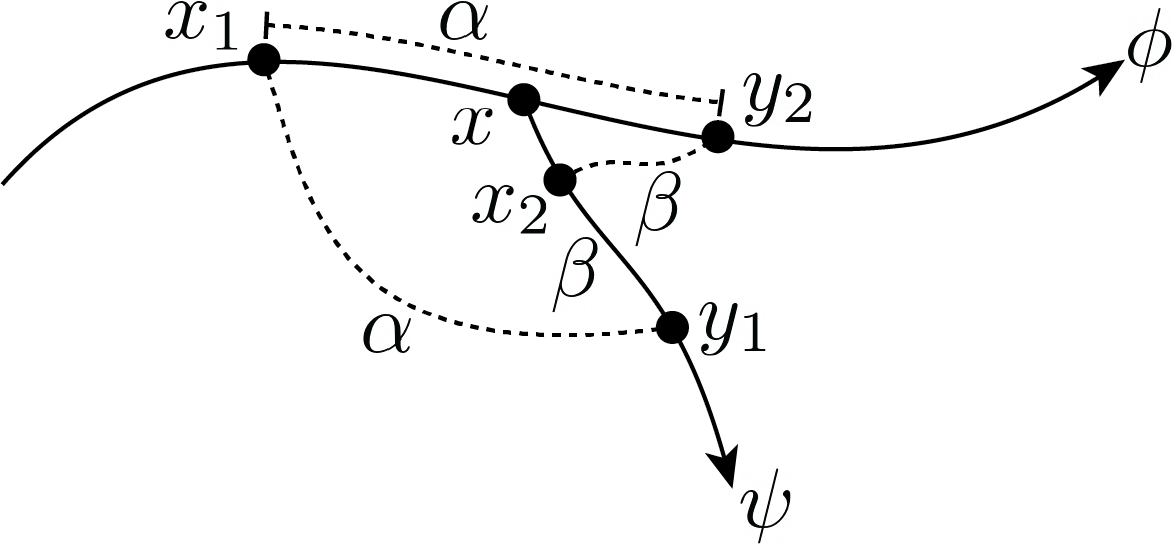}
    \caption{Schematic figures for the proof of Theorem \ref{thm:curvature}.}
    \label{fig:proof_schematics}
\end{figure}
                   
Next, suppose that $(X,d)$ is non-branching. We may also assume that $X$ is $\mathrm{CAT}(\kappa)$ for some $\kappa$---otherwise the claim follows immediately, since $X$ embeds isometrically into $\SP$ by Proposition \ref{prop:geometric_facts}. We construct a configuration of points in $X$ which yields nonunique geodesics between arbitrarily close points in $\mathsf{SP}^2 X$ as follows. We claim that there exists a point $x \in X$ which lies in the relative interior of a geodesic such that the image of the geodesic is not surjective onto any neighborhood of $x$. Indeed, 
let $a \in X$ be a point with no 1-manifold or 1-manifold with boundary chart. Choose a minimizing geodesic $\alpha$ from $a$ to a point $b$ in a small neighborhood of $a$. This geodesic is not onto any neighborhood of $a$, so we may choose a geodesic $\beta$ from the midpoint $m$ of $\alpha$ to another point $c$ not in the image of $\alpha$. If there exists a neighborhood $U$ of $m$ such that $\mathrm{Im}(\alpha) \cap \mathrm{Im}(\beta) \cap U = \{m\}$, then we set $x = m$. Otherwise, the nonbranching and $\mathrm{CAT}(\kappa)$ conditions imply that $a$ lies in the relative interior of $\beta$ and we set $x = a$.  Now we fix an $x \in X$ in the relative interior of a geodesic path $\gamma$ which is not surjective onto any neighborhood of $x$. Without loss of generality, we parameterize $\gamma:[-1,1] \to X$ with $\gamma(0) = x$. Choose another geodesic $\psi:[0,1] \to X$ with $\psi(0) = x$ and $\gamma(1)$ not in the image of $\gamma$. We have that $\mathrm{Im}(\gamma) \cap \mathrm{Im}(\psi) \cap U = \{x\}$ for some small neighborhood $U$ of $x$, by the assumption that $X$ is non-branching and $\mathrm{CAT}(\kappa)$. Set $x_1 := \gamma(t_1)$ for some arbitrarily small $t_1 < 0$, $y_1 := \psi(s_1)$ for some arbitrarily small $s_1 > 0$ and $\alpha := d(x_1,y_1) > 0$. Next, let $s_2 := \alpha - c \cdot |t_1|$, where $c$ is the constant speed of $\gamma$, and $y_2 := \gamma(s_2)$, so that $d(x_1,y_2) = \alpha$.
Consider the path from $y_2$ to $y_1$, through $x$, following the images of the geodesic paths $\gamma$ and $\psi$. By continuity, there is a point $x_2$ on this path such that $d(x_2,y_1) = d(x_2,y_2) =: \beta$. Note that $x_2 \neq x_1$. A schematic of this configuration is shown in the righthand side of Figure \ref{fig:proof_schematics}.

The configuration constructed above has the property that either matching of $\x := (x_1,x_2)$ and $\y := (y_1,y_2)$ is optimal with respect to $W_2$. Let $\gamma_i$ (respectively, $\gamma_i'$) be the geodesic from $x_i$ to $y_i$ (respectively, to $y_{\pi(i)}$), where $\pi$ is the non-identity element of $S_2$), so that $[\gamma_1,\gamma_2]$ and $[\gamma_1',\gamma_2']$ are both minimizing geodesics in $\mathsf{SP}^2 X$ from $[\x]$ to $[\y]$. We claim that these minimizing geodesics are distinct. For sufficiently small $\eta > 0$, we have $\gamma_1(t) \neq \gamma_1'(t)$ for all $t \in (0,\eta)$. This is because $X$ is non-branching, so $\gamma_1$ and $\gamma_2$ can't coincide on any interval (as their endpoints are distinct), while the $\mathrm{CAT}(\kappa)$ condition on $X$ implies that the distinct geodesics $\gamma_1$ and $\gamma_2$ can't intersect arbitrarily close to $x_1$. On the other hand $\gamma_1(0) = x_1 \neq x_2 = \gamma_2'(0)$, so $\gamma_1(t) \neq \gamma_2'(t)$ for all $t$ in some small interval $[0,\eta)$, by continuity. In particular, there exists $t$ such that $\gamma_1(t) \not \in \{\gamma_1'(t),\gamma_2'(t)\}$, which implies $[\gamma_1,\gamma_2] \neq [\gamma_1',\gamma_2']$. Since this construction can be done in an arbitrarily small neighborhood of $x$, this completes the proof.
\end{proof}

\begin{Example}
The symmetric product space $\mathsf{SP}^k \mathbb{R} = \mathbb{R}^k/S_k$, endowed with the $W_2$ metric, is isometric to the space $Y := \{\x \in \mathbb{R}^k \mid \x_1 \leq \x_2 \leq \cdots \leq \x_k\}$, endowed with the $\ell^2$ metric, via the map taking $[\x]$ to its sorted representation---this follows by standard results on one-dimensional optimal transport \cite[Section 3.1]{rachev1998mass}. Thus $\mathsf{SP}^k\mathbb{R}$ is $\mathrm{CAT}(0)$, and a similar argument works for $\SP$ when $X$ is isometric to an interval.
\end{Example}

\section{Barycenters in Symmetric Products}\label{sec:barycenters}

This section introduces our main object of study: barycenters of subsets of $\SP$ with respect to the $W_p$ metric. 

\subsection{Local $p$-barycenters in $\SP$}\label{sec:local}
In this section we characterize local $p$-barycenters in $\SP$ in terms of local barycenters in $X$ and optimal matchings. 

\begin{Definition}\label{barycenter}
Let $X$ be a metric space and $S \subseteq X$ be a finite subset. Then a \emph{$p$-barycenter} of $S$ is a minimizer of the \emph{$p$-\Fre functional} associated to $S$ given by:
\begin{equation}\label{frechet}
    f(x) = f_{S,p}(x) := \sum_{s \in S} d(s,x)^p 
\end{equation}
A \emph{local $p$-barycenter} is a local minimum of the functional $f$.
\end{Definition}

We will use the following technical lemma.

\begin{Lemma} \label{localmatchings}
Let $\x, \y \in X^k$. Then there is an $\varepsilon > 0$ such that for any $\x'$ with $d_{\ell^p}(\x', \x) < \varepsilon$, every optimal matching of $\y$ to $\x'$ is an optimal matching of $\y$ to $\x$.
\end{Lemma}

\begin{proof}
If all matchings are optimal from $\y$ to $\x$, then there is nothing to prove. Otherwise, let $\psi$ be a matching from $\y$ to $\x$ which minimizes $\sum d(y_i, x_{\psi(i)})^p$ among all non-optimal matchings from $\y$ to $\x$, and set 
\[
\delta :=  \left( \sum d(y_i, x_{\psi(i)})^p \right)^{1/p} - k^{1/p} \cdot  W_p([\y],[\x]).
\]
By definition, $\delta > 0$. Suppose that for some $\x'$, there is an optimal matching $\phi'$ from $\y$ to $\x'$ which is not optimal from $\y$ to $\x$. Pick any optimal matching $\phi$ from $\y$ to $\x$. We have that
\[
R := d_{\ell^p}(\y, \phi \x)  \geq d_{\ell^p}(\y, \phi\x') - d_{\ell^p}(\phi\x', \phi\x) = d_{\ell^p}(\y, \phi\x') - d_{\ell^p}(\x', \x).
\]
Since $\phi'$ is an optimal matching from $\y$ to $\x'$ and $\phi$ is not, we have
\begin{align*}
d_{\ell^p}(\y, \phi\x') - d_{\ell^p}(\x', \x) & > d_{\ell^p}(\y, \phi'\x') - d_{\ell^p}(\x', \x) \\ 
&\geq d_{\ell^p}(\y, {\phi'}\x) - d_{\ell^p}({\phi'}\x, {\phi'}\x') - d_{\ell^p}(\x', \x) \\
& = d_{\ell^p}(\y, {\phi'}\x) - 2d_{\ell^p}(\x', \x)  \geq R + \delta - 2d_{\ell^p}(\x', \x).
\end{align*}
Putting everything together gives $R > R + \delta - 2 d_{\ell^p}(\x', \x)$, which implies that $d_{\ell^p}(\x',\x) > \frac{\delta}{2}$. Therefore it is only possible to find an optimal matching of $\y$ to $\x'$ which is not an optimal matching of $\y$ to $\x$ if $d_{\ell^p}(\x',\x) > \delta/2$, so the claim follows for any $\varepsilon < \delta/2$.
\end{proof}

\begin{Definition} \label{stationary}
Let $\x \in X^k$ and let $S \subseteq X^k$ be a finite subset. We say that $\x$ is \emph{stationary} with respect to $S$ if the following holds: for every choice of $\Phi = (\phi_{\s})_{\s \in S}$, where $\phi_{\s}$ is an optimal matching from $\x$ to $\s$, and for every $1 \leq i \leq k$, $\x_i$ is a local $p$-barycenter of the set 
$$
S^{\Phi}_{i} := \{ \s_{\phi_{\s}(i)} \mid \s \in S \}.
$$
\end{Definition}

\begin{Theorem}\label{localbarycenters}
Let $\x \in X^k$ and let $S \subseteq X^k$ be a finite subset. Then $[\x]$ is a local $p$-barycenter in $\SP$ of $[S] = \{ [\s] \mid \s \in S\}$ if and only if $\x$ is \emph{stationary} with respect to $S$. 
\end{Theorem}
\begin{proof}
First suppose that $[\x]$ is a local $p$-barycenter of $[S]$ but $\x$ is not stationary. That means there exist optimal matchings $\Phi = (\phi_{\s})_{\s \in S}$ and an $i$ such that $\x_i$ is not a local $p$-barycenter of $S^{\Phi}_{i}$. Thus for every $\varepsilon > 0$, there is an $\x_i'$ within $\varepsilon$ of $\x_i$ such that $f(x_i') < f(x_i)$ where $f$ is the $p$-\Fre functional associated to $S^{\phi}_{i}$. Let $\x'$ be $\x$ with the $i^{th}$ coordinate replaced by $\x_i'$. We have
\begin{align*}
f(\x') = \sum_{\s \in S} d([\x'], [\s])^p  = \sum_i \sum_{\s \in S} d(\x'_i, \s_{\phi'_{\s}(i)})^p,
\end{align*}
where $\phi'_{\s}$ is an optimal matching from $\x'$ to $\s$. Since these matchings are optimal,
\begin{align*}
f(\x') = \sum_{\s \in S} d([\x'], [\s])^p  & \leq \sum_i \sum_{\s \in S} d(\x'_i, \s_{\phi_{\s}(i)})^p  \\
&< \sum_i \sum_{\s \in S} d(\x_i, \s_{\phi_{\s}(i)})^p = \sum_{\s \in S} d([\x], [\s])^p  = f(\x).
\end{align*}
But since $\x'_i$ is within $\varepsilon$ of $\x_i$, we have that $d([\x], [\x']) \leq \varepsilon$. Thus $[\x]$ cannot be a local $p$-barycenter of $[S]$.

To prove the converse, suppose that $\x$ is stationary. For any choice $\Phi = (\phi_s)_{s\in S}$ of optimal matchings as in Definition \ref{stationary}, $\x$ is a local minimum of the function
\[
f_\Phi(\x) = \sum_i \sum_{\s \in S} d(\x_i, \s_{\phi_{\s}(i)})^p 
\]
because each $\x_i$ is a local $p$-barycenter of the set $S^{\Phi}_{i} = \{ \s_{\phi_{\s}(i)} \mid \s \in S \}$. Since there are only finitely many choices for $\Phi$, we can fix $\varepsilon > 0$ such that $\x$ is a mininum of $f_{\Phi}$ on $B(\x, \varepsilon)$ for any choice of $\Phi$. By Lemma \ref{localmatchings} and the fact that there are finitely many elements in $S$, there is a $\varepsilon > \varepsilon' > 0$ such that if $d(\x, \x') < \varepsilon'$ then $\Phi$ can be chosen such that it is also a set of optimal matchings from $\x'$ to each element of $S$. 

It follows that if $d([\x], [\x']) < \varepsilon'$, then there is a $\x'' \sim \x'$ with $d_{\ell^p}(\x, \x'') < \varepsilon'$ such that $\Phi$ can be chosen to be a set of optimal matchings from $\x''$ to elements of $S$. We have
\begin{align*}
\sum_{s\in S} W_p([\x'], [\s])^p = \sum_{s\in S} W_p([\x''], [\s])^p & = \frac{1}{k}\sum_i \sum_{\s \in S} d(\x''_i, \s_{\phi_{\s}(i)})^p \\
& \geq \frac{1}{k} \sum_i \sum_{\s \in S} d(\x_i, \s_{\phi_{\s}(i)})^p  = \sum_{s\in S} W_p([\x], [\s])^p
\end{align*}
We conclude that $[\x]$ is a local $p$-barycenter for $[S]$.
\end{proof}

\begin{Example}\label{orderstats}
Consider a subset $S \subseteq \mathsf{SP}^k \mathbb{R}$. Then a $2$-barycenter of $S$ consists of points $(a_1,\ldots,a_k)$ where for each $i$, $a_i$ is the average of all the $i^{th}$ largest entries of the points in $S$. Indeed, an optimal $W_2$ matching between two vectors in $\mathbb{R}^k$ is given by pairing up the $i^{th}$ largest values for each $i$, and Theorem \ref{localbarycenters} dictates that each point in the barycenter is a $2$-barycenter (in this case, mean) of the points matched to it. We can therefore view our barycenter method for labeling data as a generalization of order statistics to more complicated spaces.
\end{Example}

\begin{figure}
\begin{algorithm}[H]
\caption{Local $p$-barycenter in $\SP$}\label{algorithm}
\begin{algorithmic}
\STATE $(\mbox{stationary},\bar{\x},n) \gets (\mbox{false},\x_0,|S|)$\; 
\WHILE{stationary = false} 
{
\FOR{$j=1$ \TO $n$}
{
\STATE $R_j \gets \{ \s' \sim \s^j \mid  W_p([\s^j], [\bar{\x}])^p =  \frac{1}{k} d_{\ell^p}(\s', \bar{\x})^p \}$ \;
}
\ENDFOR 
\STATE $\mbox{stationary} \gets \mbox{true}$\;
\FOR {$(\mathbf{t}^1,\ldots,\mathbf{t}^n) \in \prod_j R_j$}
{
\FOR{$i = 1$ \TO $k$}{
\STATE $\hat{\x}_i \gets \LB(\{\mathbf{t}^j_i \mid j \in [n]\}, \bar{\x}_i)$\;
}
\ENDFOR
\IF {$\bar{\x} \neq \hat{\x}$}{
\STATE $(\mbox{stationary},\bar{\x}) \gets (\mbox{false},\hat{\x})$\;
}
\ENDIF
}
\ENDFOR
}
\ENDWHILE
\RETURN $[\bar{\x}]$
\end{algorithmic}
\end{algorithm}
\end{figure}
  
\subsection{Computing local barycenters}\label{sec:algorithm}
We will now describe an iterative algorithm aimed at finding local $p$-barycenters in $\SP$, inspired by the algorithm for persistence diagrams in \cite{turner2014frechet}. We state the algorithm in a very general context where convergence to a solution in finite time is not guaranteed, before giving some cases where a finite number of iterations produces a local $p$-barycenter.

\begin{Definition}\label{def:descent}
Let $X$ be a metric space and $F(X)$ the set of all finite subsets of $X$. By a \emph{$p$-descent operator} we mean a function $\LB: F(X) \times X \to X $ such that for any $S \in F(X)$ and $x \in X$ one of the following is true:
\begin{enumerate}
    \item[(a)] either $\sum_{s \in S} d(\LB(S,x), s)^p < \sum_{s \in S} d(x, s)^p$, or
    \item[(b)] $\LB(S,x) = x$ and $x$ is a local $p$-barycenter of $S$,
\end{enumerate}
\end{Definition}

In particular, if $X$ admits unique $p$-barycenters and $\LB(S,x)$ is defined to be the $p$-barycenter of $S$ for any $x$, then $\LB$ is a $p$-descent operator. More generally, a $p$-descent operator might be one or more steps of a gradient descent method. Our method for finding local $p$-barycenters in $\SP$ given a subset $S$, a $p$-descent operator $\LB$ on $X$ and an initial seed $\x_0 \in X^k$ is described in Algorithm \ref{algorithm}.

\begin{Remark}\label{cuturialgorithm2}
Algorithm 2 in \cite{cuturi2014fast} describes a method for finding approximate $2$-barycenters in $P_k(\mathbb{R}^d, \Theta)$, the subset of $W_2(\mathbb{R}^d)$ consisting of those measures with support of size at most $k$ and weights in some chosen set of $k$-tuples $\Theta$. When $\Theta$ contains only the $k$-tuple $(1/k,\ldots,1/k)$, Algorithm 2 in \cite{cuturi2014fast} (with the parameter choice $\theta=1$) is equivalent to Algorithm 1 in this section with $X = \mathbb{R}^d$ and $\LB(S,x) = \frac{1}{|S|}\sum_{\s \in S} \s$. 
\end{Remark}

If the $\mathsf{while}$ loop terminates in Algorithm \ref{algorithm}, then we have the exact conditions for $\bar{\x}$ to be stationary with respect to $S$ (Definition \ref{stationary}), so by Theorem \ref{localbarycenters}, we have found a local $p$-barycenter. As mentioned above, however, whether or not the algorithm terminates might depend on the definition of $\LB$. In all cases, we can observe the quantity
\[
D = \sum_j d_{\ell^p}(\bar{\x}, \mathbf{t}^j)^p  = \sum_{i,j} d_X(\bar{\x}_i, \mathbf{t}^j_i)^p,
\]
where $\mathbf{t}^j \in R_j$ for each $j$ (note that $D$ does not depend on the choice of $\mathbf{t}^j$), and note that $D$ is strictly reduced during all but the last iteration of the $\mathsf{while}$ loop. Indeed, redefining $R_j$ in Line 6 cannot increase $D$, while the definition of the local $p$-barycenter operator $\LB$ guarantees that when Line 14 is executed, one of the sums $\sum_i d_X(\bar{\x}_i, \mathbf{t}^j_i)^p$ is reduced. 

\begin{Proposition}
Suppose Algorithm \ref{algorithm} always terminates in a finite number of steps for some $X$ and $\LB$. For $S \subseteq X$ and $x \in X$, define $\mathcal{A}(S,x)$ to be the output of Algorithm \ref{algorithm} applied to the subset $S$ with initial value $\x_0 = x$. Then $\mathcal{A}$ defines a $p$-descent operator.
\end{Proposition}
\begin{proof}
Note that $D = \sum_{\s^i \in S} W_p(\bar{\x}_0, \s^i)^p$ at Line 8 on the first iteration of the \texttt{while} loop, after which $\sum_{\s^i \in S} W_p(\bar{\x}_0, \s^i)^p \leq D$. Since $D$ is strictly reduced during all but the last iteration of the \texttt{while} loop, $\mathcal{A}$ satisfies Definition \ref{def:descent}(a). The only way the \texttt{while} loop terminates after a single iteration is if the output $[\bar{\x}]$ equals the initial value $\x_0$, which shows condition (b).
\end{proof}

Since $D$ is strictly reduced at each iteration, if we can show that there are only finitely many possible values $D$ can take, then we are done. If $\LB(S,\cdot)$ admits finitely many values for each $S$, then $D$ can take only finitely many values. Indeed, in this case the sets $R_j$ constrain $\bar{x}$, and hence $D$, to finitely many values. Moreover, there are only finitely many possibilities for the sets $R_j$ since each $R_j$ is a subset of $[\s^j]$. Thus we get the following cases where Algorithm \ref{algorithm} is guaranteed to terminate.

\begin{Proposition}\label{termination}
Let $X$ be one of the following:
\begin{itemize}
    \item[(a)] a space such that every finite subset has a non-empty, finite set of local $p$-barycenters,
    \item[(b)] $\mathsf{SP}^k Y$ where $Y$ is a space satisfying (a) above.
\end{itemize}
Suppose $\LB$ is chosen so that $\LB(S,x)$ is always a local $p$-barycenter of $S$. Then Algorithm \ref{algorithm} terminates after finitely many steps.
\end{Proposition}
\begin{proof}
Given the above discussion, we have only to show that $\LB(S,\cdot)$ admits finitely many values for fixed $S$. For spaces of type (a) this is obvious. For spaces of type (b) we note that by Theorem \ref{localbarycenters}, local barycenters in $\mathsf{SP}^k Y$ are completely determined by the sets $S^\phi_i$ in Definition \ref{stationary} and a choice of local $p$-barycenter for each $S^\phi_i$. Since there are only finitely many possibilities for the $S^\phi_i$, we get the required result.
\end{proof}

Examples of spaces satisfying condition (a) in Proposition \ref{termination} for $p=2$ include $\mathbb{R}^n$ and more generally any Hadamard manifold~\cite{cartan1928geometrie}. An example of a space which does not satisfy (a) is the $2$-sphere: any point along the equator is a $2$-barycenter of the set consisting of the north and south poles. In the applications to redistricting in Section \ref{sec:redistrict} we will work with $\SP$ where $X = \mathsf{SP}^M \mathbb{R}^2$, and the $p$-descent operator will be Algorithm \ref{algorithm}, but applied to $X$; Corollary \ref{termination} thus applies doubly to this setting (see Section \ref{sec:redistrict} for details).

\begin{Observation}\label{W2global}
A local $p$-barycenter in any Wasserstein space $W_p(X)$ is necessarily a global $p$-barycenter. Indeed, if $\mu$ is a local $p$-barycenter of $\mu_1, \ldots, \mu_k$ and there exists a $\mu'$ with $
\sum_{i} W_p(\mu_i,\mu)^p > \sum_{i} W_p(\mu_i,\mu')^p
$ then for any small $\varepsilon > 0$ let $\mu_\varepsilon$ be the convex combination $\mu_\varepsilon = \varepsilon \mu' + (1-\varepsilon) \mu$. Taking appropriate convex combinations of optimal plans between $\mu_i$ and $\mu$ and $\mu'$ respectively, one can show $\sum W_p(\mu_i,\mu_\varepsilon)^p < \sum W_p(\mu_i,\mu)^p$.
This is a contradiction since $\mu$ is a local $p$-barycenter and the distance between $\mu$ and $\mu_\varepsilon$ can be made arbitrarily small in $W_p$ distance. This all remains true if we restrict ourselves to the subset of $W_p(X)$ consisting of measures with finite support. 
\end{Observation}

\begin{Remark}\label{rem:subspace}
One can also consider local $p$-barycenters of a set $S$ \emph{on a subset $\Omega \subseteq X$}, that is, local minimizers in the subspace $\Omega$ of the functional $f_{S,p}$ from Equation \ref{frechet}. For example, in \cite{cuturi2014fast} the authors consider minimizing $f_{S,p}$ on the subset of $W_2(X)$ given by discrete measures with support of size at most $m$. Theorem \ref{localbarycenters} holds if we replace ``local $p$-barycenter'' by ``local $p$-barycenter on $\Omega$'' in the definition of stationary (Definition \ref{stationary}) and replace ``local $p$-barycenter in $\SP$'' with ``local $p$-barycenter in $\SP$ on $\mathsf{SP}^k \Omega$''. The proof is the same. If we replace ``local $p$-barycenter'' in the definition of $p$-descent operator (Definition \ref{def:descent}) by ``local $p$-barycenter on $\Omega$'', then Algorithm \ref{algorithm} becomes an algorithm for finding local $p$-barycenters on $\mathsf{SP}^k \Omega$. This more general setup will become relevant when we apply Algorithm \ref{algorithm} to unbalanced partitions of point clouds in Section \ref{sec:clustering_algorithms}.
\end{Remark}

\subsection{Indexing by the barycenter}\label{sec:matching}

Let $S$ be a set of points in $X^k$, and let $[\x]$ be a (local) barycenter of the set $[S] \subseteq \SP$. A choice of representative $\hat{\x} \in [\x]$ gives rise naturally to a way of choosing a representative $\hat{\s} \in [\s]$ for each $[\s] \in S$. Indeed, for $[\s] \in S$, choose an optimal matching $\pi$ from $\hat{\x}$ to $\s$, and define $\hat{\s}_i = \s_{\pi(i)}$. In other words, label the local barycenter first, and then reorder the points in $S$ so that they are each labeled by a best matching to the local barycenter. In many applications, this reordering of the elements of $S$ is at least as important as the local barycenter itself. In general, however, there may be multiple choices of optimal matching $\pi$, so an additional condition is required for the choice of representative $\hat{\s}_i$ to be unique. In this section we outline some cases for which $\hat{\s}_i$ is defined, focussing on $2$-barycenters and the $2$-Wasserstein distance.

\begin{Definition} \label{onepointchange}
Let $X$ be a metric space. We say that $X$ has the \emph{$p$-barycenter one-point-change ($p$-OPC) property} if for any finite set of points $x_1,\ldots, x_k$ with local $p$-barycenter $\hat{x}$, the point $\hat{x}$ is not a local $p$-barycenter of $x_1',\ x_2\, \ldots, x_k$ if $x_1 \neq x_1'$.
\end{Definition}

In other words, the OPC property means that local barycenters change whenever exactly one point changes. For the case $p = 2$ and $X = \mathbb{R}^n$ with the Euclidean metric, (local) barycenters are given by the coordinate means, and so it is clear that Euclidean spaces have the $2$-OPC property. Note the OPC property does not require that the local barycenter change if another point is removed or added. Indeed, in $\mathbb{R}^n$ the barycenter of a subset $S$ is the same as the barycenter of the subset $S \cup \{\bar{s}\}$ where $\bar{s}$ is the barycenter of $S$. 

\begin{Proposition}
If $X$ has the $p$-OPC property, then so does $\SP$ with the $W_p$ metric.
\end{Proposition}
\begin{proof}
If the subset $[S] \subseteq \SP$ has local $p$-barycenter $[\x]$ then by Theorem \ref{stationary}, for every $1 \leq i \leq k$, $\x_i$ is a local $p$-barycenter of the set $
S^{\Phi}_{i} = \{ \s_{\phi_{\s}(i)} \mid \s \in S \}$, where $\Phi = (\phi_{\s})_{\s \in S}$ is a set of optimal matchings from $\x$ to each $\s \in S$. If one of the elements of $[S]$ changes then at least one of the $S^{\Phi}_{i}$ changes by exactly one element. Thus since $X$ has the $p$-OPC property, $[\x]$ is no longer a local $p$-barycenter.
\end{proof}

\begin{Proposition}
Let $M$ be a connected, compact Riemannian manifold of dimension $n$. Then $M$ has the $2$-OPC property.
\end{Proposition}
\begin{proof}
Consider a finite set of points $x_1,\ldots, x_k$ with local $2$-barycenter (in $X$) $\bar{x}$. For each $i$, let $\gamma_i: [0,1] \to M$ be a minimal geodesic from $\bar{x}$ to $x_i$. We claim that for any $\varepsilon \in [0,1]$, $\bar{x}$ is a $2$-barycenter of $x'_1,\ldots,x'_k$ where $x_i' = \gamma_i(\varepsilon)$. Suppose that $\sum d(y,x_i')^2 < \sum d(\bar{x},x_i')^2$ for some $y \in M$. Then
\begin{align*}
\sum_i d(y, x_i)^2 & \leq  \sum_i \left( d(y, x_i') + d(x_i', x_i) \right)^2  = \sum_i \left(d(y, x_i')^2 +  d(x_i', x_i)^2 + 2 d(y, x_i')d(x_i', x_i)\right).
\end{align*}
Applying the Cauchy-Schwarz Inequality and the facts that $d(\bar{x}, x_i') = \varepsilon d(\bar{x}, x_i)$ and $d(x_i', x_i) = (1-\varepsilon) d(\bar{x}, x_i)$, we see that the latter quantity is upper bounded by
\begin{align*}
&\sum_i  d(y, x_i')^2 + \sum_i d(x_i', x_i)^2 + 2 \left[ \sum_i d(y, x_i')^2 \cdot \sum_i d(x_i', x_i)^2 \right]^{1/2} \\
& \qquad < \sum_i d(\bar{x}, x_i')^2 + \sum_i d(x_i', x_i)^2 + 2 \left[ \sum_i d(\bar{x}, x_i')^2 \cdot \sum_i d(x_i', x_i)^2\right]^{1/2} \\
& \qquad = \sum_i d(\bar{x}, x_i')^2 + \sum_i d(x_i', x_i)^2 + 2 \varepsilon (1-\varepsilon) \sum_i d(\bar{x}, x_i)^2 = \sum_i d(\bar{x}, x_i)^2.
\end{align*}
Since $\bar{x}$ is a local $2$-barycenter of the $x_i$, this implies that $y$ cannot be too near $\bar{x}$, so we have that $\bar{x}$ is indeed a local $2$-barycenter of the $x_i'$ as well. Thus, by replacing $x_i$ with $x_i'$ for sufficiently small $\varepsilon$, we can assume from here on that the $x_i$ are contained in a neighborhood of $\bar{x}$ for which the exponential map $\exp_{\bar{x}}: \mathbb{R}^n \to M$ is a diffeomorphism. In particular, $\log_{\bar{x}} x_i$ is defined for each $i$. We now recall that in this situation the gradient of the function $p \mapsto d(x_i, p)^2$ at $\bar{x}$ is given by $-2 \log_{\bar{x}} x_i$. Thus if $\bar{x}$ is a local $2$-barycenter, we have $\sum_i \log_{\bar{x}} x_i = 0$, so that $x_1$ is completely determined by $\bar{x}$ and the other $x_i$, as required.
\end{proof}

When $X$ has the OPC property, optimal representative choices are always unique:

\begin{Proposition}
Let $X$ be a space with the $p$-OPC property. Let $S$ be a set of points in $X^k$, and let $[\mathbf{x}]$ be a (local) barycenter of the set $[S] \subseteq \SP$. Then for every $\mathbf{s} \in S$, there is a unique representative $\hat{\s} \in [\s]$ minimizing
\begin{equation}\label{aligning}
\sum_{i=1}^k d(\hat{\textbf{x}}_i, \hat{\s}_{i})^p,
\end{equation}
\end{Proposition}
\begin{proof}
For each $\s$, choose a representative $\hat{\s}$ as above and suppose that $\tilde{\s}$ was another possible choice of representative for $\s$ with $\tilde{\s}_m \neq \hat{\s}_m$. By Theorem \ref{localbarycenters}, $\x_m$ is a local $p$-barycenter of both $Y := \{ \hat{\s}_m \mid \s \in S \}$ and $(Y \setminus \hat{\s}_m) \cup \tilde{\s}_m$. But this is a contradiction since $X$ has the $p$-OPC property.
\end{proof}

We now prove that certain $2$-Wasserstein spaces have the $2$-OPC property. For the absolutely continuous case (Proposition \ref{ACcase}), we will need the following two results.

\begin{Theorem}[\cite{brenier1987decomposition} for $\mathbb{R}^n$, \cite{mccann2001polar} for $M$]\label{mccann}
Let $M$ be a connected, compact Riemannian manifold or $\mathbb{R}^n$. Let $\mu_1,\mu_2 \in W_2(M)$ be two probability measures which are absolutely continuous with respect to volume. Then the Wasserstein distance $W_2(\mu_1,\mu_2)$ is realized by a unique measure $\gamma$. Moreover, this $\gamma$ is concentrated on the graph of a measureable mapping $T$ over $\mu_1$.
\end{Theorem}

\begin{Theorem}[\cite{agueh2011barycenters, kim2017wasserstein}]\label{kim}
Let $M$ be a connected, compact Riemannian manifold or $\mathbb{R}^n$ and let $W_2(M)$ be the $2$-Wasserstein space of probability measures on $M$. Let $\mu_1,\ldots, \mu_k$ in $W_2(M)$ be a finite set of elements of $W_2(M)$ which are all absolutely continuous with respect to volume. Then:
\begin{itemize}
    \item the $\mu_1,\ldots, \mu_k$ admit a unique barycenter $\bar{\mu}$ which is absolutely continuous with respect to volume, and
    \item if we denote by $T_k$ the optimal map from $\bar{\mu}$ to $\mu_k$ as guaranteed by Theorem \ref{mccann} then for $\bar{\mu}$-almost every $z$, $z$ is the unique barycenter of the points $T_1(z), T_2(z),\ldots,T_k(z)$.
\end{itemize}
\end{Theorem}
\begin{proof}
For the case where $M$ is a connected, compact Riemannian manifold, the statements are special cases of Theorem 5.1 and Lemma 4.3 in \cite{kim2017wasserstein} respectively. For the case $M = \mathbb{R}^n$, the first statement was first proven by Agueh and Carlier in \cite{agueh2011barycenters}, and the proof of Lemma 4.3 in \cite{kim2017wasserstein} works for the second part without modifications.
\end{proof}

\begin{Proposition}\label{ACcase}
Let $M$ be a connected, compact Riemannian manifold or $\mathbb{R}^n$ and let $W_2(M)$ be the $2$-Wasserstein space of probability measures on $M$. Let $\mu_1, \ldots, \mu_k$ in $W_2(M)$ be a finite set of elements of $W_2(M)$ which are absolutely continuous with respect to volume. Then if $\bar{\mu}$ is a (local) $2$-barycenter of $\{\mu_1,\mu_2, \ldots, \mu_k\}$  in $W_2(M)$, then $\bar{\mu}$ is not a (local) $2$-barycenter of $\mu_1',\mu_2\, \ldots, \mu_k$ when $\mu_1 \neq \mu_1'$.
\end{Proposition}
\begin{proof}
Assume for contradiction that $\bar{\mu}$ is a $2$-barycenter of $\mu_1',\mu_2\, \ldots, \mu_k$ when $\mu_1 \neq \mu_1'$. By Theorem \ref{kim}, for $\bar{\mu}$-almost every $z$, $z$ is the unique barycenter of the points $T_1(z)$,  $T_2(z)$,$\ldots$,$T_k(z)$, where $T_k$ is defined as in Theorem \ref{kim}. In addition, for $\bar{\mu}$-    almost every $z$, $z$ is the unique barycenter of the points $T'_1(z), T_2(z),\ldots,T_k(z)$, where $T'_1$ is the optimal map from $\bar{\mu}$ to $\mu_1'$. Since $M$ has the $p$-OPC property, this means that $T'_1(z) = T_1(z)$ for $\bar{\mu}$-almost every $z$. It follows that for any measureable subset $A$, $\mu_1(A) = \bar{\mu}(T_1^{-1}(A)) = \bar{\mu}(T_1^{\prime -1}(A)) = \mu'_1(A)$, a contradiction since $\mu_1 \neq \mu_1'$.
\end{proof}

We now consider the situation where $\mu_1,\ldots,\mu_k$ are finitely supported distributions over $\mathbb{R}^n$, since these are often used in applications to data, or to approximate continuous distributions. In this situation Anderes, Borgwadt and Miller~\cite{anderes2016discrete} showed the existence of a finite subset $S$ such that any $2$-barycenter of  $\mu_1,\ldots,\mu_k$ has support contained in $S$. Let $\mathcal{S} = \{s_0, s_1, \ldots, s_K\}$ denote the union of the above set $S$ and all the supports of $\mu_1,\ldots,\mu_k$, and let $\bar{\mu}$ be any $2$-barycenter of $\mu_1,\ldots,\mu_k$. For each $i$, choose an optimal coupling between each $\mu$ and $\mu_i$, encoded as a list of values $y_{i}^{j\ell}$ where $y_{i}^{j\ell}$ denotes the coupling between point $s_j$ in $\mu_i$ and point $s_\ell$ in $\mu$. In particular, we have the following identities
\[
\forall_i \forall_\ell \sum_j y_{i}^{j\ell} = \mu(\{s_\ell\}), \quad \forall_i \forall_j \sum_\ell y_{i}^{j\ell} = \mu_i(\{s_j\}).
\]

\begin{Lemma}[Lemma 1 in \cite{anderes2016discrete}]\label{anderes}
With the notation above, for any $s_j \in \supp(\bar{\mu})$, there is for each $i$ a unique point $s_{j_i} \in \mathcal{S}$ for which $y_{i}^{ j_i j} > 0$. Moreover, we have $s_j =\sum_i s_{j_i} /k $, so that in particular $s_j$ is the $2$-barycenter of the $s_{j_i}$.
\end{Lemma}

Intuitively this means that the discrete situation in $\mathbb{R}^n$ is very similar to the absolutely continuous case originally considered in \cite{agueh2011barycenters}: the support of the Wasserstein barycenter consists of the (metric) barycenters of points it is coupled to. It is therefore unsuprising that the $p$-OPC property holds in this setting.

\begin{Corollary}
Let $W_2(\mathbb{R}^n)$ be the $2$-Wasserstein space of probability measures on $\mathbb{R}^n$. Let $\mu_1, \ldots, \mu_k$ in $W_2(M)$ be a finite set of discrete distributions with finite support. Then if $\bar{\mu}$ is a (local) $2$-barycenter of $\{\mu_1,\mu_2, \ldots, \mu_k\}$  in $W_2(\mathbb{R}^n)$, then $\bar{\mu}$ is not a (local) $2$-barycenter of $\mu_1',\mu_2\, \ldots, \mu_k$ when $\mu_1 \neq \mu_1'$.
\end{Corollary}
\begin{proof}
Define the couplings $y_{i}^{j\ell}$ as above, and define $y_{1}^{\prime j\ell}$ to be an optimal coupling between $\mu$ and $\mu_1'$. Applying Lemma \ref{anderes} to the set of couplings $y_{i}^{j\ell}$ as well as this set with $y_{1}^{j\ell}$ replaced by $y_{1}^{\prime j\ell}$, we see that since $\mathbb{R}^n$ has the $2$-OPC property, we have $y_{1}^{\prime j\ell} > 0 \Leftrightarrow y_{1}^{j\ell} > 0$. This, together with $\sum_j y_{1}^{\prime j\ell} = \mu(\{s_\ell\})$ gives us $y_{1}^{\prime j\ell} = y_{1}^{j\ell}$ for all $j, \ell$, hence $\mu_1 = \mu'_1$. 
\end{proof}

Summarizing, we get the following sufficient conditions for unique optimal representatives. 

\begin{Corollary}\label{onepointcases}
Suppose $X$ is one of the following spaces:
\begin{itemize}
    \item a connected, compact manifold $M$;
    \item the subset of $W_2(M)$ consisting of absolutely continuous measures with respect to volume, where $M$ is a connected, compact manifold;
    \item the subset of $W_2(\mathbb{R}^n)$ consisting of finitely supported measures;
    \item $\mathsf{SP}^k Y$ for $Y$ either a compact connected manifold or $\mathbb{R}^n$.
\end{itemize}
Let $S$ be a set of points in $X^k$, and let $[\mathbf{x}]$ be a (local) barycenter of the set $[S] \subseteq \SP$ where $\SP$ is endowed with the $2$-Wasserstein metric. Then for every $\mathbf{s} \in S$, there is a unique representative $\hat{\s} \in [\s]$ minimizing $\sum_{i=1}^k d(\hat{\textbf{x}}_i, \hat{\s}_{i})^2$.
\end{Corollary}

The next example gives a class of spaces which do not have the $2$-OPC property.

\begin{Example}
Let $X$ be any space that admits a branching geodesic $\phi:[a,b] \to X$. Without loss of generality, suppose $\phi(b) \neq \bar{\phi}(b)$ and that the geodesic branches at $t_0 = (a+b)/2$. It is easy to check that $\phi(t_0)$ is the $2$-barycenter of the set $\{\phi(a), \phi(b)\}$ as well as of the set  $\{\bar{\phi}(a), \bar{\phi}(b)\}$. Since $\phi(a) = \bar{\phi}(a)$, $X$ does not have the $2$-OPC property. For a concrete example of a space $X$ with the branching property, consider the space $X$ obtained by gluing three copies of the ray $[0,\infty)$ at zero with $d$ the shortest path metric.
\end{Example}

\section{Examples}\label{sec:examples}

To illustrate the flexibility and intuitive outputs of Algorithm \ref{algorithm}, we provide in this section some experimental results on simple datasets.

\subsection{A simple non-Euclidean example}\label{sec:circle}
In this section we demonstrate Algorithm \ref{algorithm} in a case where the descent operator does not always pick out a local barycenter, but where the algorithm nonetheless converges to a local barycenter in finite time. In this example $X$ will also not be a Euclidean space, but rather the circle $S^1$ with the usual geodesic metric. We will use Algorithm \ref{algorithm} to compute the local $2$-barycenter of a set of points in $\mathsf{SP}^2 S^1$. For two points $x,y$, we denote by $\angle xy$ the signed clockwise angle from $x$ to $y$ measured in $[-\pi, \pi]$, addressing the choice between $-\pi$ and $\pi$ as a special case when necessary. The distance between two points is therefore equal to the absolute value of the angle between them.

For $S \subseteq S^1$ and $x \in S^1$, we define $\LB(S,x)$ by the formula $\angle x \LB(S,x) = \frac{1}{|S|} \sum_{s \in S} \angle xs
$. If $S$ contains antipodes of $x$, we define $\angle xs$ for all of the antipodes to be whichever of $\{-\pi, \pi\}$ makes $
D' = \sum_{s \in S} (\angle xs -  \frac{1}{|S|} \sum_{s \in S} \angle xs)^2$ smaller. If $D'$ is the same for $\pm \pi$, we choose $\pi$.

\begin{Proposition}\label{LBcircle}
With $\LB$ defined as above, $\LB$ is a $2$-descent operator. Moreover, for fixed $S$, $B(S,-)$ takes only finitely many values.
\end{Proposition}
\begin{proof}
It is easy to check that if $\LB(S,x) \neq x$ then $\sum_{s \in S} d(s, \LB(S,x))^2 < \sum_{s \in S} d(s, x)^2$. Suppose then that $\LB(S,x) = x$. We first observe that none of the points in $S$ are diametrically opposite $x$. Indeed, if some were, then $\frac{1}{|S|} \sum_{s \in S} \angle xs$ must change based on the choice between $-\pi$ and $\pi$. Denote one set of chosen angles by $\angle_0 sx$ and the other by $\angle_1 sx$, and suppose without loss of generality that $\frac{1}{|S|} \sum_{s \in S} \angle_0 xs = 0$. Then we get
\[
D' = \sum_{s \in S} (\angle_1 xs -  \frac{1}{|S|} \sum_{s \in S} \angle_1 xs)^2 < \sum_{s \in S} ( \angle_1 xs )^2 = \sum_{s \in S} ( \angle_0 xs )^2
\]
so that $\angle_1 sx$ is the correct choice, contradicting our assumption that $\LB(S,x) = x$. Since $S$ does not contain an antipode of $x$, it is easy to check that $x$ is a local barycenter of $S$. To show the second part of the statement, observe that $\LB(S,x) = \LB(S,x')$ if there are no points of $S$ on one of the arcs between the antipodes of $x$ and $x'$. It follows that $\LB(S,x)$ is constant on at most $|S|$ arcs whose complement is at most $|S|$ points, so $B(S,x)$ has only finite many possible values for a fixed subset $S$.
\end{proof}

We therefore conclude based on the discussion in Section \ref{sec:algorithm} that Algorithm \ref{algorithm} terminates after a finite number of iterations when we define $\LB$ as above. To demonstrate this experimentally, we generate 10 random unordered pairs of points on the circle, shown as the leftmost plot in Figure \ref{fig:circle1}. We then run Algorithm \ref{algorithm} with two different seeds and plot the barycenter locations as well as the points matched to each component of the barycenter. These are shown in the rightmost two columns of Figure \ref{fig:circle1}. The two different seeds produce different local $2$-barycenters of the input set.

\begin{figure}
    \centering
    \resizebox{!}{6cm}{
    \begin{tikzpicture}
    \node at (-9,-4.1) {Unordered pairs in $S^1$};
    \node at (-9,-6) {\includegraphics[width=.2\textwidth]{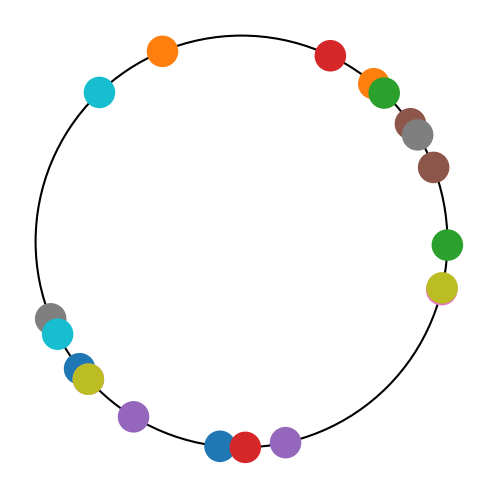}};
    
    \node at (-4, -2.1) {Barycenter};
    \node at (1.5, -2.1) {Matching};
    \node[rotate=90] at (-6.2, -4) {Dark blue seed};
    \node[rotate=90] at (-6.2, -8) {Orange seed};
    
    \node at (-4,-4) {\includegraphics[width=.2\textwidth]{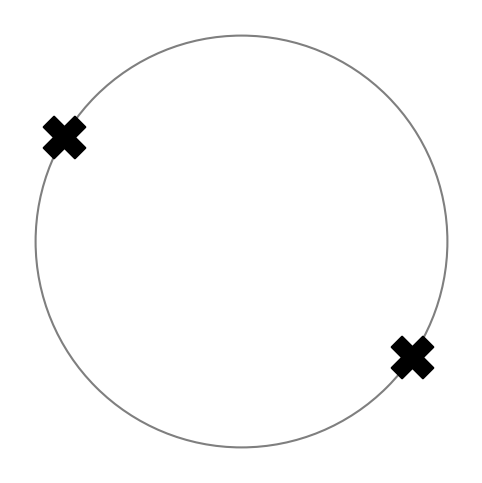}};
    \node at (-0,-4) {\includegraphics[width=.2\textwidth]{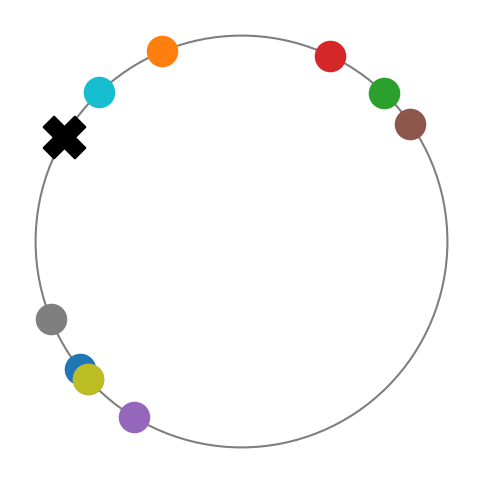}};
    \node at (3,-4) {\includegraphics[width=.2\textwidth]{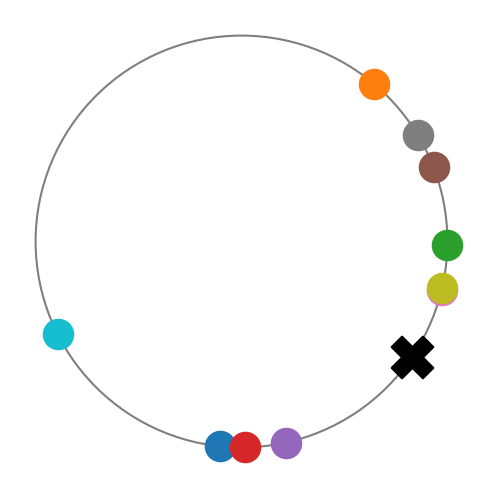}};
    \node at (-4,-8) {\includegraphics[width=.2\textwidth]{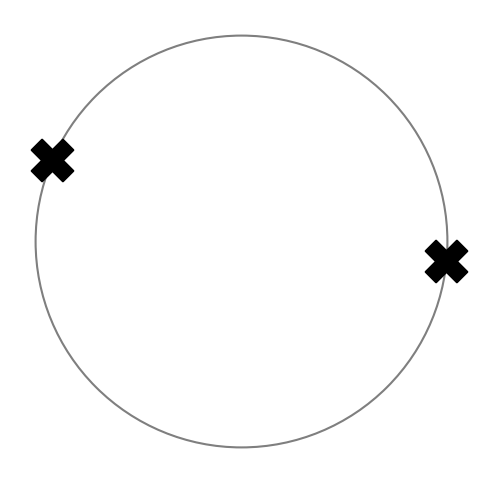}};
    \node at (-0,-8) {\includegraphics[width=.2\textwidth]{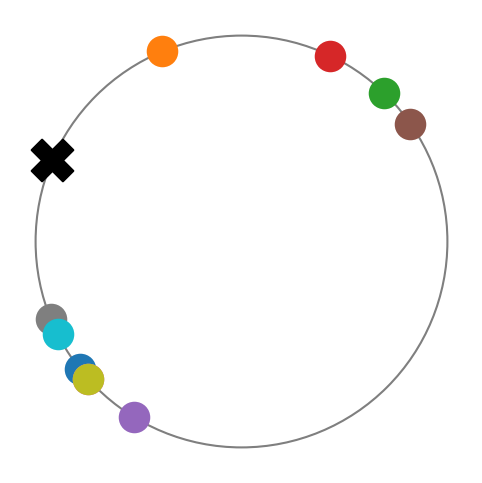}};
    \node at (3,-8) {\includegraphics[width=.2\textwidth]{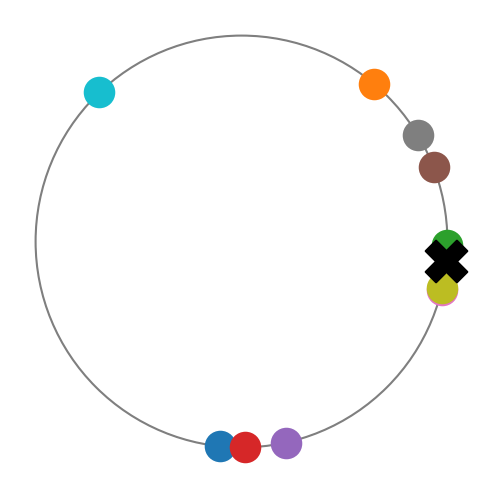}};
    \end{tikzpicture}
    }
    \caption{Demonstrating Algorithm \ref{algorithm} for a set of points in $\mathsf{SP}^2 S^1$ using the $2$-descent operator defined in Section \ref{sec:circle}. For two different initial seeds, the barycenter and the points matched to each point in the barycenter via a unique optimal matching are shown.}
    \label{fig:circle1}
\end{figure}

\subsection{Clustering algorithm consistency}\label{sec:clustering_algorithms}
We can use our method to compare the consistency of different clustering methods on point cloud data. To do this, we generate point clouds $P \subseteq \mathbb{R}^2$ consisting of 5,000 points each based on the test datasets provided by \texttt{scikit-learn} \cite{scikit-learn}, labeled A--E in Figure \ref{fig:clustering}. For each point cloud, we randomly subsample 10 sets $\{P_1,\ldots,P_{10}\}$ of 500 points each and partition these subsets using a clustering method. If we identify a point cloud with its empirical distribution, we can view a partitioned point cloud as an element of $\mathsf{SP}^k W_2(\mathbb{R}^2)$, where $k$ is the number of clusters. We can therefore compute a local $2$-barycenter on $\mathsf{SP}^k \Omega \subseteq \mathsf{SP}^k W_2(\mathbb{R}^2)$ of the ten partitioned point clouds, where $\Omega$ is a suitable subspace of $\mathsf{SP}^k W_2(\mathbb{R}^2)$ (see Remark \ref{rem:subspace}). For our experiments, we chose $\Omega$ to be the set of atomic distributions on 100 atoms, which is isometric to $\mathsf{SP}^{100} \mathbb{R}^2$ by Proposition \ref{prop:geometric_facts}. The descent operator we used for this application is the restricted version of Algorithm 2 in \cite{cuturi2014fast} implemented in \cite{flamary2021pot}. For the initial seed $\x_0$, we under- and oversample the clusters in the first partition as needed to obtain $k$ sets of 100 points each.

For this experiment we test four different clustering methods from \texttt{scikit-learn}: $k$-means, spectral clustering, Ward agglomerative clustering and single linkage clustering. We set $k = 3$ for datasets A, D and E and $k=2$ for the remaining two. For each point cloud and clustering method, we compute a barycenter $B$ as described above and display them in the middle four columns of Figure \ref{fig:clustering}. We also compute the distance $W_2(B, P_i)$ from the barycenter to each partitioned point cloud and show the distribution of these distances as boxplots. 

The $W_2(B, P_i)$ distances shown in the boxplots can be read as a constistency measure in that they indicate how similar the partitioned point clouds were. Note that these values are not a measure of the quality of the clusterings, only their consistency across sub-datasets. Indeed, $k$-means was very consistent on dataset D despite the barycenter not resembling the intuitively ``correct'' clustering. Single linkage clustering was the least consistent on datasets A, C and E, and the most consistent on dataset B. The components of the single linkage barycenter for A and E resemble affine transformations of the full data set; this is because all of the partitions contained one large component and other very small clusters.

\begin{figure}
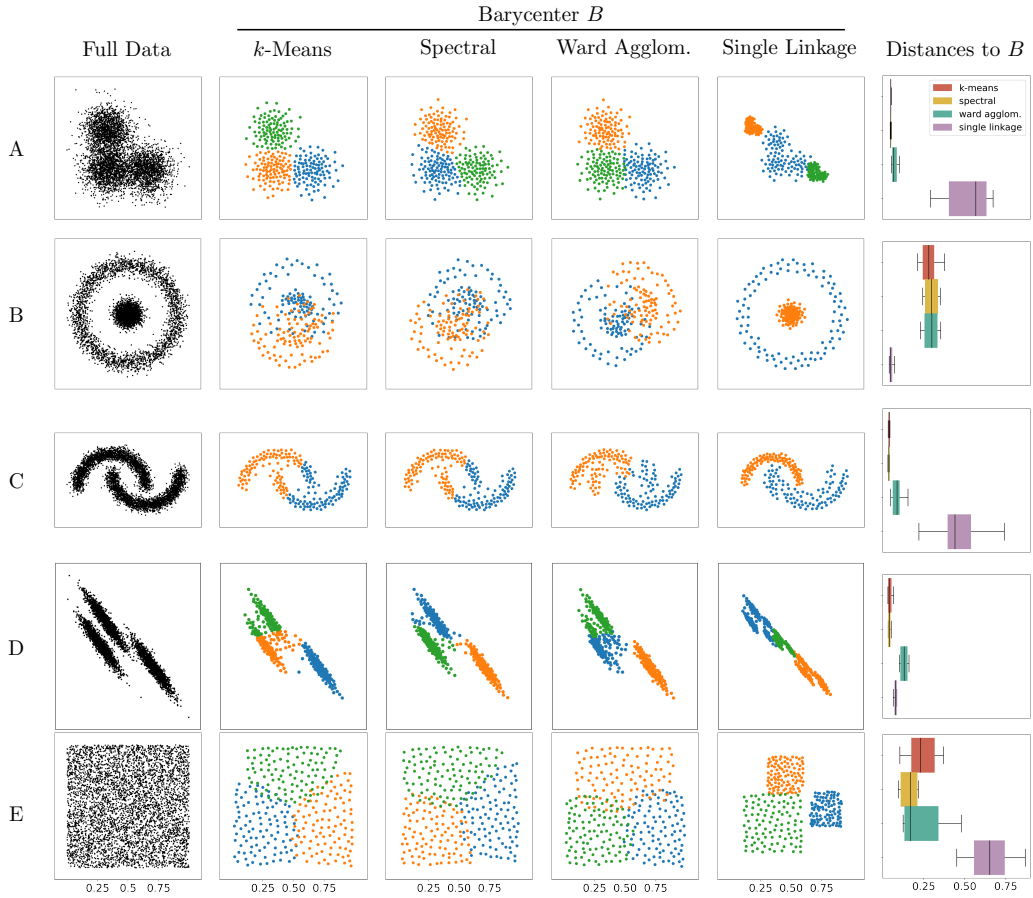

    \centering
    \resizebox{!}{12cm}{
    \begin{tikzpicture}
    \foreach \method [count=\i] in {$k$-Means, Spectral, Ward Agglom., Single Linkage}{
    \node at (3*\i, -1.2) {\method};
    }
    \node at (0, -1.2) {Full Data};
    \draw[thick] (2,-0.85) -- (13,-0.85);
    \node at (7.5, -0.6) {Barycenter $B$};
    \foreach \model [count=\j] in {gaussians, circles, moons, aniso, square} {
        \node at (0, -3*\j) {\includegraphics[width=0.18\textwidth]{clustering/sample_\model _False.png}};
    }
    \foreach \model [count=\j] in {A,B,C,D,E} {
        \node at (-2,-3*\j) {\model};
    }
    \foreach \method [count=\i] in {kMeans, spectral, agglomerative_ward, agglomerative_single}{
    \foreach \model [count=\j] in {gaussians, circles, moons, aniso, square} {
        \node at (3*\i, -3*\j) {\includegraphics[width=0.18\textwidth]{clustering/bary_\model _False_\method .png}};
    }
    }
    \node at (15, -1.2) {Distances to $B$};
    \foreach \model [count=\j] in {gaussians, circles, aniso, moons, square} {
        \node at (15, -3*\j) {\includegraphics[width=0.18\textwidth]{clustering/distortions_\model.png}};
    }
    \end{tikzpicture}
    }
    \caption{Comparing the consistency of clustering algorithms by subsampling. For each full data set (shown on the left), we generate 10 randomly subsampled datasets and cluster each of them into 2 or 3 clusters using a particular clustering method. We then compute the barycenter of the resulting partitioned point clouds (the middle plots). The distribution of distances from each partition to the barycenter are shown as boxplots (right).}
    \label{fig:clustering}
\end{figure}

\section{Applications to redistricting}\label{sec:redistrict}
In this section, we demonstrate applications of our method to political redistricting. Given a state $X$, which can be viewed as a polygonal planar region once an appropriate projection is chosen, we view a district $D \subseteq X$ as a probability distribution $P_D$ on $X$ whose support approximates that district. In practice, districts are built out of smaller territorial units such as Census blocks or voting precincts, so to discretize the problem, we assume that the measure of each territorial unit under $P_D$ is concentrated at its centroid. The distribution $P_D$ is therefore a sum of Dirac distributions with support all the centroids of territorial units assigned to $D$. In this paper, we will consider two different ways to define $P_D$: the \emph{area-weighted representation} in which the measure of a territorial unit $p$ under $P_D$ is proportional to its area, and a \emph{population-weighted representation} in which the measure of a precinct $p$ under $P_D$ is proportional to its population. 

In order to further reduce the computational cost of the algorithm, we approximate each $P_D$ by sampling $M$ points from it. We choose $M=40$ in this section based on the stability analysis in Section \ref{sec:stability}. We can therefore identify $P_D$ with an element of $\mathsf{SP}^M \mathbb{R}^2$, and a $k$-district redistricting plan with an element of $\mathcal{P} = \mathsf{SP}^k \mathsf{SP}^M \mathbb{R}^2$. To compute a local $2$-barycenter in $\mathcal{P}$, we use Algorithm \ref{algorithm} with a descent operator $\LB_1(S,x)$ given by applying Algorithm \ref{algorithm} to $S \subseteq \mathsf{SP}^M \mathbb{R}^2$ using seed $x$. The descent operator for this inner algorithm is given by $\LB_2(S,x) = \frac{1}{|S|} \sum_{\s \in S} \s$. Theorem \ref{localbarycenters} and Proposition \ref{termination} guarantee that Algorithm \ref{algorithm} returns a local $2$-barycenter in finite time and Corollary \ref{onepointcases} guarantees that there are unique labelings for each redistricting plan that arise from an optimal matching to that barycenter. Note that while enacted redistricting plans often come pre-equipped with identifiers such as ``1st Congressional District'', ``2nd Congressional district'', and so on, plans generated by a computer do not, hence the need for the geometry-aware labeling provided by our method. 

We demonstrate our method on Congressional redistricting in North Carolina, an area which has seen a great deal of litigation in the last decade, and which has already been extensively studied using ensemble methods in the literature \cite{mattingly1, amicus_math, duchinneedham21,amicus_geographers}. For some context on North Carolina's political geography, we show the 2016 Presidential two-way vote shares by precinct in Figure \ref{fig:labelledNC}. An examination of the stability of our method with respect to choice of initial seed and $M$ can be found in Section \ref{sec:stability}; in particular, we give evidence for a high degree of stability for both population- and area-weighted representations. For the sake of brevity, we will only show the population-weighted versions of plots in the main text and relegate the area-weighted version to Appendix (we will however mention any notable differences between the two). More details about the computational aspects of the study are available in Section \ref{sec:compdetails} as well as a link to our code base.

\begin{figure}
  \begin{minipage}[c]{0.5\textwidth}
    \includegraphics[width=\textwidth]{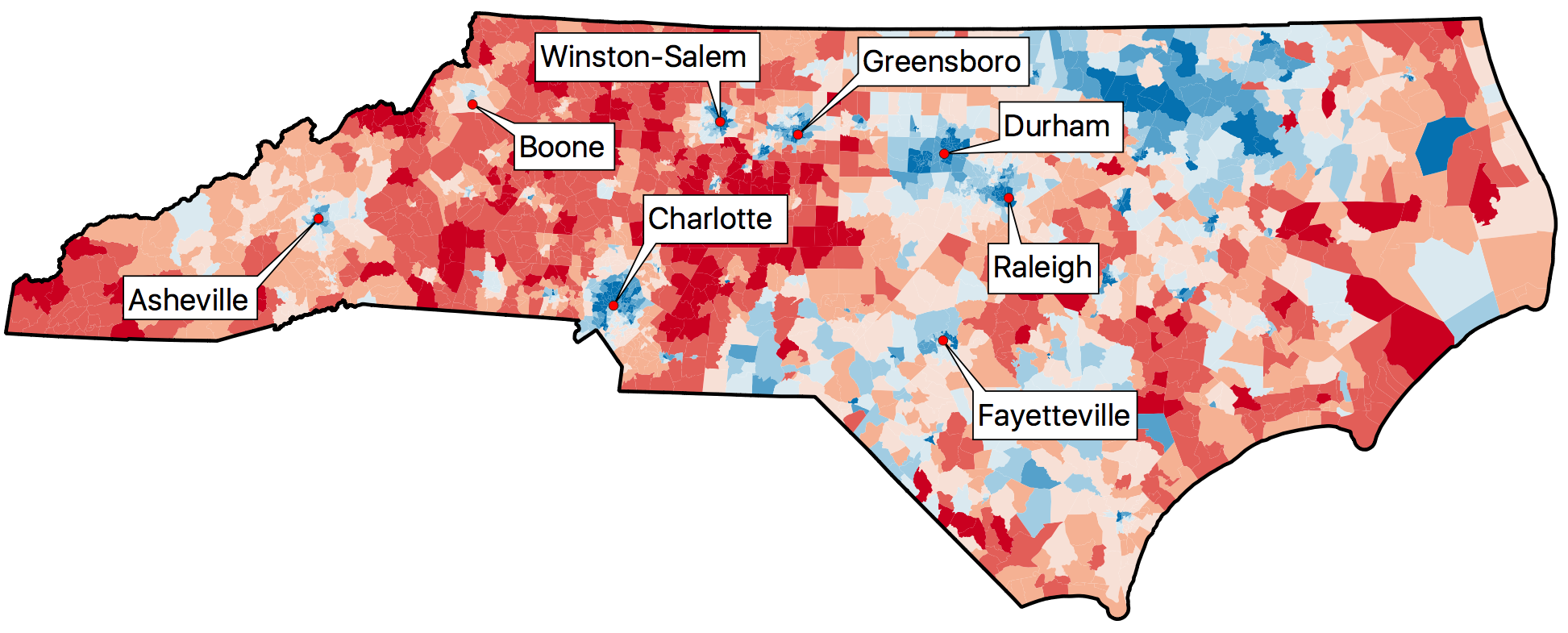}
  \end{minipage}\hfill
  \begin{minipage}[c]{0.48\textwidth}
    \caption{North Carolina's precincts showing two-way vote shares for the 2016 Presidential election. Red indicates more Republican and blue indicates more Democratic.}\label{fig:labelledNC}
  \end{minipage}
\end{figure}

\subsection{Visualizing ensembles}
We generate 1,000 Congressional plans for North Carolina using the ReCom Markov chain algorithm \cite{DeFord2021Recombination}, each with $k=13$ districts.\footnote{Note that in the next redistricting cycle, North Carolina will have $k=14$ Congressional districts, but we keep with the numbers for the 2010-2020 cycle since we want to compare plans enacted during that era to the ensemble.} We then represent each plan in $\mathsf{SP}^k \mathsf{SP}^M \mathbb{R}^2$ as described above and compute a local $2$-barycenter. Figure \ref{fig:neutral} shows the barycenter as an element of $\mathsf{SP}^k \mathsf{SP}^M \mathbb{R}^2$ for the population-weighted representation; each set of $M=40$ points of a given color gives one component of the barycenter which is then matched to one district in each plan in the ensemble. The accompanying heat maps show where the districts matched to each component lie. The heat maps each appear concentrated and have low overlap with one another, demonstrating that our method is able to successfully label districts based on their geography. Figure \ref{fig:neutralarea} shows the area-weighted analysis; the heat maps are almost identical, while the form of the barycenters is very different: the points in the population-weighted barycenter concentrate on population-dense areas such as cities, while the points in each component of the area-weighted barycenter are more evenly spread.

Figures \ref{fig:otherstates_pop} and \ref{fig:otherstates_area} show barycenters for ReCom ensembles on a selection of 19 states whose precinct data is available from \cite{mggg_states}. To determine the number of districts for each state, we use Congressional apportionments from the 2020 Census, which as of the writing of this paper had not gone into effect.\footnote{Population balance is still based on 2010 Census data as that is what was available from \cite{mggg_states} at time of writing} In particular, Figures \ref{fig:otherstates_pop} and \ref{fig:otherstates_area} show North Carolina with 14 districts instead of the 13 used in the analysis in this section.

\begin{figure}
    \centering
    \resizebox{32pc}{!}{
    \begin{tikzpicture}
    \begin{scope}
    \node at (-4, 2.2) {Location of matched districts};
    \node at (5, 2.2) {Barycenter};
     \node at (5, -2.3) {Overlaid heat maps};
    
    \node at (5,0) {\includegraphics[width=0.7\textwidth]{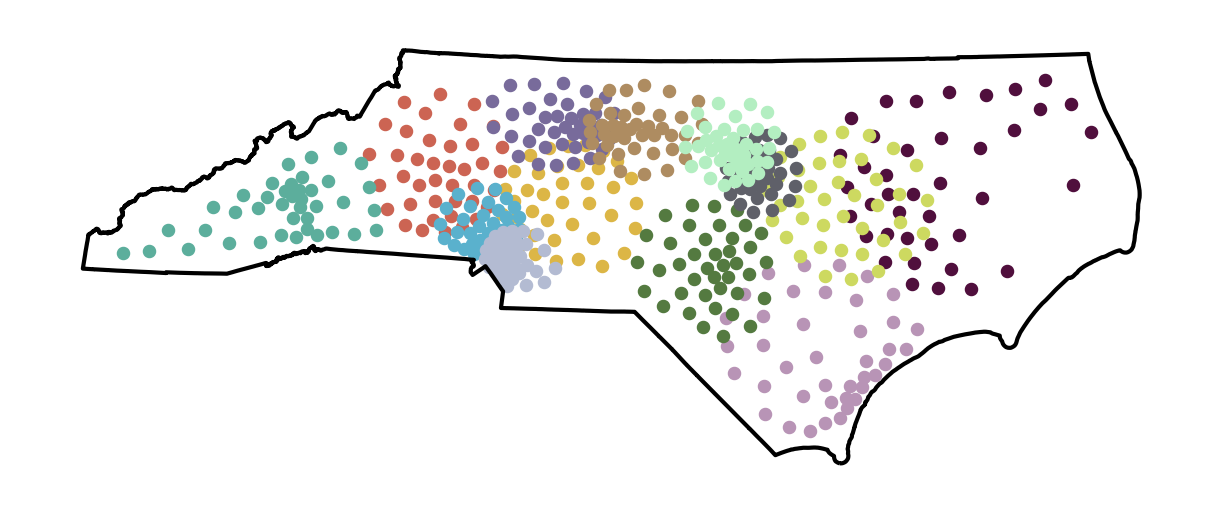}};
    \node at (5,-4.5) {\includegraphics[width=0.7\textwidth]{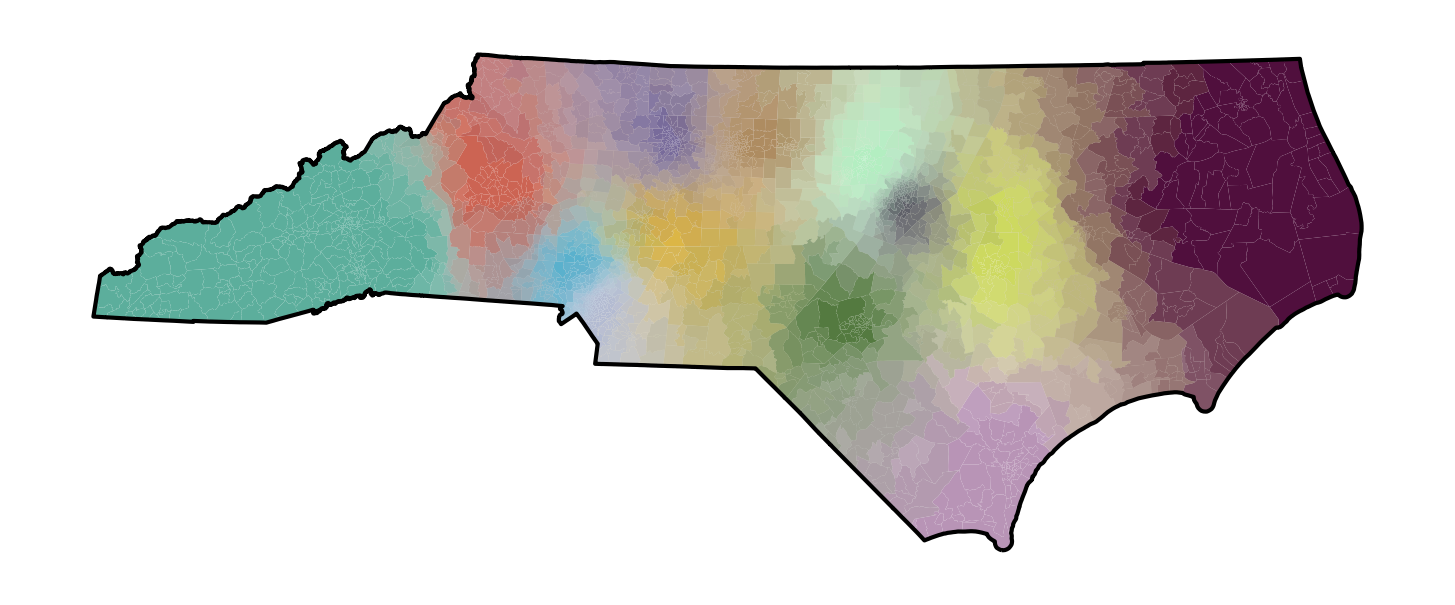}};
    
    \foreach \i in {1,3,5,7,9,11,13} {
        \node at (-7.5,2.25-0.75*\i) {$\mathsf{\i}$};
    }
    \foreach \i in {2,4,6,8,10,12} {
        \node at (-3.5,3-0.75*\i) {$\mathsf{\i}$};
    }
    
    \node at (-6,1) {\includegraphics[width=0.25\textwidth]{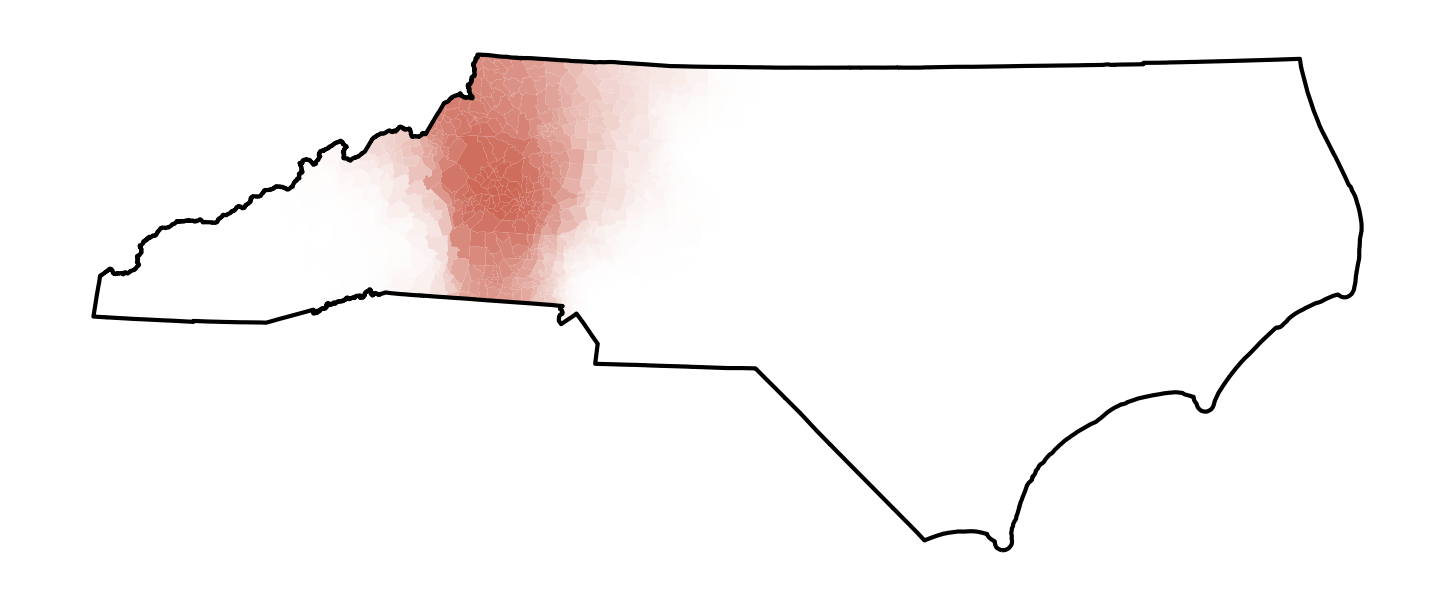}};
    \node at (-2,1) {\includegraphics[width=0.25\textwidth]{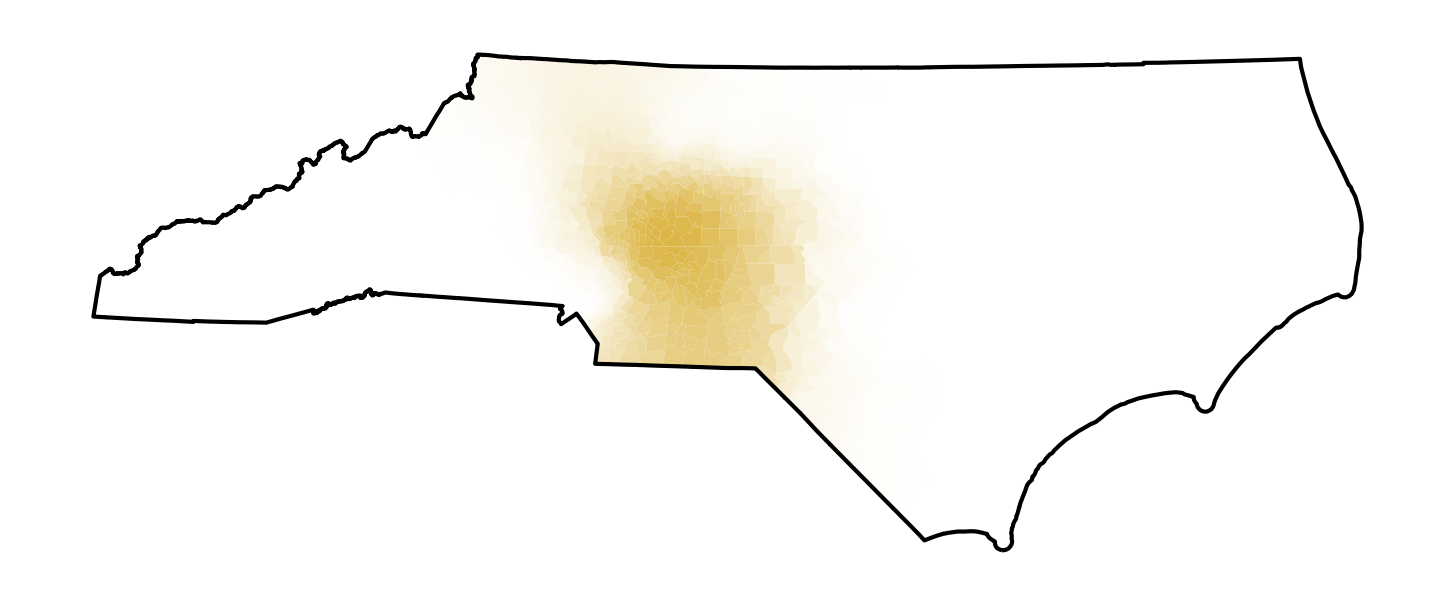}};
    \node at (-6,-0.5) {\includegraphics[width=0.25\textwidth]{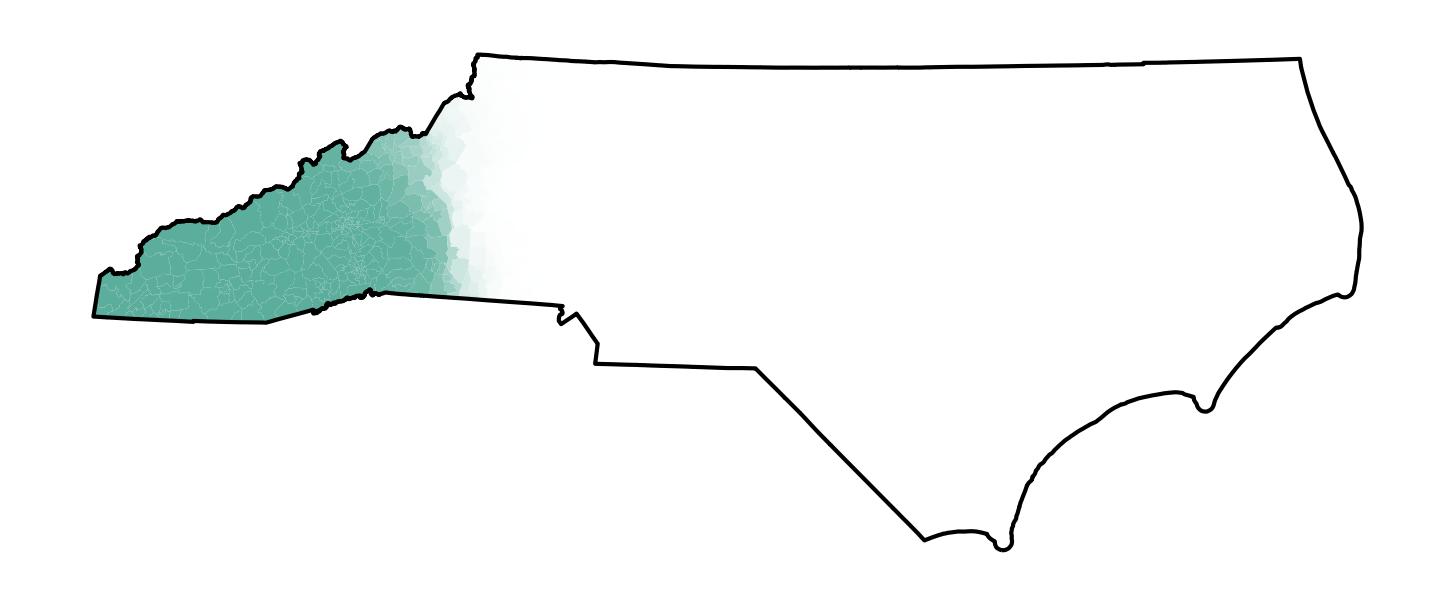}};
    \node at (-2,-0.5) {\includegraphics[width=0.25\textwidth]{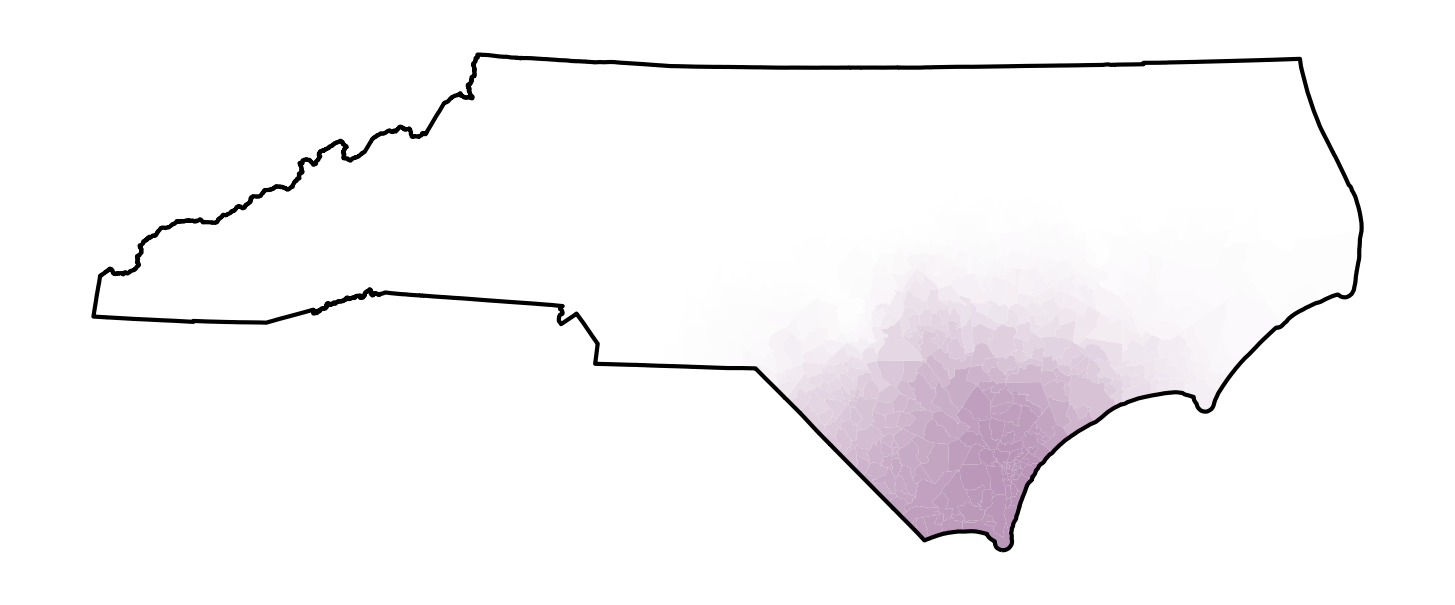}};
    \node at (-6,-2) {\includegraphics[width=0.25\textwidth]{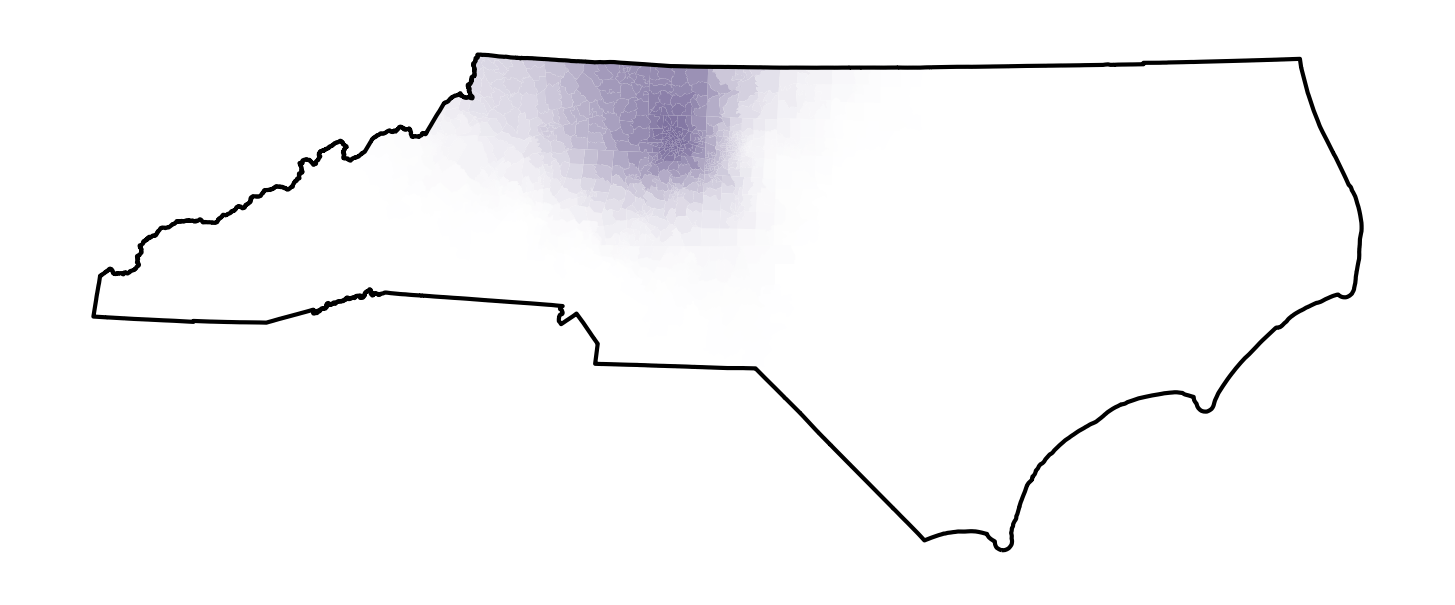}};
    \node at (-2,-2) {\includegraphics[width=0.25\textwidth]{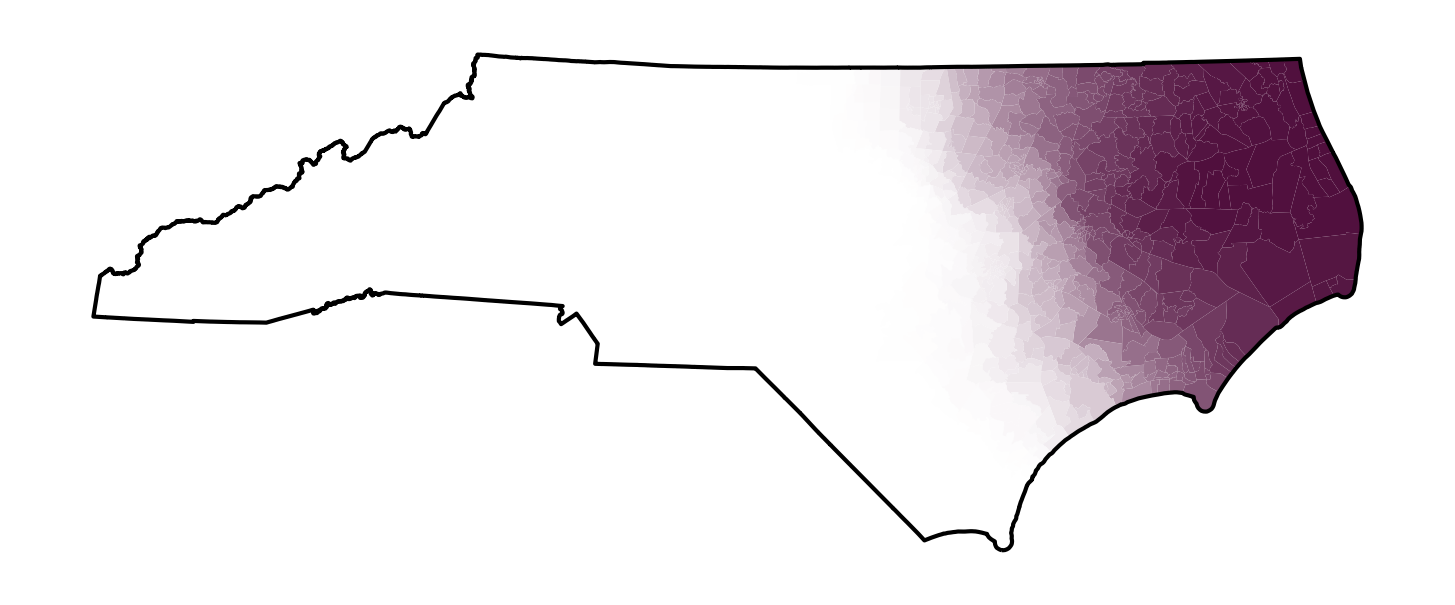}};
    \node at (-6,-3.5) {\includegraphics[width=0.25\textwidth]{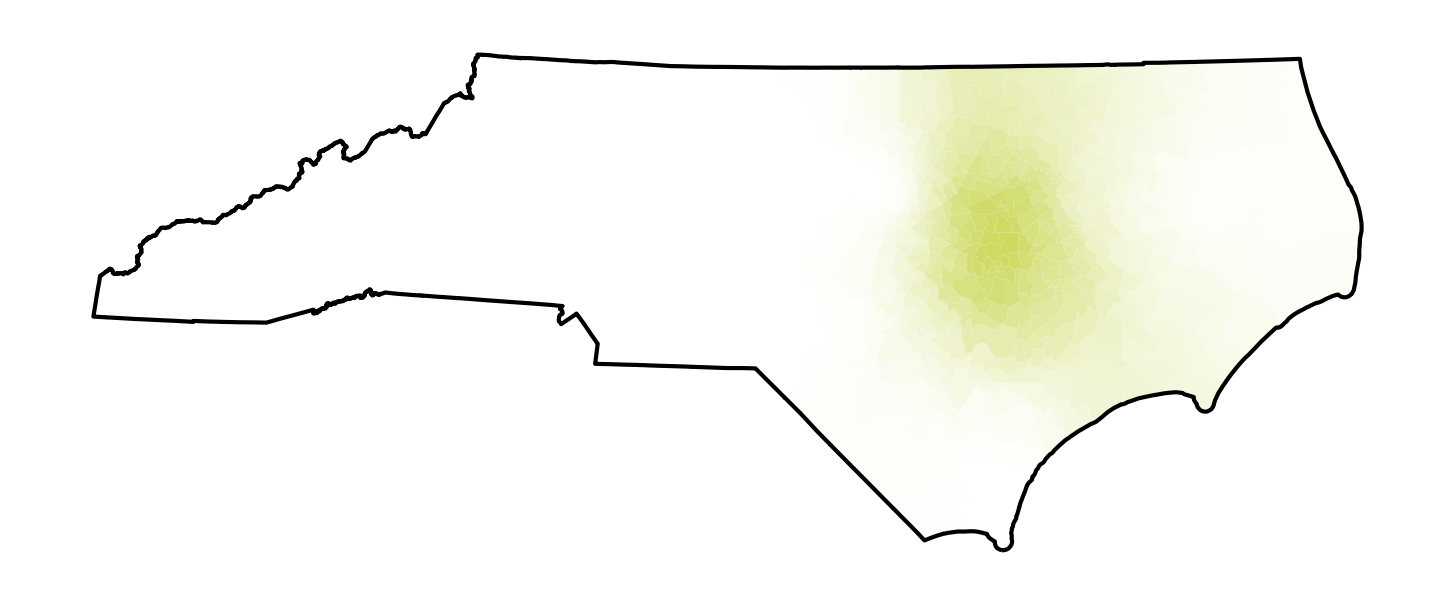}};
    \node at (-2,-3.5) {\includegraphics[width=0.25\textwidth]{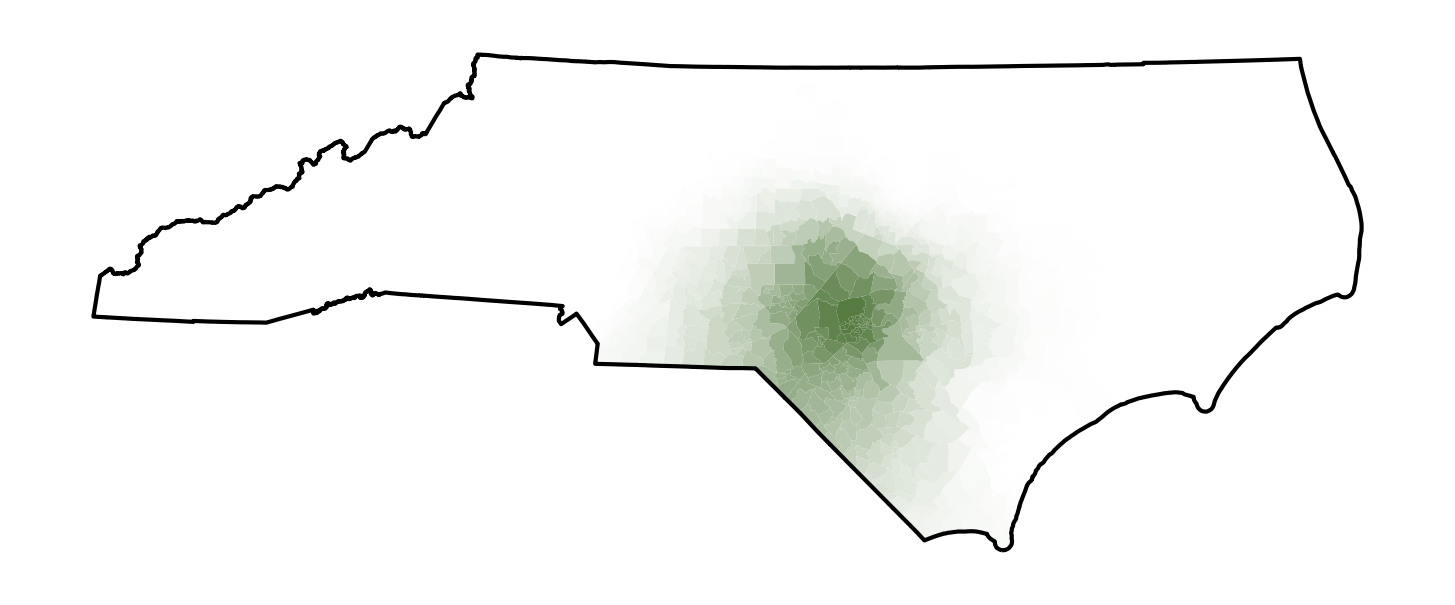}};
    \node at (-6,-5) {\includegraphics[width=0.25\textwidth]{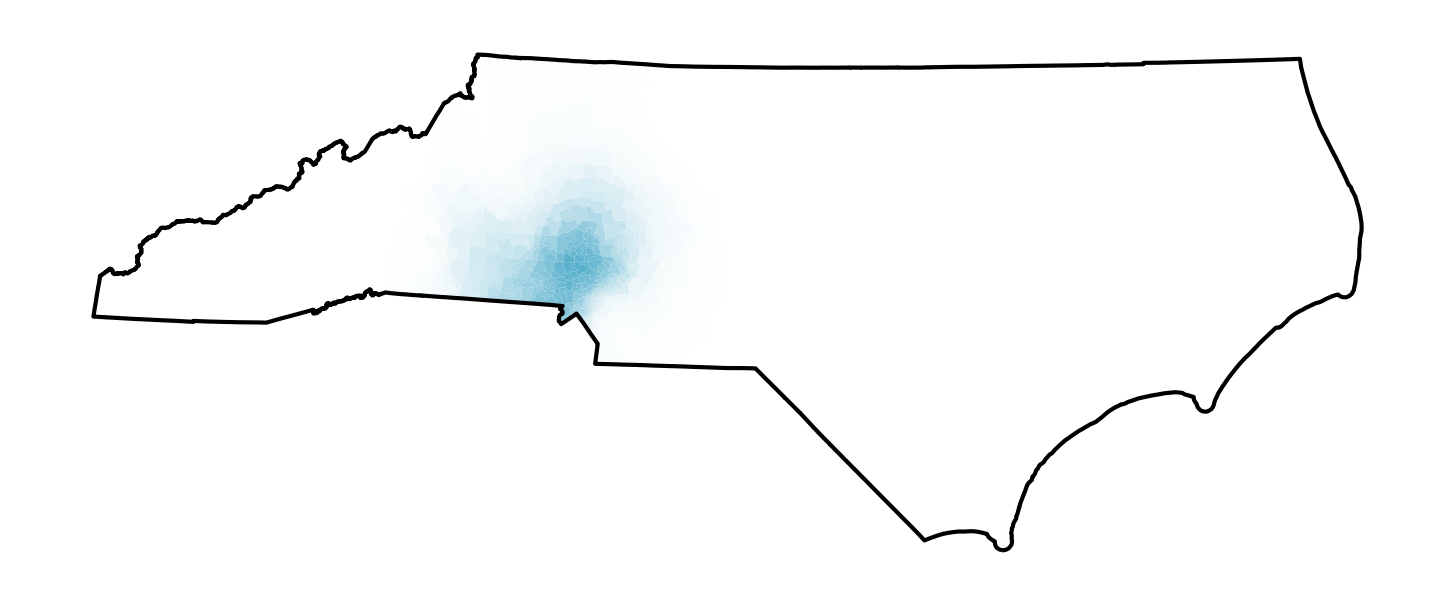}};
    \node at (-2,-5) {\includegraphics[width=0.25\textwidth]{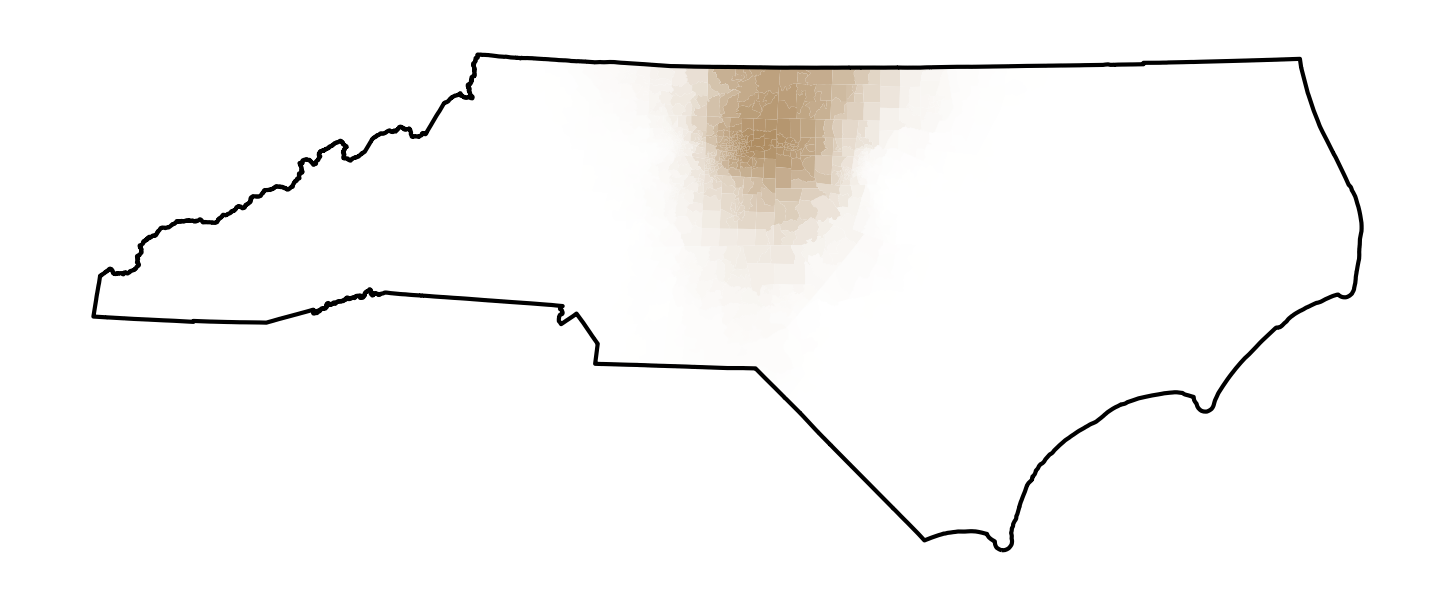}};
    \node at (-6,-6.5) {\includegraphics[width=0.25\textwidth]{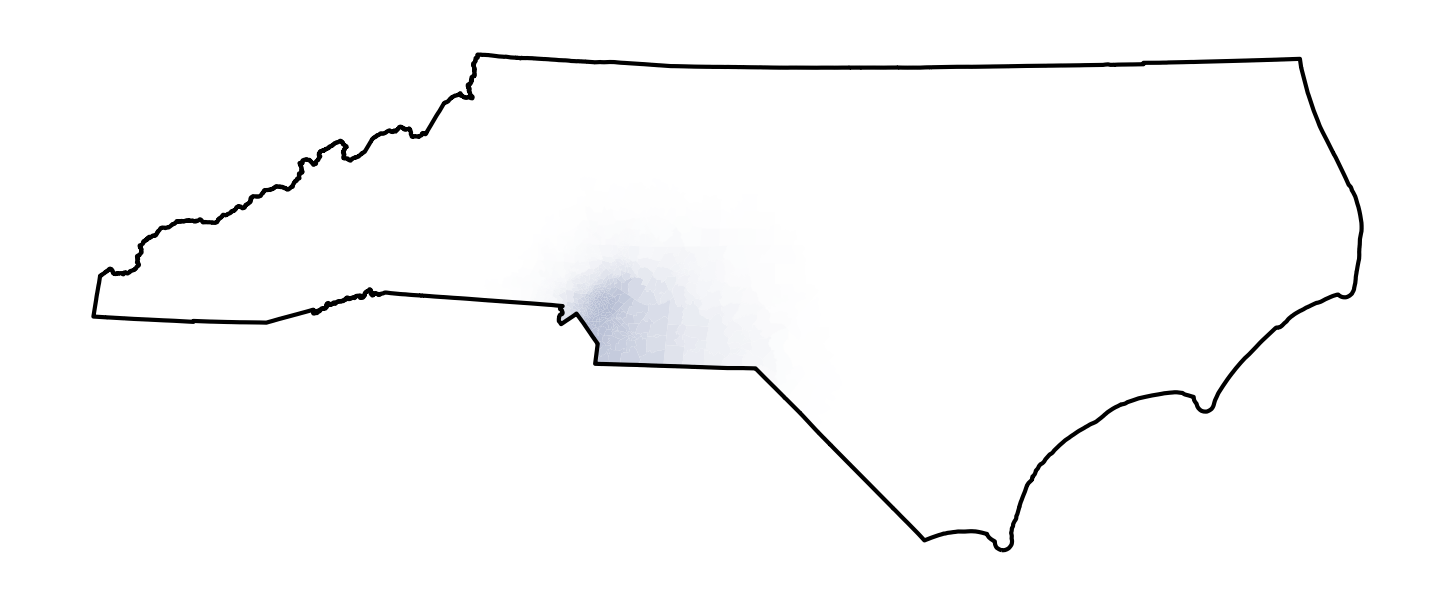}};
    \node at (-2,-6.5) {\includegraphics[width=0.25\textwidth]{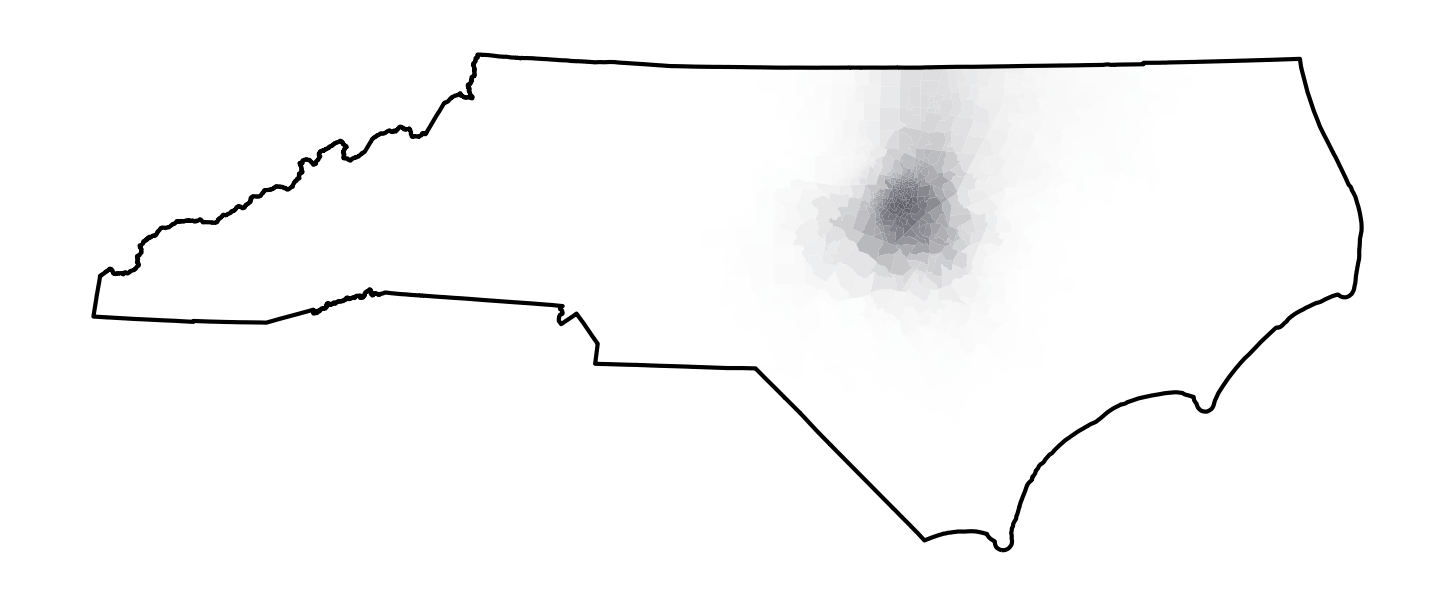}};
    \node at (-6,-8) {\includegraphics[width=0.25\textwidth]{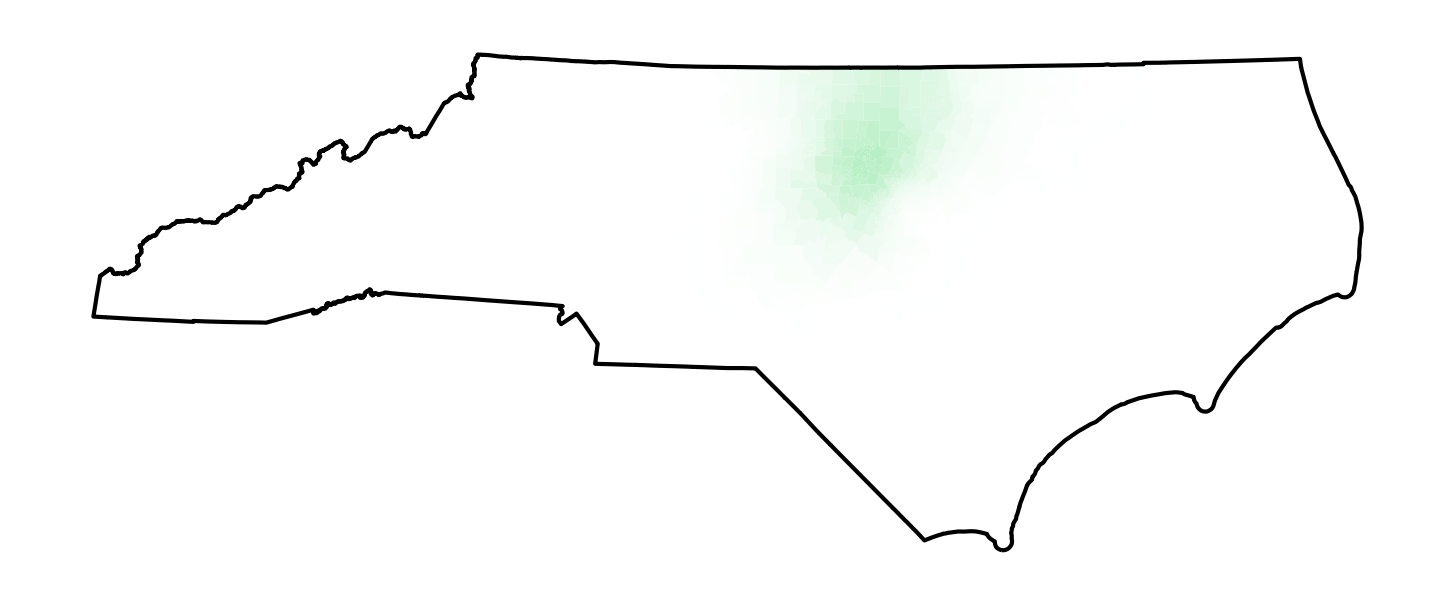}};
    \end{scope}
    
    \begin{scope}
    \node at (4,-7.5) {Presidential 2012 (left) vs Presidential 2016 (right)};
    \node at (4,-10.3) {\includegraphics[width=0.85\textwidth]{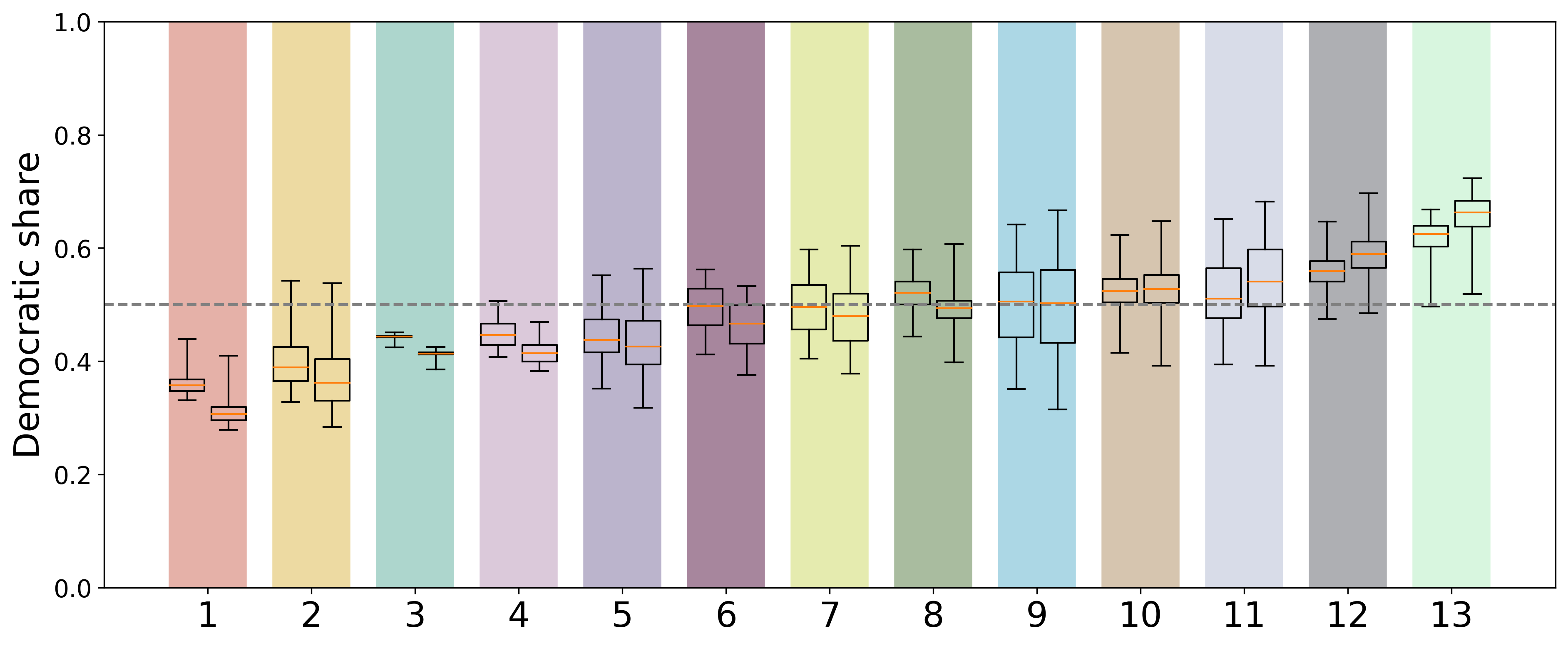}};
    \end{scope}
    \end{tikzpicture}
    }
    \caption{Barycenter for one ensemble of Congressional plans for North Carolina, with population-weighted representations used for the districts, approximated by $40$-point samples. The heat maps show the location of districts matched to each component of the barycenter. The boxplots show vote shares for the ensemble the 2012 and 2016 Presidential elections.}
    \label{fig:neutral}
\end{figure}

\subsection{Comparing elections}
In this section, we compare two-way Democratic vote shares for our ensemble based on data from the 2012 and 2016 Presidential elections. Figure \ref{fig:neutral} shows these vote shares for each district label\footnote{Following the convention in \cite{deford2019recom}, the boxes show the 25th--75th percentile range, while the whiskers indicate the 1st and 99th percentiles.}. For ease of visualization, we sort the district labels by mean Democratic vote share under the 2016 data. The statewide two-way vote shares for these two elections differ by less than 1 percent: Obama received $49\%$ of the two-way vote in 2012 and Clinton received $48\%$ in 2016. A common hypothesis used to model shifts in voting patterns in the political science literature is the \emph{uniform swing hypothesis} (see e.g.~\cite{katz2020theoretical}); in this scenario this hypothesis would dictate that Democratic vote shares in every district or even every precinct would shift by $1\%$ between these two elections. Our method allows us to probe the validity of this hypothesis at the district level for a typical plan. In this case the shift is noticeably non-uniform. Indeed, between 2012 and 2016, the more Republican Districts 1--8, mostly larger, more rural districts, tended to become even more Republican. The Democratic-leaning districts 11--13 near Charlotte and Raleigh, however, became more Democratic. This finding is consistent with the examination of (non)-uniformity of vote shifts in North Carolina using topological data analysis in \cite{duchinneedham21}.

\subsection{Comparing enacted plans in NC}
Since the districts in a computer-generated plan are unsorted, a common approach found in redistricting research and litigation is to sort them by Democratic vote share (under some choice of vote data) and do the same to the plan being evaluated \cite{mattingly1, amicus_math, mgggva, DeFord2021Recombination}. In other words, these methods are based on the \emph{order statistics} of the vote shares in an ensemble. Figure \ref{fig:rankorder} demonstrates this technique for our ensemble using 2016 Presidential vote data, along with heat maps showing where the districts for each rank lie. Clearly, the districts at a particular rank need not be geographically similar. The result is that if a district from a proposed plan is an outlier compared to districts from the ensemble of the same rank, it is impossible to say whether that district is itself unusually drawn or if it is merely placed in an unusually low or high rank as a result of unusual vote shares in other districts (which may be in completely different part of the state). 

A method based on comparing districts which are geographically similar has already been proposed by Mattingly in a blog post \cite{mattingly} and was also presented in \cite{amicus_math}. Both of these methods rely on grouping districts who share a particular geographic unit; our method takes a more geometric approach by looking at the geometric distance between districts to group them into geographic clusters. Our method can be applied to any ensemble, and answers the call by Mattingly in \cite{mattingly} for a ``more geographically localized analysis'' than order statistics.

\begin{figure}
    \centering
    \begin{subfigure}{0.33\textwidth}
    \includegraphics[width=\textwidth]{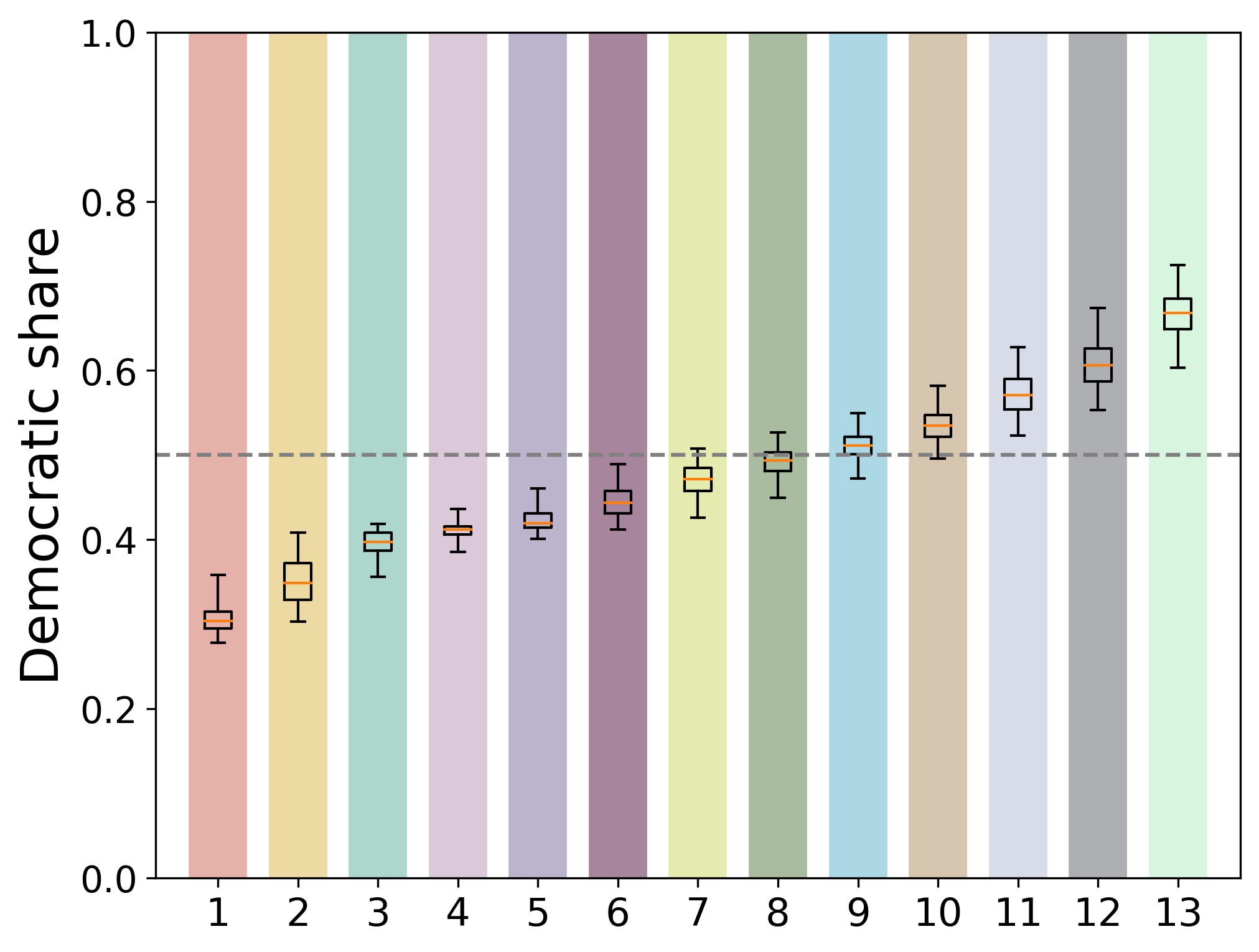}
    \end{subfigure}
    \begin{subfigure}{0.6\textwidth}
    \begin{tikzpicture}[scale=0.9]
    \foreach \i in {0,...,4}{
        \node at (2*\i,0) {\includegraphics[width=0.18\textwidth]{heatmaps/neutral_heatmaps_byDvote_NC_d\i _k13_M40_n1000.png}};
    }
    \foreach \i in {1,...,5}{
        \node at (2*\i-0.7-2,0.4) {\footnotesize $\mathsf{\i}$};
    }
    \foreach \i in {5,...,8}{
        \node at (2*\i-9,-1.4) {\includegraphics[width=0.18\textwidth]{heatmaps/neutral_heatmaps_byDvote_NC_d\i _k13_M40_n1000.png}};
    }
    \foreach \i in {6,...,9}{
        \node at (2*\i-0.7-9-2,0.4-1.4) {\footnotesize $\mathsf{\i}$};
    }
    \foreach \i in {9,...,12}{
        \node at (2*\i-17,-2.8) {\includegraphics[width=0.18\textwidth]{heatmaps/neutral_heatmaps_byDvote_NC_d\i _k13_M40_n1000.png}};
    }
    \foreach \i in {10,...,13}{
        \node at (2*\i-0.7-17-2,0.4-2.8) {\footnotesize $\mathsf{\i}$};
    }
    \end{tikzpicture}
    \end{subfigure}
    \caption{Another way to label districts found in redistricting literature is to order them by Democratic vote share. Applying this with 2016 Presidential vote data to the neutral ensemble used in this section results in the boxplots above showing the vote shares for each label. The heat maps indicate the location of the districts assigned to each label, which are geographically very diffuse.}\label{fig:rankorder}
\end{figure}

We compare the ensemble against the Congressional plans used in 2012, 2016 and 2020 respectively, as well as a plan proposed by a bipartisan panel of judges (which we call the `Judges' plan). The 2012, 2016 and Judges plan were analyzed in \cite{mattingly} and the 2012 and 2016 plans were found to be gerrymanders based on a order-statistics comparison with a neutral ensemble (generated by a different algorithm). We represent these plans in the same way as those in the ensemble and match them to the barycenter to get a labeling (which need not agree with the legal names of these districts). Figure \ref{fig:enacted} shows the results. 

We will make the somewhat arbitrary choice to call a district an \emph{outlier} if it has Democratic vote share which is more than one percentage point ($0.01$ on the plots) outside the 1st--99th percentile range of the ensemble. In the 2012 plan, we note that there are five outliers: Districts 2, 3, 7, 12 and 13. District 2 (legally named the $12^{th}$ Congressional District) in the 2012 plan has a highly unusual shape and is therefore not just a vote share outlier but a geometric outlier too, making it hard to compare with the ensemble. However, since 2012's District 2 has a higher Democratic vote share than any district in the entire ensemble, it does not really matter where it is placed, it will still be an outlier.  In the 2016 plan we find three outliers: Districts 3, 7 and 9. Note the unusually low Democratic vote shares in District 3 in both the 2012 and 2016 plan despite the very narrow range in ensemble values. This phenomenon was observed in \cite{duchinneedham21} using other methods. The Judges plan has no outliers, while the 2020 plan has one: District 9. In the area-weighted analysis (Figure \ref{fig:enactedbyarea} in the Appendix), some labelings change, but not the number of outliers in each plan.

\begin{figure}
    \centering
    \begin{tikzpicture}
    \begin{scope}
    \node at (0,0) {\includegraphics[width=0.28\textwidth]{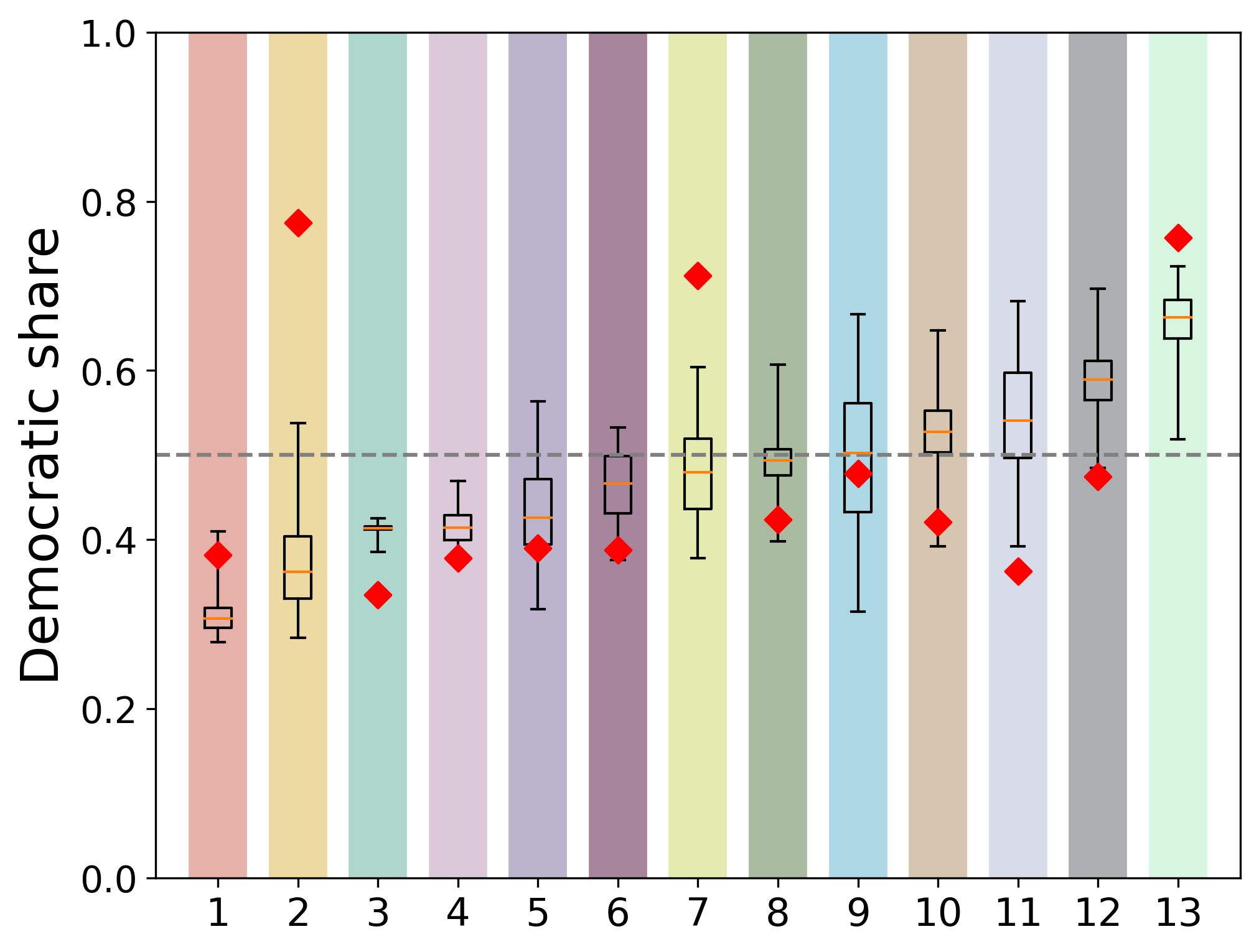}};
    \node at (3.6, 0.8) {2012};
    \node at (3.6,0) {\includegraphics[width=0.2\textwidth]{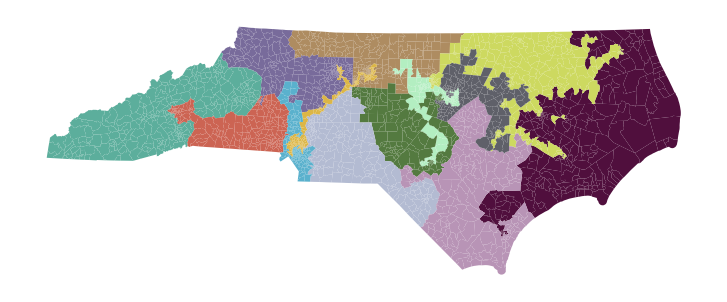}};
    \end{scope}
    
    \begin{scope}[xshift=7.4cm]
    \node at (0,0) {\includegraphics[width=0.28\textwidth]{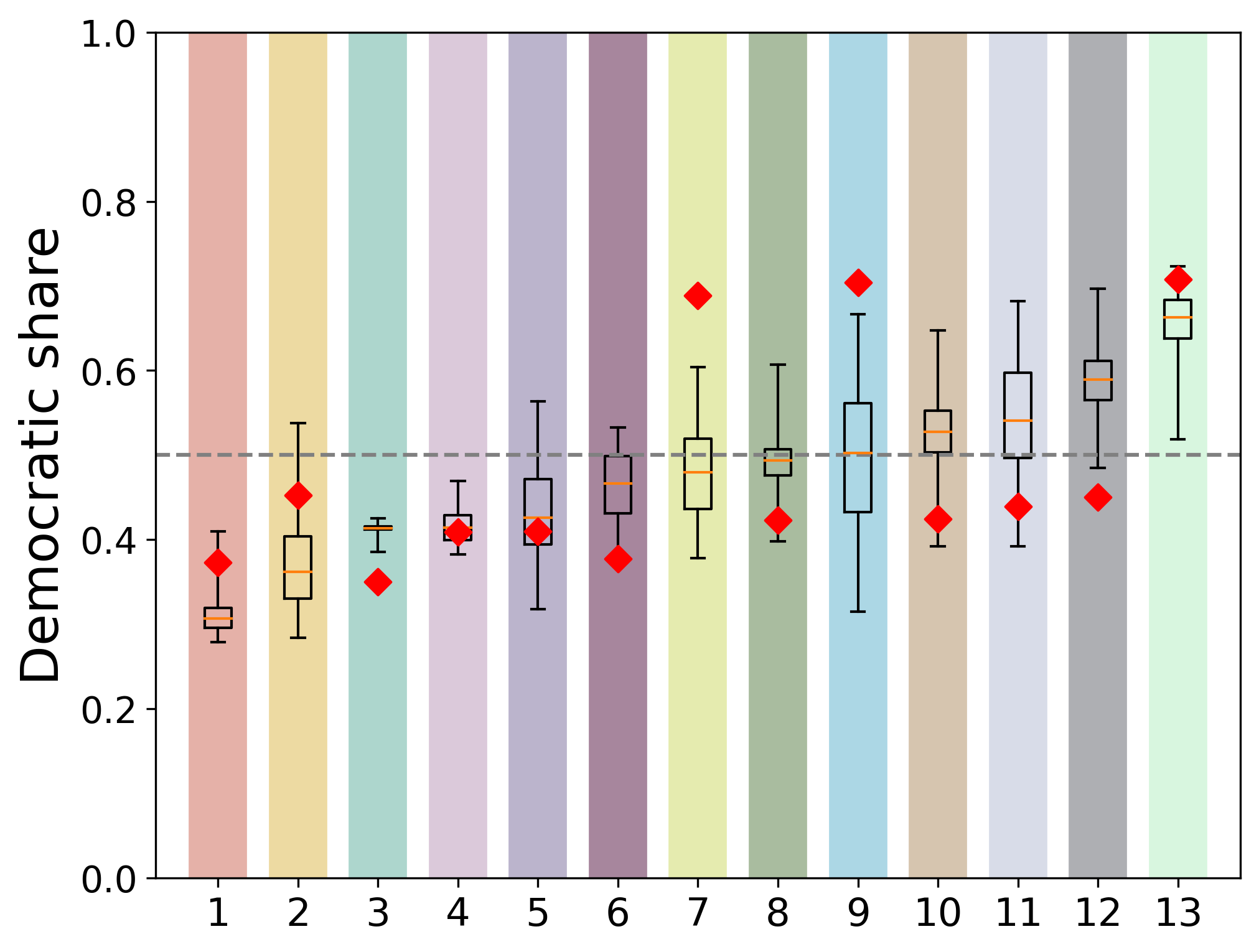}};
    \node at (3.6, 0.8) {2016};
    \node at (3.6,0) {\includegraphics[width=0.2\textwidth]{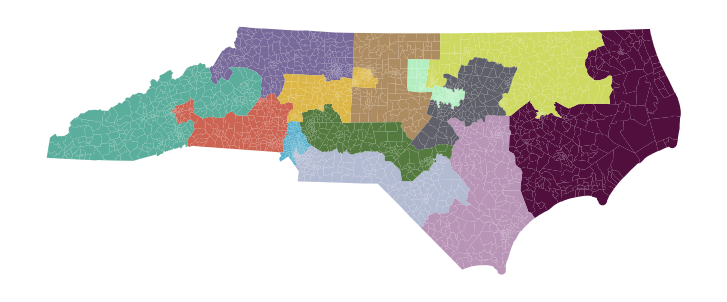}};
    \end{scope}
    
    \begin{scope}[yshift=-3.8cm]
    \node at (0,0) {\includegraphics[width=0.28\textwidth]{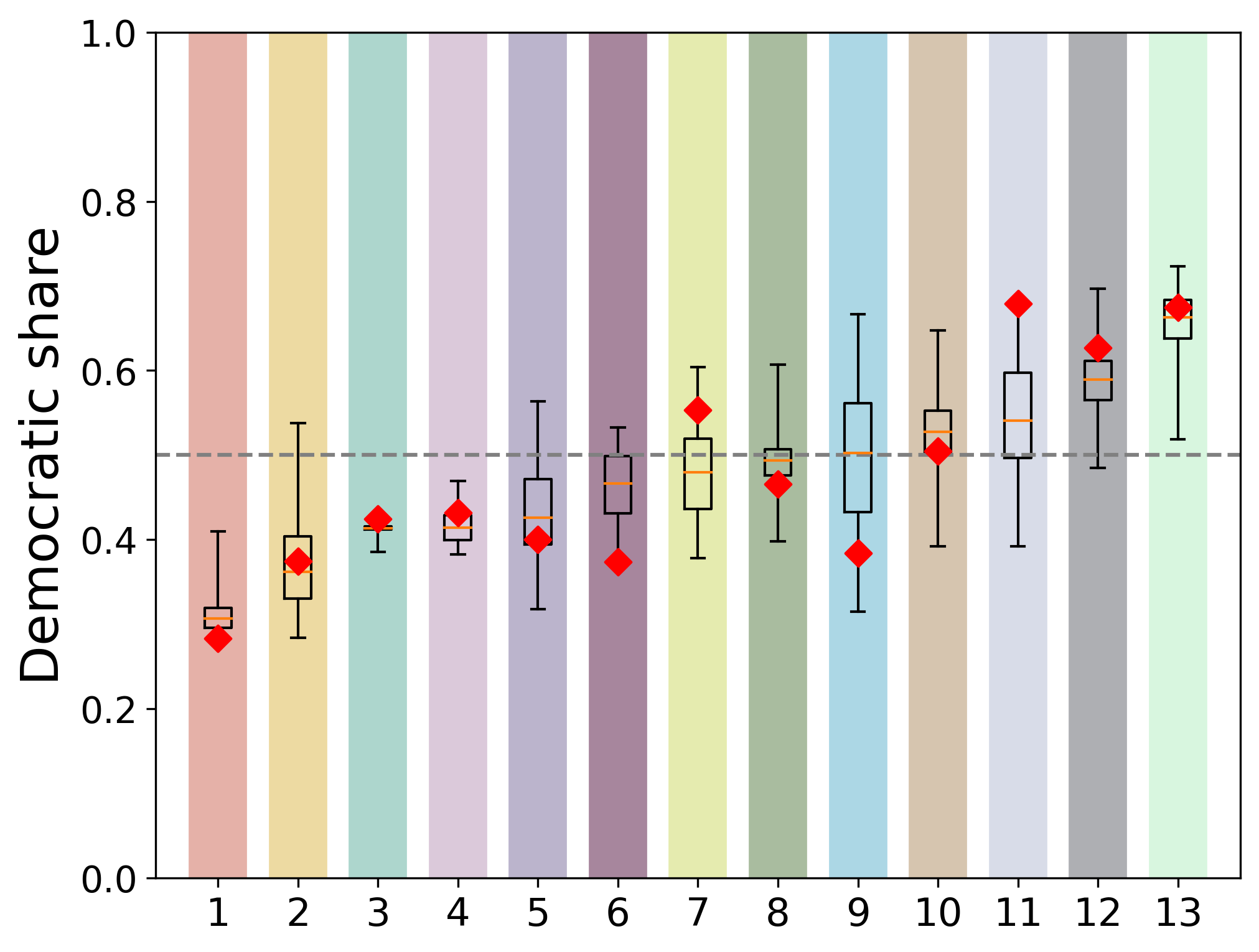}};
    \node at (3.6, 0.8) {Judges};
    \node at (3.6,0) {\includegraphics[width=0.2\textwidth]{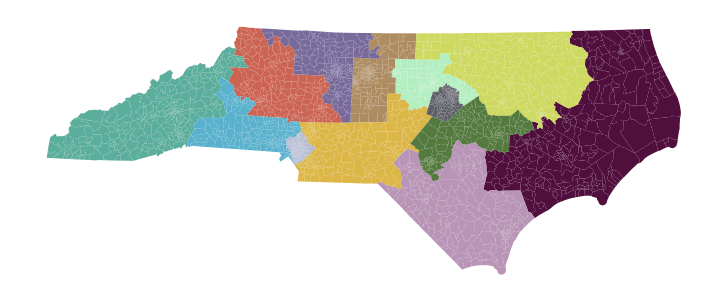}};
    \end{scope}
    
    \begin{scope}[yshift=-3.8cm, xshift=7.4cm]
    \node at (0,0) {\includegraphics[width=0.28\textwidth]{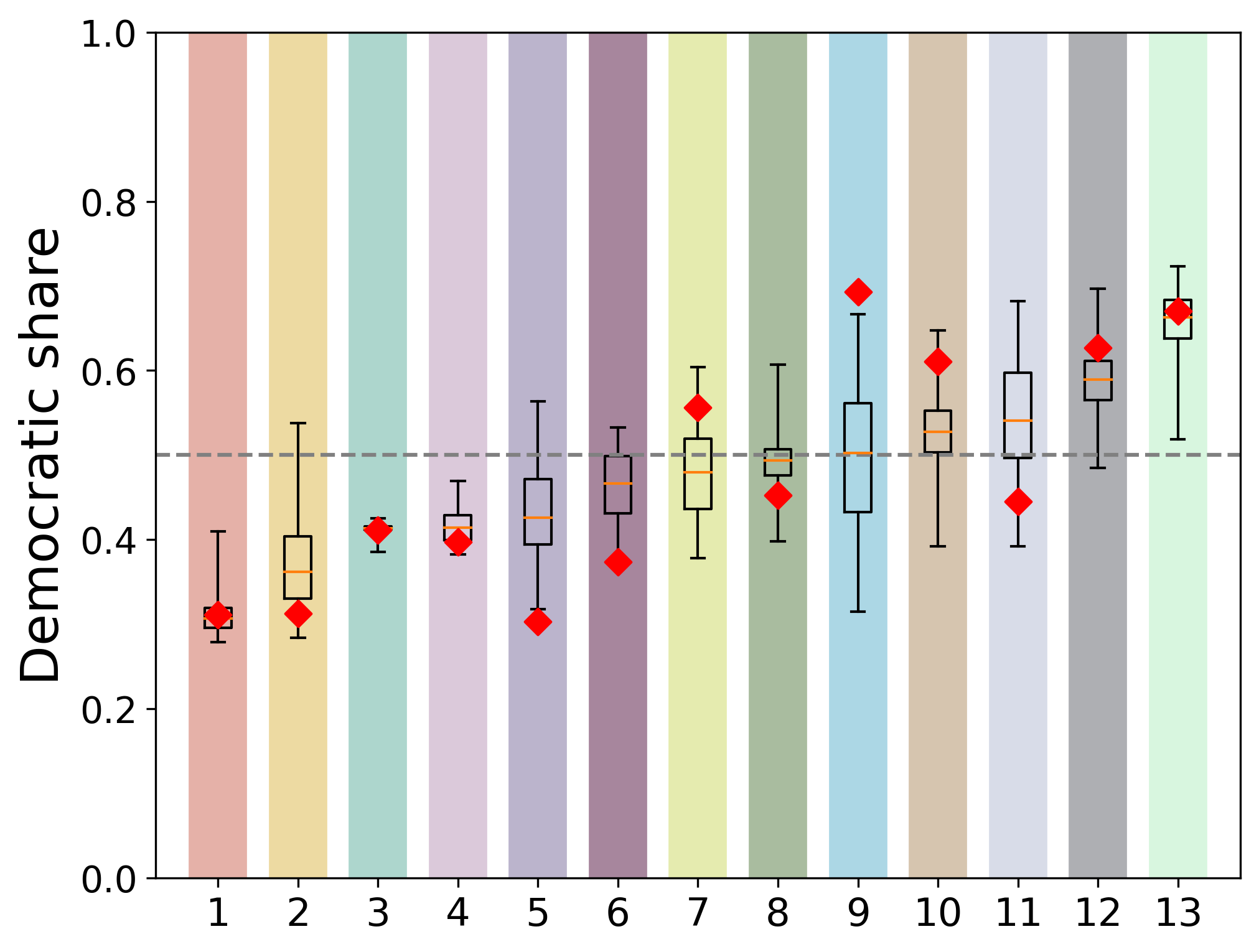}};
    \node at (3.6, 0.8) {2020};
    \node at (3.6,0) {\includegraphics[width=0.2\textwidth]{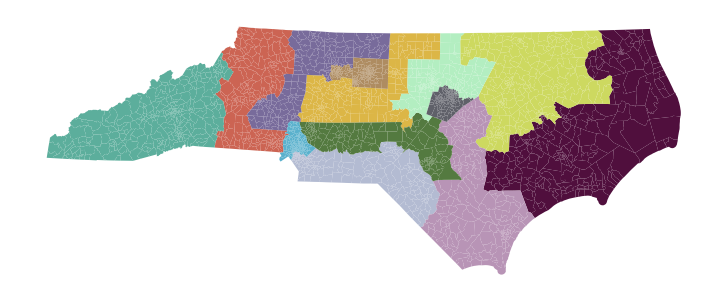}};
    \end{scope}
    
    \end{tikzpicture}
    \caption{Comparison of four enacted or proposed Congressional plans for North Carolina using votes shares from the Presidential 2016 race. Population-weighting was used. The red diamonds in each boxplot indicate the vote shares of the plan being evaluated for the district matched to that component of the barycenter. The maps on the right show the plan in question colored by a best matching to the barycenter.}\label{fig:enacted}
\end{figure}

\section{Conclusion and Future Work}

We have introduced a method for labeling unordered $k$-tuples in a geometrically coherent way using local barycenters in symmetric product spaces. The algorithm (Algorithm \ref{algorithm}) for computing these local barycenters is very general and depends on a inner local barycenter operation in a modular way. We have demonstrated how this method enables a new analysis technique for redistricting ensembles, effectively summarizing and organizing large sample sets from the non-linear and extremely diverse set of possible redistricting plans for a state. Beyond redistricting, we expect this approach to have applications to problems in machine learning involving successive applications of $k$-means or other classifiers to partition multiple datasets into unlabeled clusters. 

This work suggests many directions of future research. For example, it would be interesting to study theoretical properties of local $p$-descent operators on non-Euclidean data, as in Section \ref{sec:circle} for the circle. The development of more advanced statistical machine learning methods for symmetric product spaces  would be useful for more in-depth analysis of clustering algorithms, as in Section \ref{sec:clustering_algorithms}, or for studying spaces of districting plans in further detail.

\section*{Acknowledgements}
The authors would like to thank Justin Solomon for discussions early on in the project, and Olivia Walch for the color scheme used for districts in the paper. 

\bibliographystyle{siamplain}
\bibliography{bib}

\FloatBarrier

\newpage 

\appendix

\section{Choice of seed and $M$ for redistricting application}\label{sec:stability}
In this section we test the dependence of the ensemble barycenter on the choice of seed for Algorithm \ref{algorithm} and also motivate the choice of $M=40$ sample points per district. Since we mainly interested in the labeling of districts induced by a given barycenter, we measure the discrepancy between seeds by the number of relabelings required, assuming good labelings for the resulting barycenters. To be precise, for a barycenter $B$ with a chosen ordering, let $B_i$ be all the districts labeled District $i$ by matching to $B$. We define the \emph{discrepancy} $D(B, B')$ to be the fraction
$$
\min_{\phi: \langle k \rangle \to \langle k \rangle} \sum_i \frac{|B_i \setminus B'_{\phi(i)}|}{|B_i|}
$$
where $\phi$ ranges over all bijections $\langle k \rangle \to \langle k \rangle$. 

We run Algorithm \ref{algorithm} 1000 times on the neutral ensemble, each time using a different plan as a seed $x_0$. For each pair $0\leq j < 1000$, we compute the discrepancy between the barycenters coming from seed $0$ (the one used in the previous sections) and seed $j$ and display these values in Figure \ref{fig:stability}. Comparing the population-weighted and area-weighted representations, we see that the population-weighted version has 20 seeds with $>5\%$ change, while the area-weighted version has 4. On the other hand, for other seeds the discrepancy for the population-weighted representation was generally lower than the area-weighted version. Overall, for both versions, at least 98\% of seeds had less than 2\% difference from seed $0$.

\begin{figure}
    \centering
    \resizebox{\textwidth}{!}{
    \begin{tikzpicture}
    \node at (0,2) {\underline{Population-weighted}};
    \node at (6,2) {\underline{Area-weighted}};
    \node[rotate=90] at (-3.1,0.2) {\tiny fraction of labels changed};
    \node at (0,0) {\includegraphics[width=0.35\textwidth]{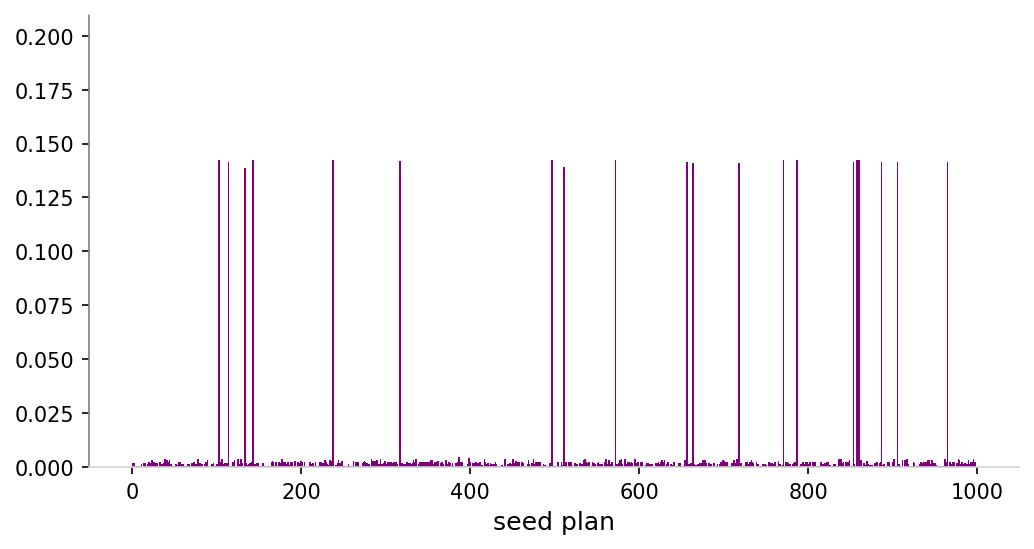}};
    \node at (6,0) {\includegraphics[width=0.35\textwidth]{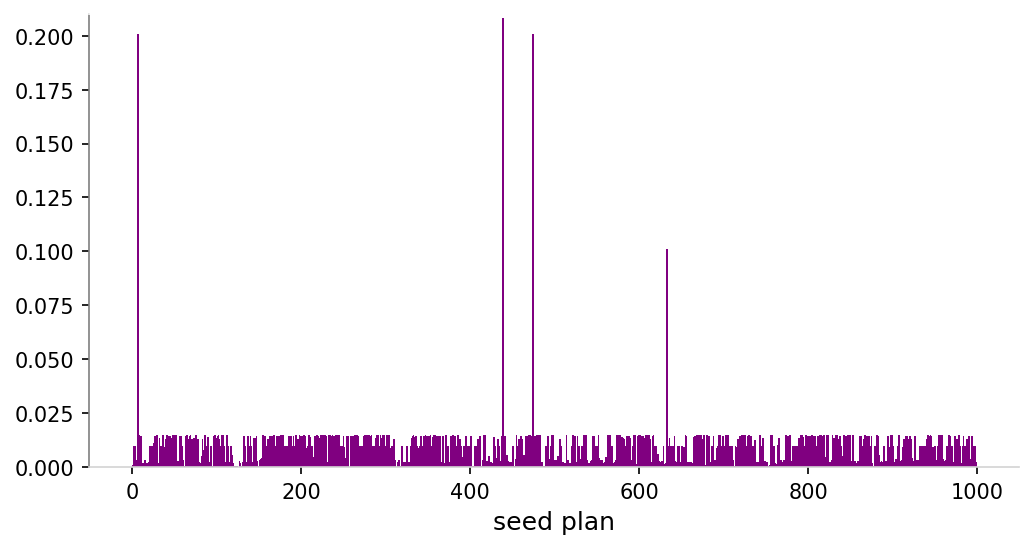}};
    \end{tikzpicture}
    }
    \caption{Measuring the dependence of the method on the choice of seed plan $\x_0$. These plots show the discrepancy between the district labels when using the first plan (seed 0) and the $i^{th}$ plan as the seed.}
    \label{fig:stability}
\end{figure}

In order to find the right number of points to sample from each district, we first sample $M=40$ points from each district. We then run Algorithm \ref{algorithm} with a fixed seed plan, using only the first $i$ sample points for each district for each $i \in \{1,\ldots,40\}$ to produce a series of $40$ barycenters and labelings. For each $t \in \{1,\ldots,39\}$, we compute the discrepancy between the barycenter with $t$ sample points and the barycenter with $t+1$ sample points. Figure \ref{fig:sampling} shows the results. We see that for both the population-weighted and area-weighted representations, the discrepancy between successive values of $t$ drops to around $1\%$ at around $t=20$ and remains low thereafter. 

\begin{figure}
  \begin{minipage}[c]{0.4\textwidth}
    \includegraphics[width=\textwidth]{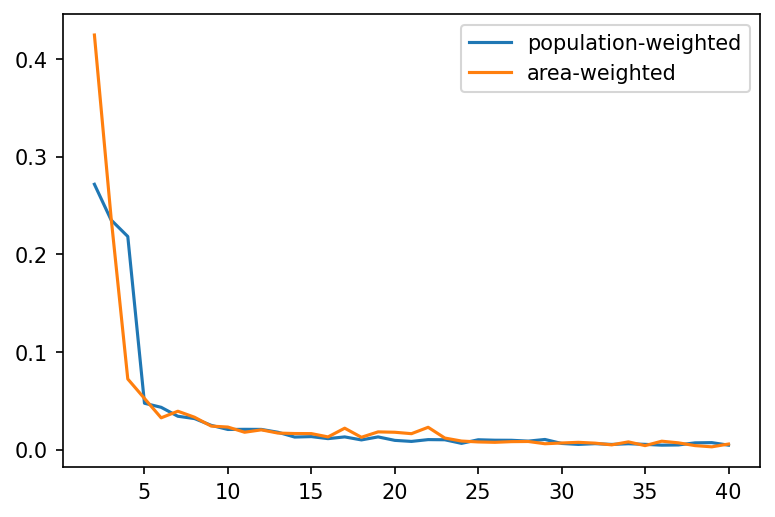}
  \end{minipage}\hfill
  \begin{minipage}[l]{0.6\textwidth}
    \caption{
    Discrepancy between using the first $t$ points and using the first $t+1$ points of the $40$-point samples to compute the barycenter for the neutral ensemble.
    } \label{fig:sampling}
  \end{minipage}
\end{figure}

\section{Computational details}\label{sec:compdetails}
For the experiments in Section \ref{sec:redistrict}, we implemented Algorithm \ref{algorithm} in Python.\footnote{Code is available at \url{https://github.com/thomasweighill/barymandering}} The (outer) $2$-descent operator $\LB_1$ is implemented using the Python Optimal Transport library \cite{flamary2021pot} (the particular function used implements the restricted version of Algorithm 2 in \cite{cuturi2014fast} discussed in Remark \ref{cuturialgorithm2}). Cleaned population and vote data was obtained from the \texttt{mggg-states} repository \cite{mggg_states}. For the neutral ensemble, we sampled every 50th plan from an ensemble of 50,000 plans generated using the ReCom algorithm \cite{deford2019recom} implemented in the Python library GerryChain~\cite{gerrychain}. The chain was constrained to generate only contiguous districts and population deviation less than 2\% from the ideal district population. To give some idea of the computational cost of the method, computing the (population-weighted) barycenter for the neutral ensemble of 1000 Congressional plans for North Carolina (represented in Figure \ref{fig:neutral}) was performed on a single core of a high performance cluster and took about 3,200 seconds (25 iterations) to complete.

\section{Supplemental and area-weighted redistricting figures}\label{sec:area}

\begin{figure}
    \centering
    \resizebox{32pc}{!}{
    \begin{tikzpicture}
    \begin{scope}
    \node at (-4, 2.2) {Location of matched districts};
    \node at (5, 2.2) {Barycenter};
     \node at (5, -2.3) {Overlaid heat maps};
    
    \node at (5,0) {\includegraphics[width=0.7\textwidth]{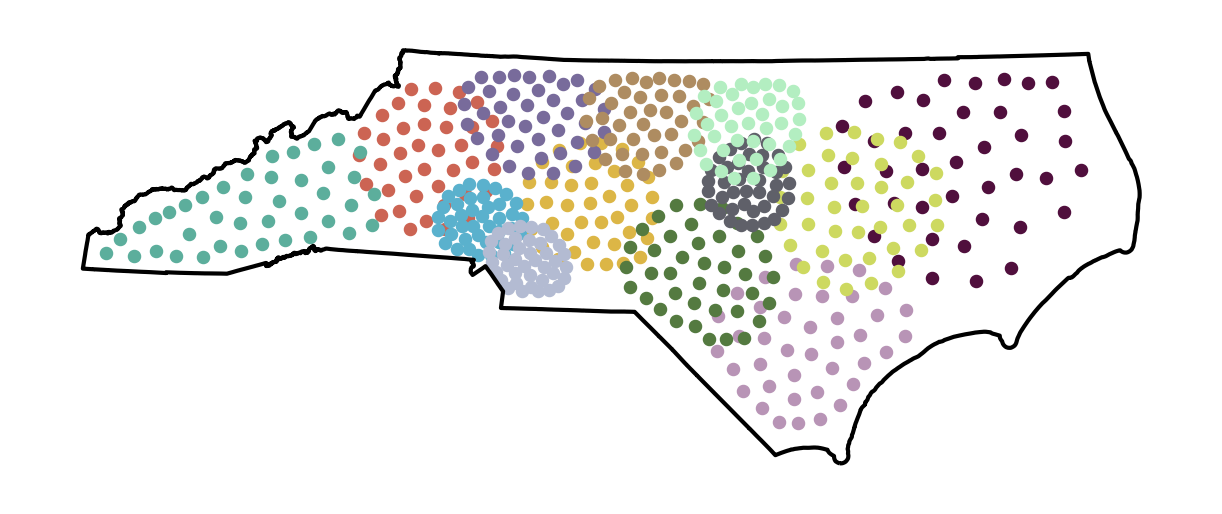}};
    \node at (5,-4.5) {\includegraphics[width=0.7\textwidth]{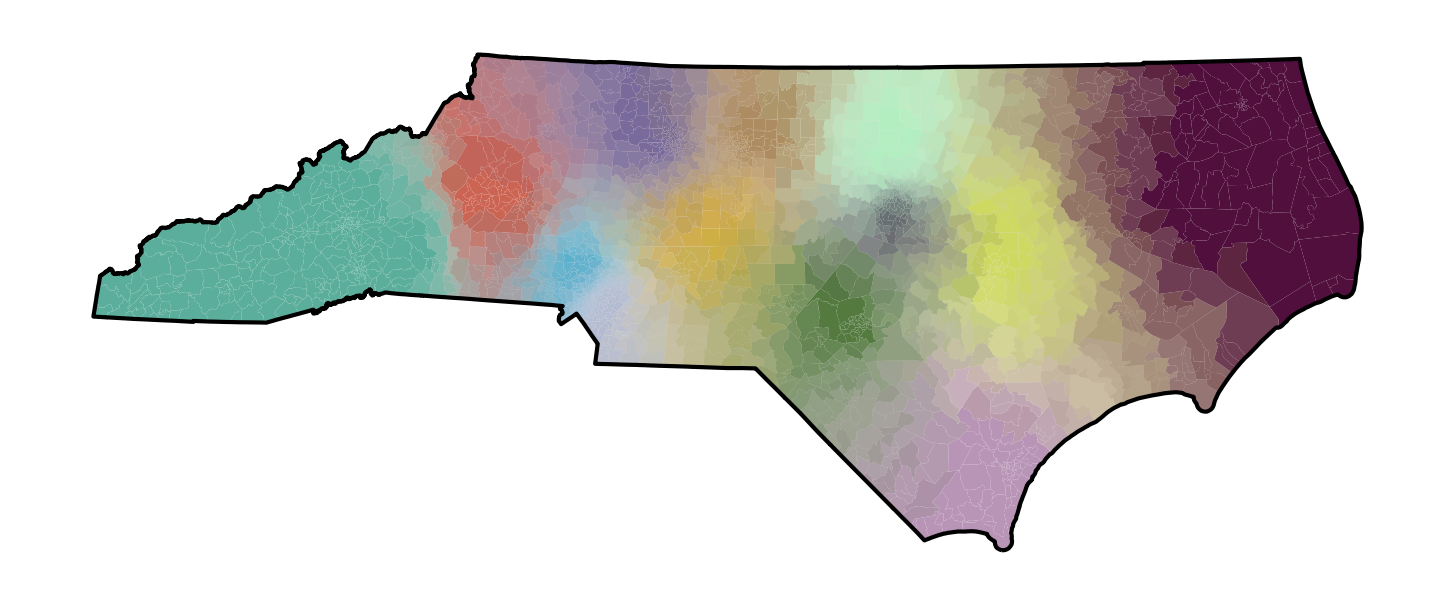}};
    
    \node at (-6,1) {\includegraphics[width=0.25\textwidth]{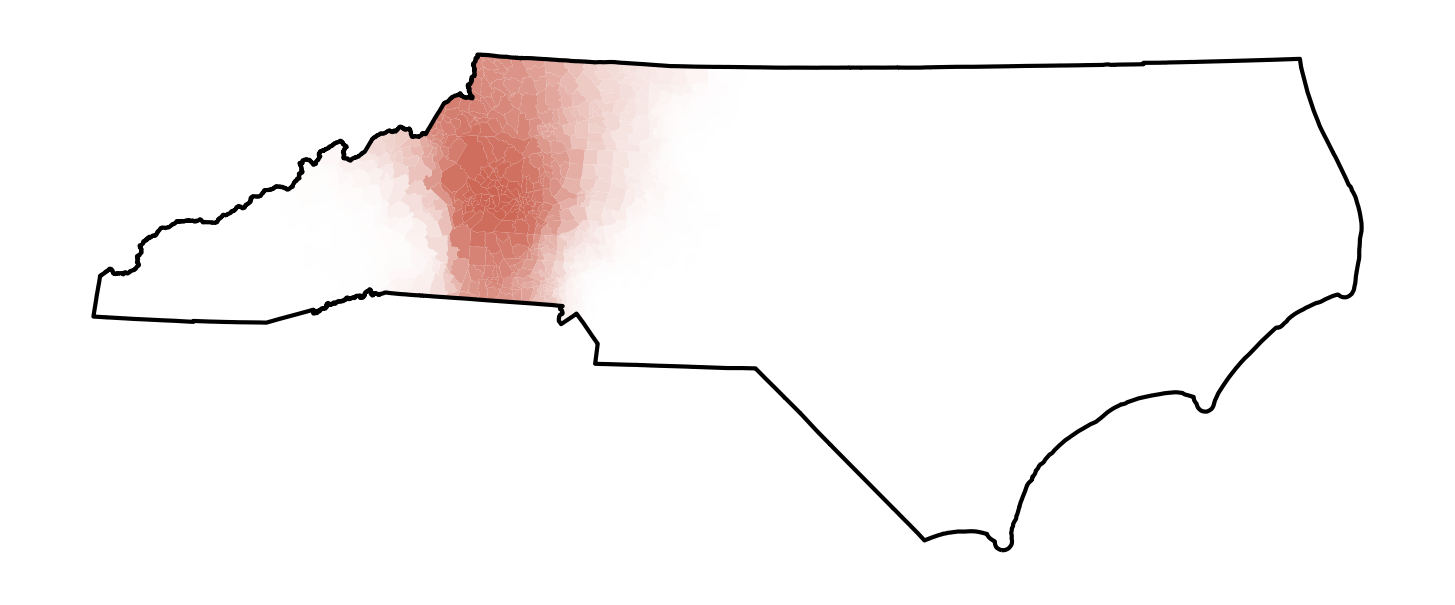}};
    \node at (-2.5,1) {\includegraphics[width=0.25\textwidth]{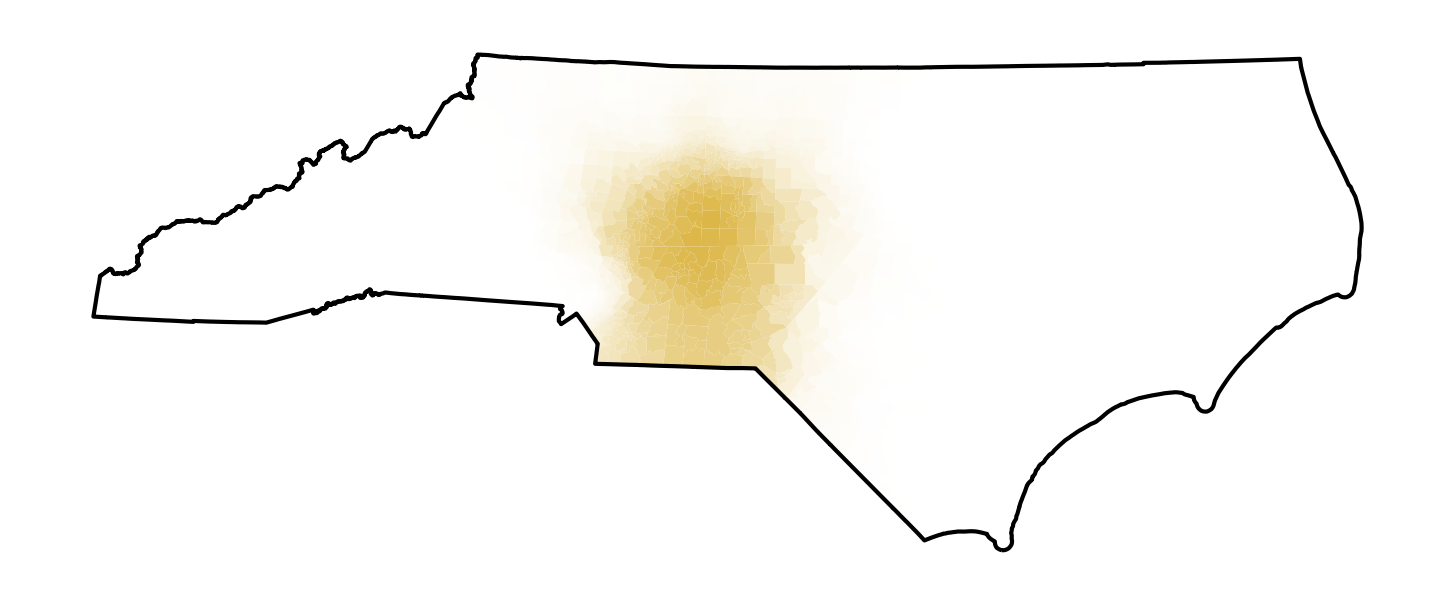}};
    \node at (-6,-0.5) {\includegraphics[width=0.25\textwidth]{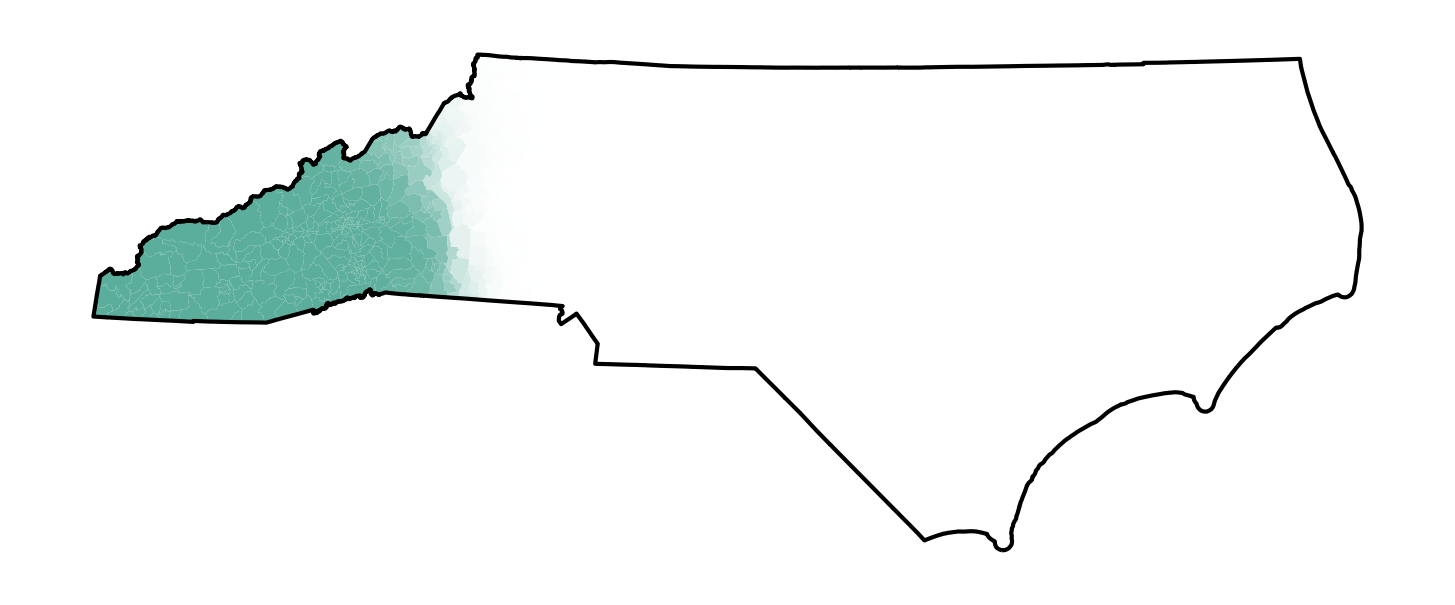}};
    \node at (-2.5,-0.5) {\includegraphics[width=0.25\textwidth]{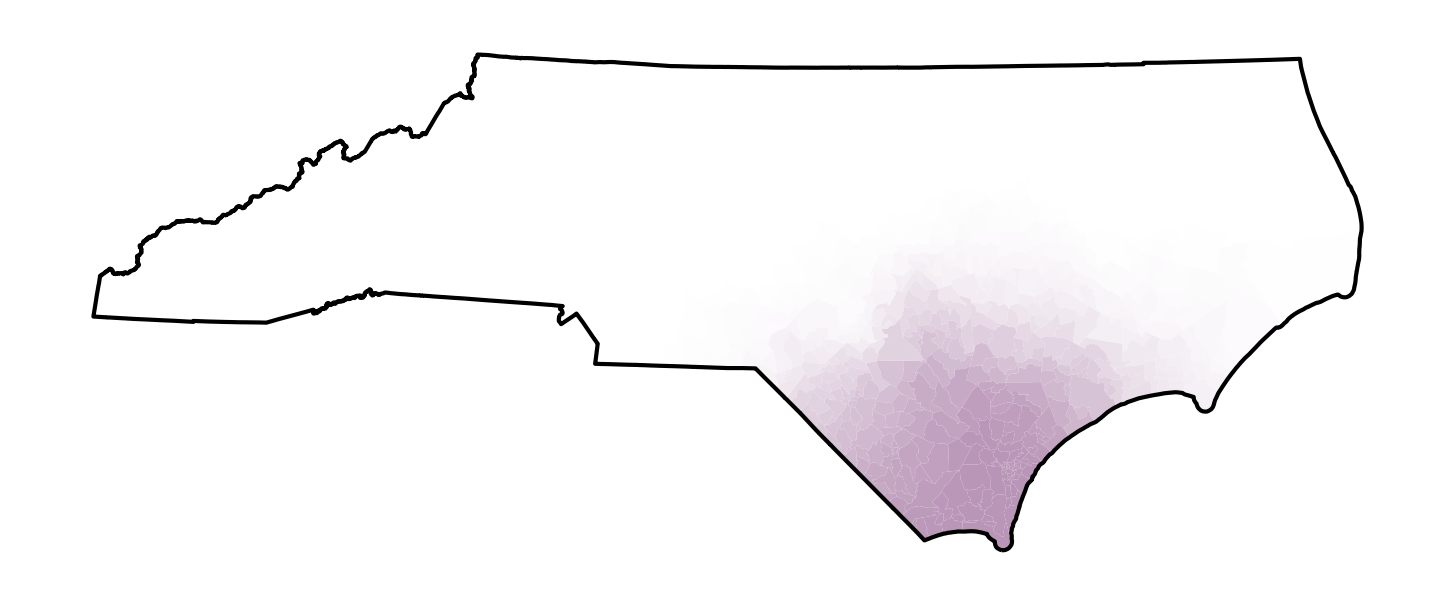}};
    \node at (-6,-2) {\includegraphics[width=0.25\textwidth]{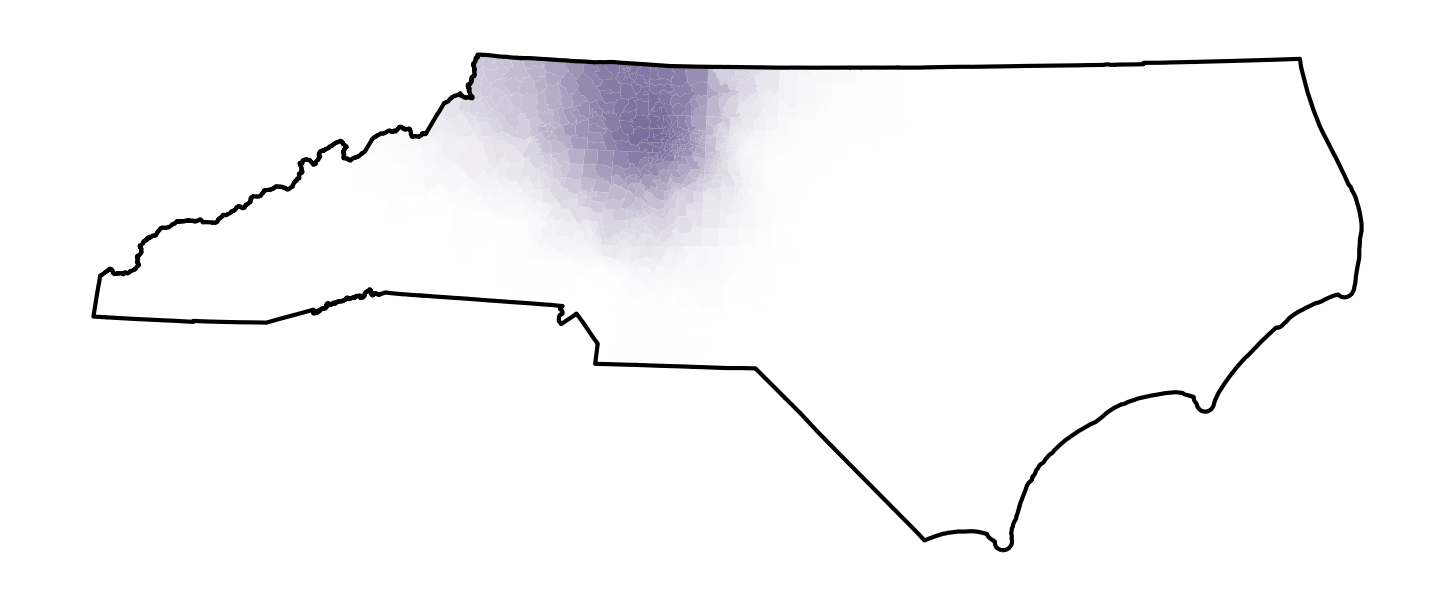}};
    \node at (-2.5,-2) {\includegraphics[width=0.25\textwidth]{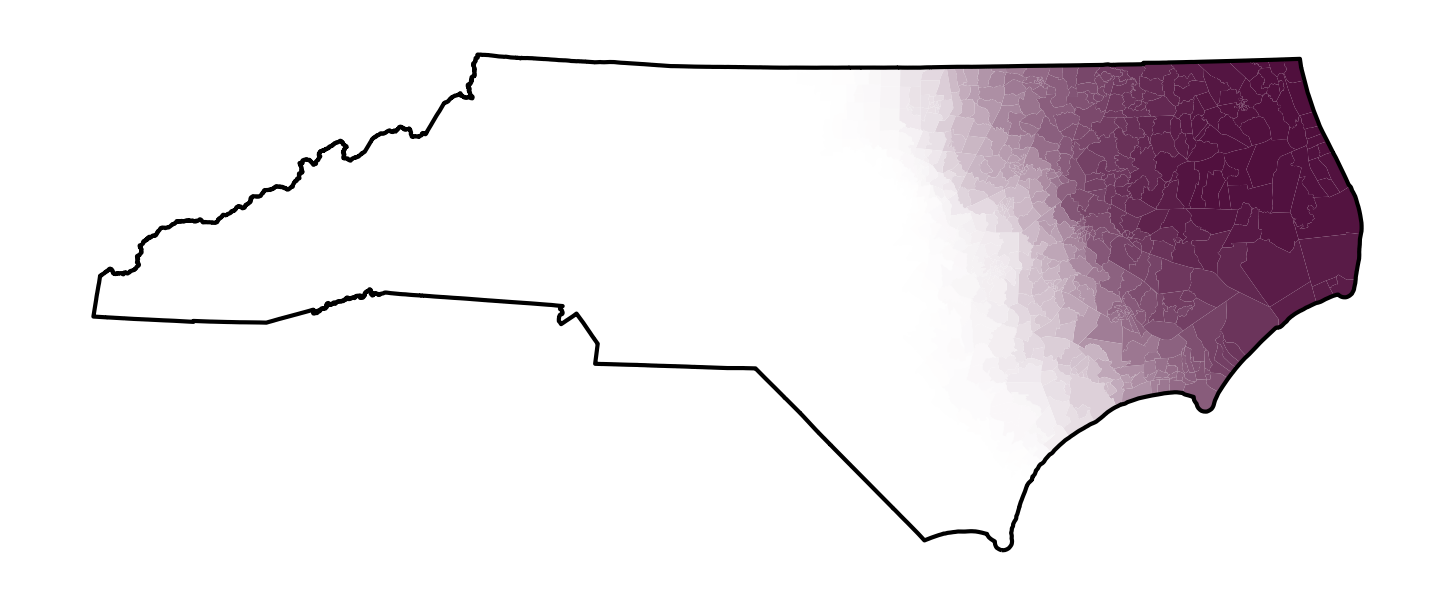}};
    \node at (-6,-3.5) {\includegraphics[width=0.25\textwidth]{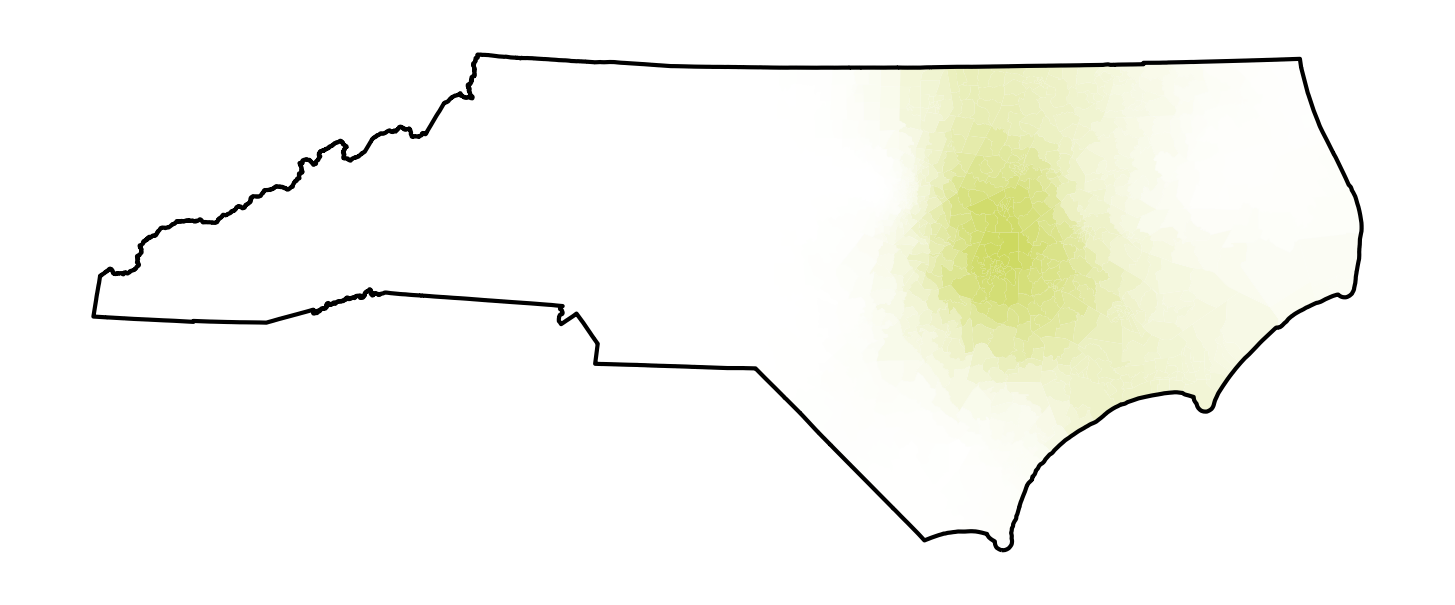}};
    \node at (-2.5,-3.5) {\includegraphics[width=0.25\textwidth]{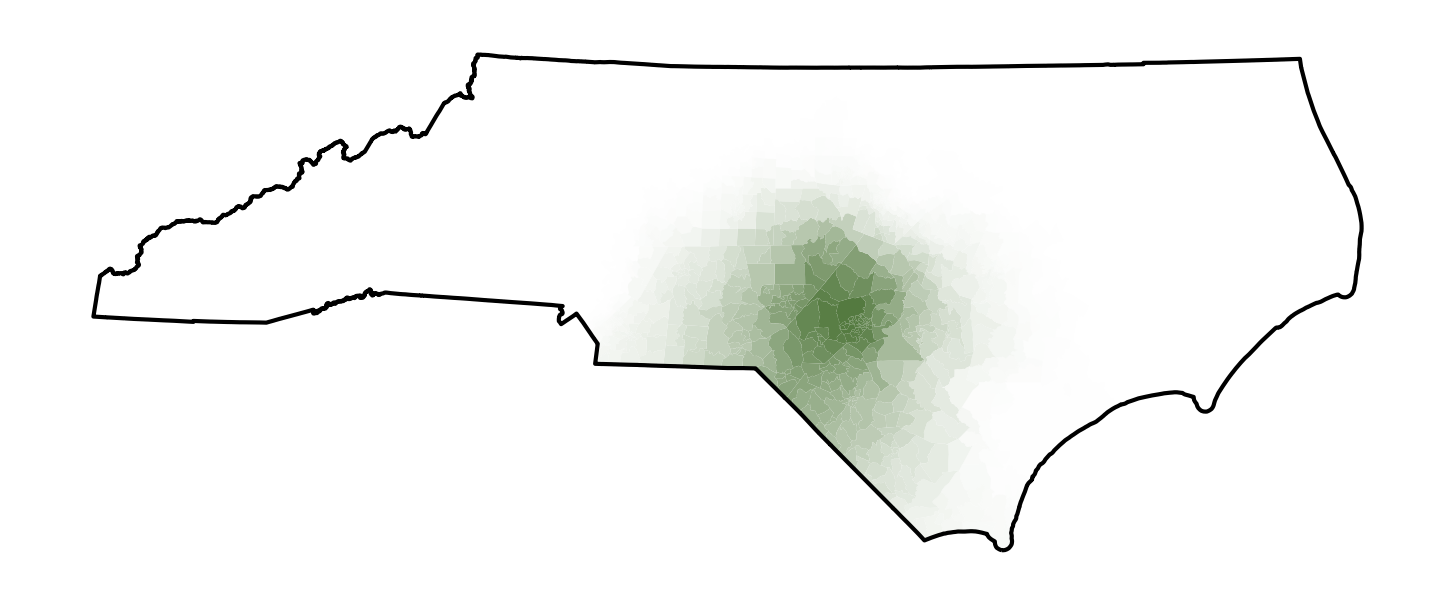}};
    \node at (-6,-5) {\includegraphics[width=0.25\textwidth]{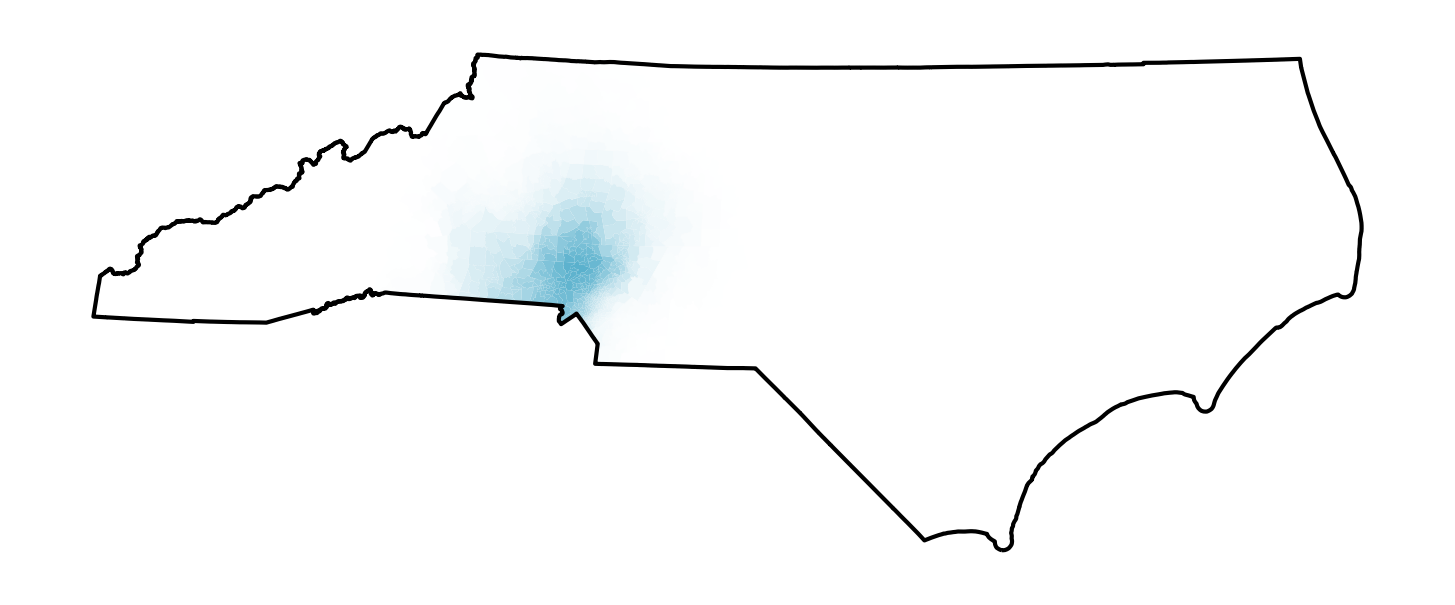}};
    \node at (-2.5,-5) {\includegraphics[width=0.25\textwidth]{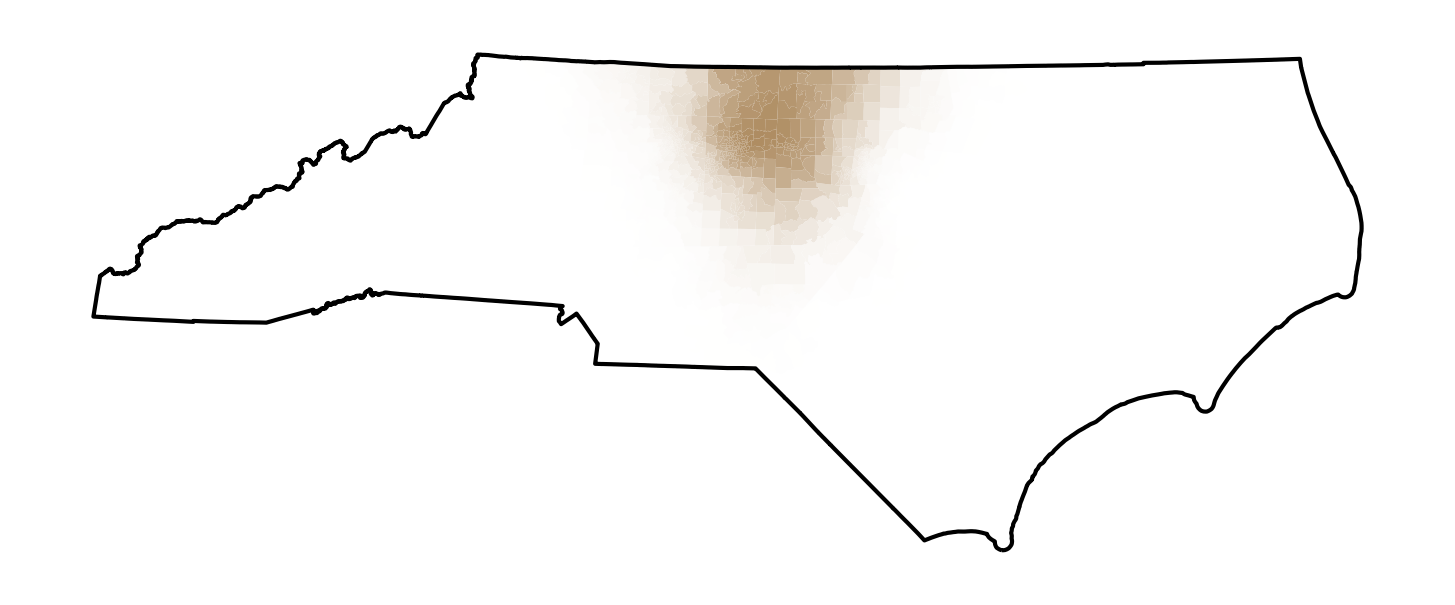}};
    \node at (-6,-6.5) {\includegraphics[width=0.25\textwidth]{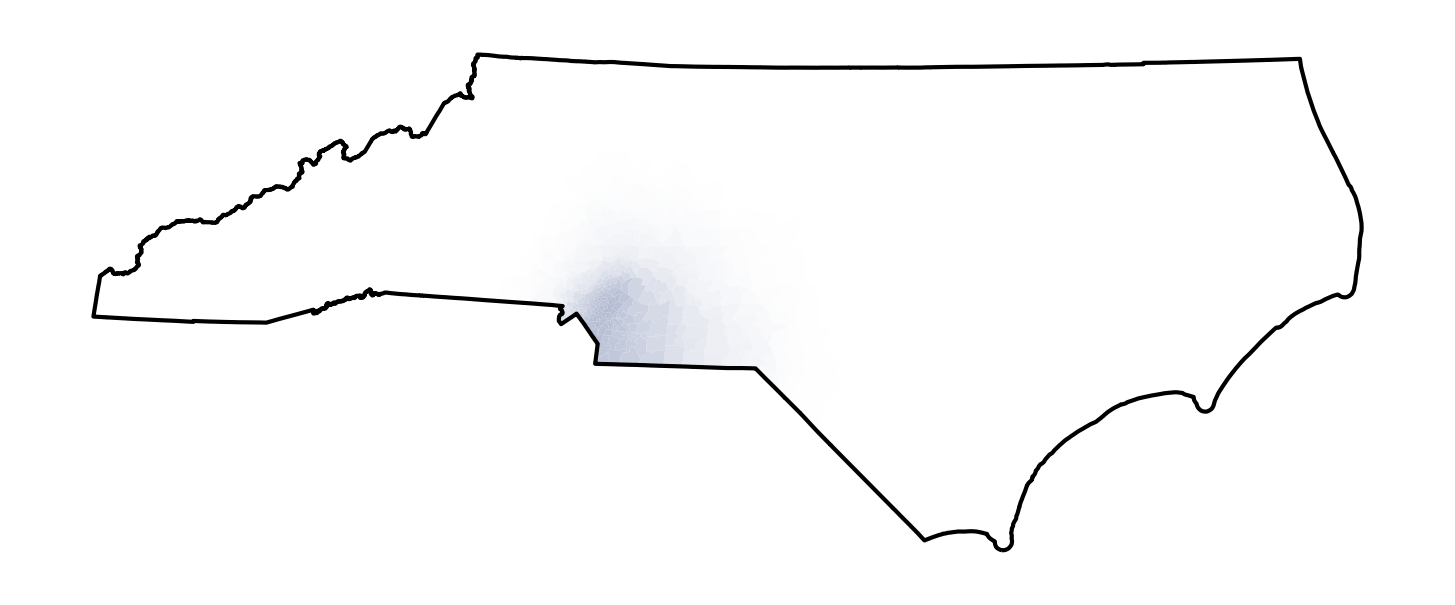}};
    \node at (-2.5,-6.5) {\includegraphics[width=0.25\textwidth]{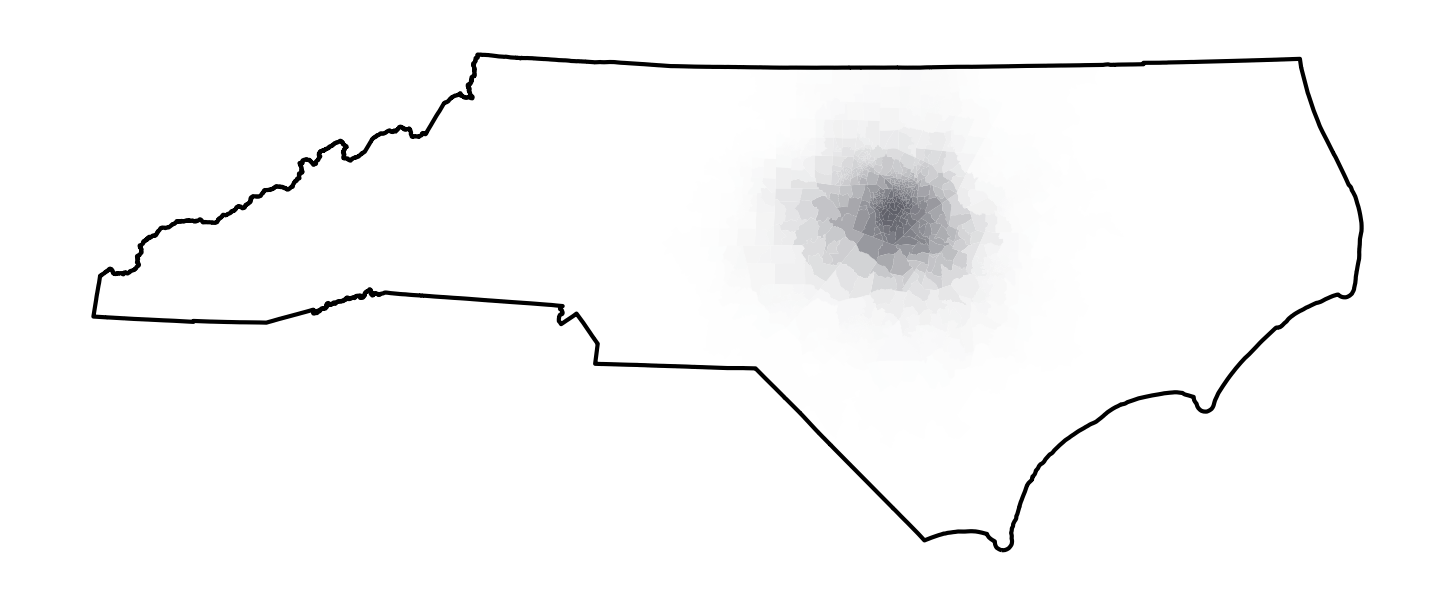}};
    \node at (-6,-8) {\includegraphics[width=0.25\textwidth]{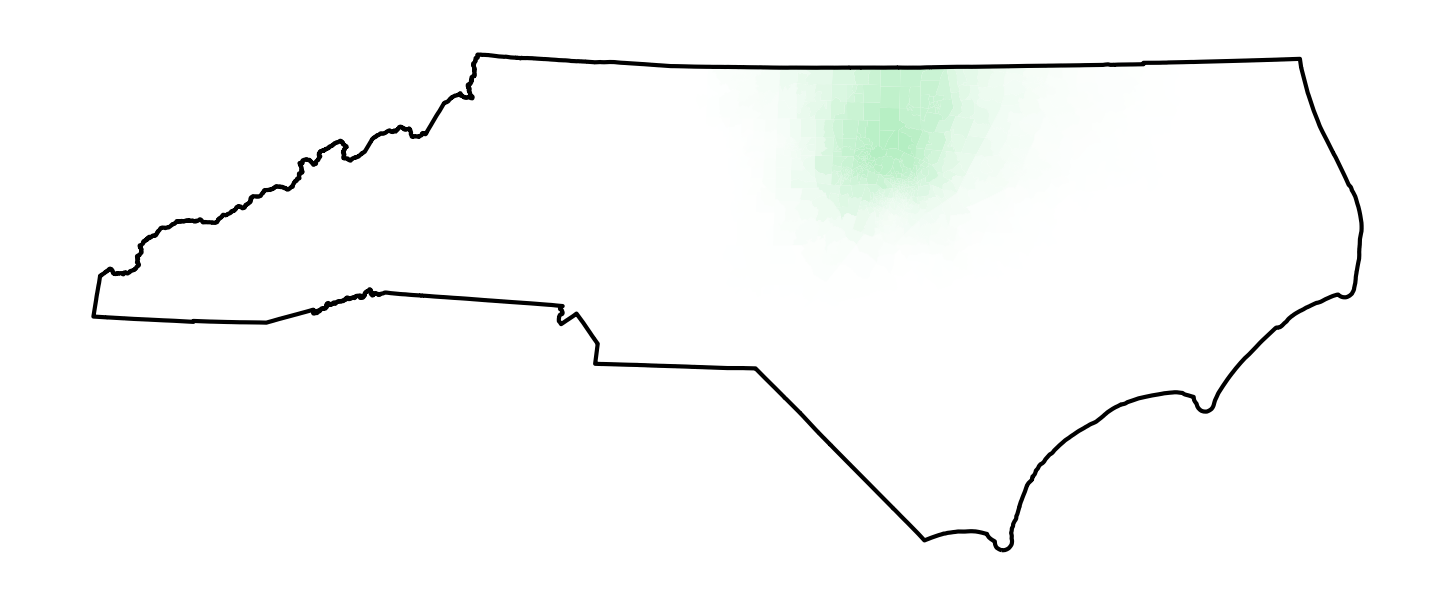}};
    \end{scope}
    
    \begin{scope}[yshift=-12cm]
    \node at (0.5,3.2) {Presidential 2012 (left) vs Presidential 2016 (right)};
    \node at (0.5,0) {\includegraphics[width=0.8\textwidth]{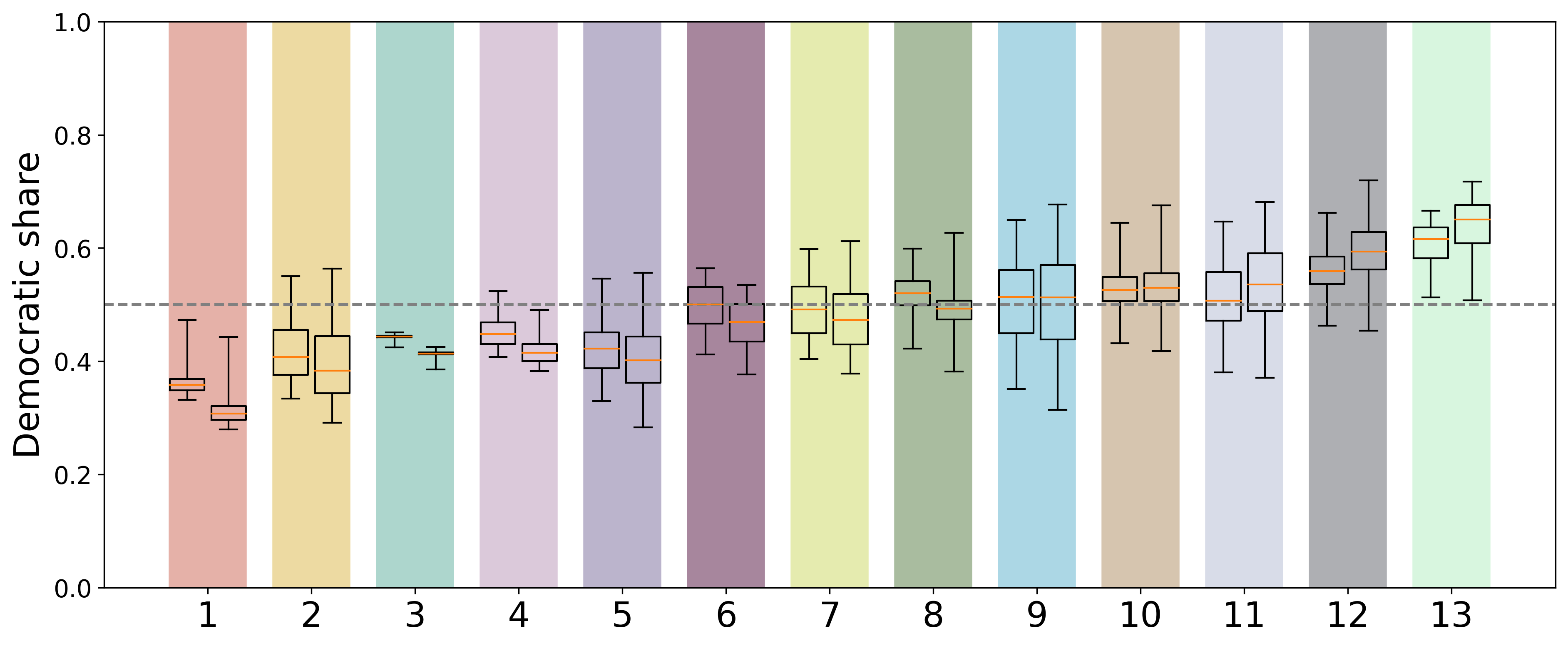}};
    \end{scope}
    \end{tikzpicture}
    }
    \caption{Barycenter for one ensemble of Congressional plans for North Carolina, with area-weighted representations used for the districts, approximated by $40$-point samples. The heat maps show the location of districts matched to each component of the barycenter. The boxplots show vote shares for the ensemble using two consecutive Presidential elections: 2012 on the left and 2016 on the right.}
    \label{fig:neutralarea}
\end{figure}

\begin{figure}
    \centering
    \begin{tikzpicture}
    \begin{scope}
    \node at (0,0) {\includegraphics[width=0.28\textwidth]{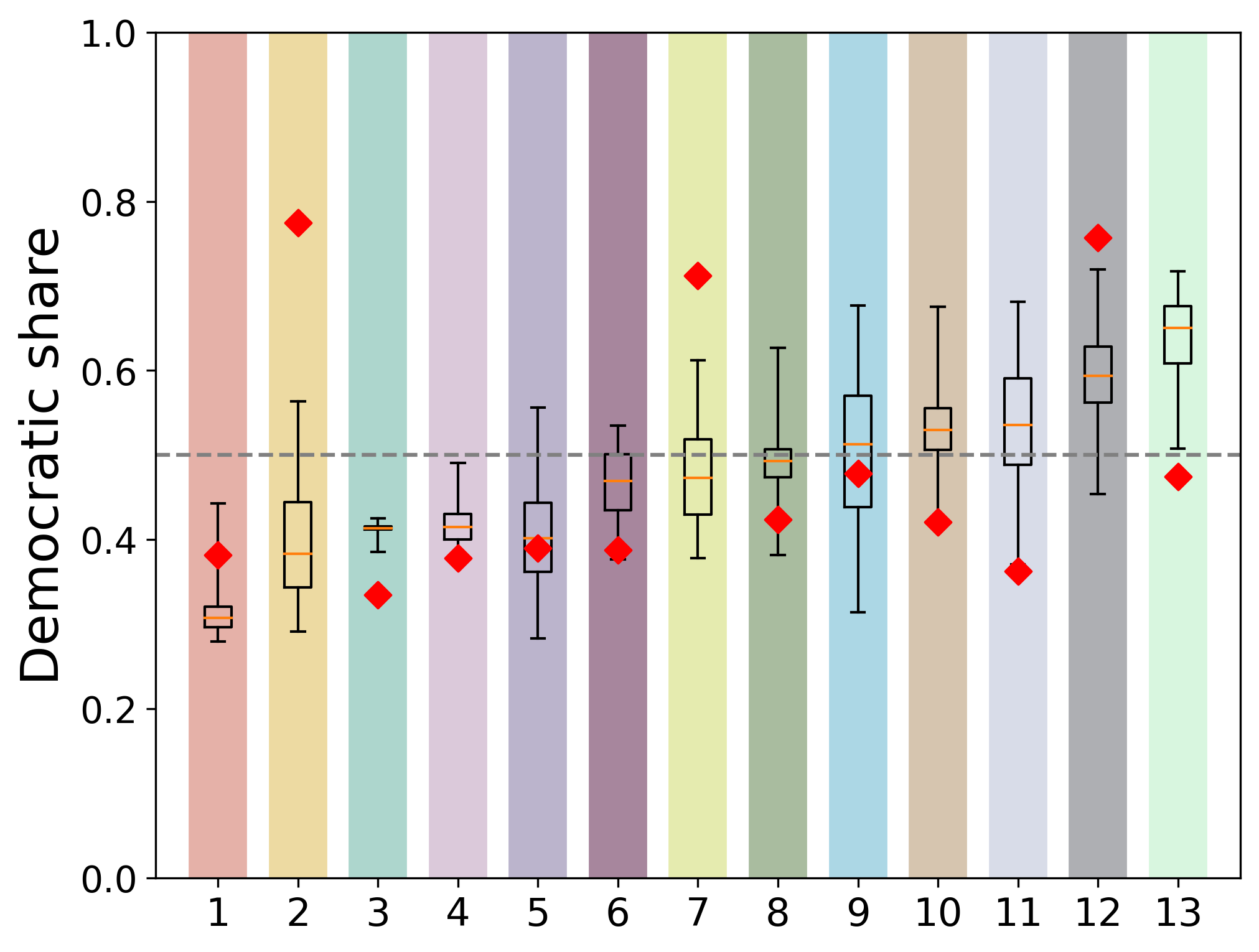}};
    \node at (3.6, 0.8) {2012};
    \node at (3.6,0) {\includegraphics[width=0.2\textwidth]{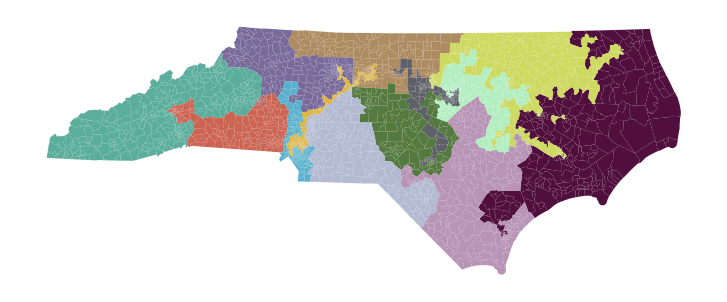}};
    \end{scope}
    
    \begin{scope}[xshift=7.4cm]
    \node at (0,0) {\includegraphics[width=0.28\textwidth]{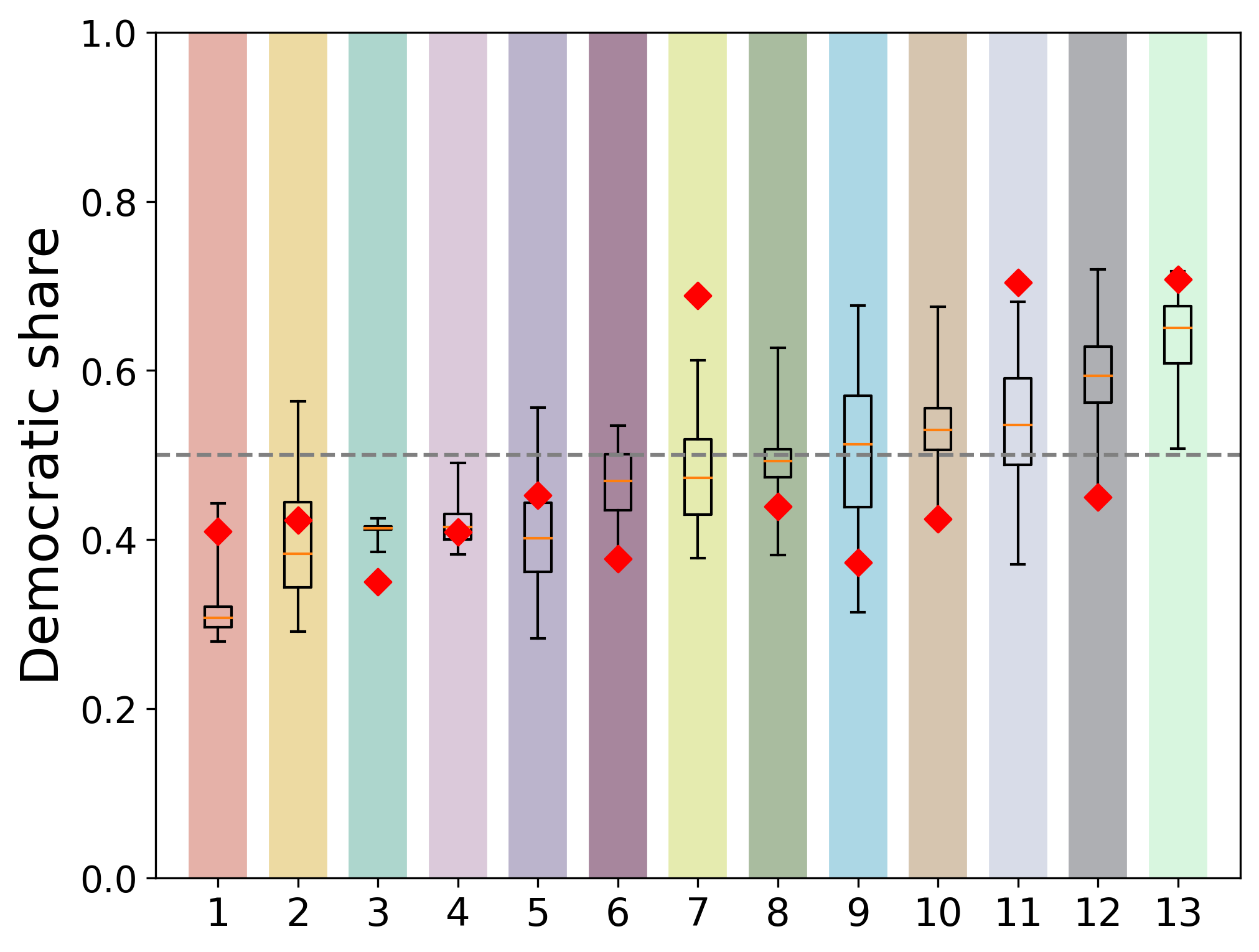}};
    \node at (3.6, 0.8) {2016};
    \node at (3.6,0) {\includegraphics[width=0.2\textwidth]{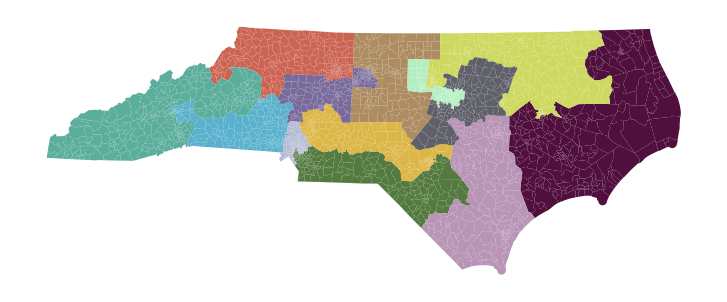}};
    \end{scope}
    
    \begin{scope}[yshift=-3.8cm]
    \node at (0,0) {\includegraphics[width=0.28\textwidth]{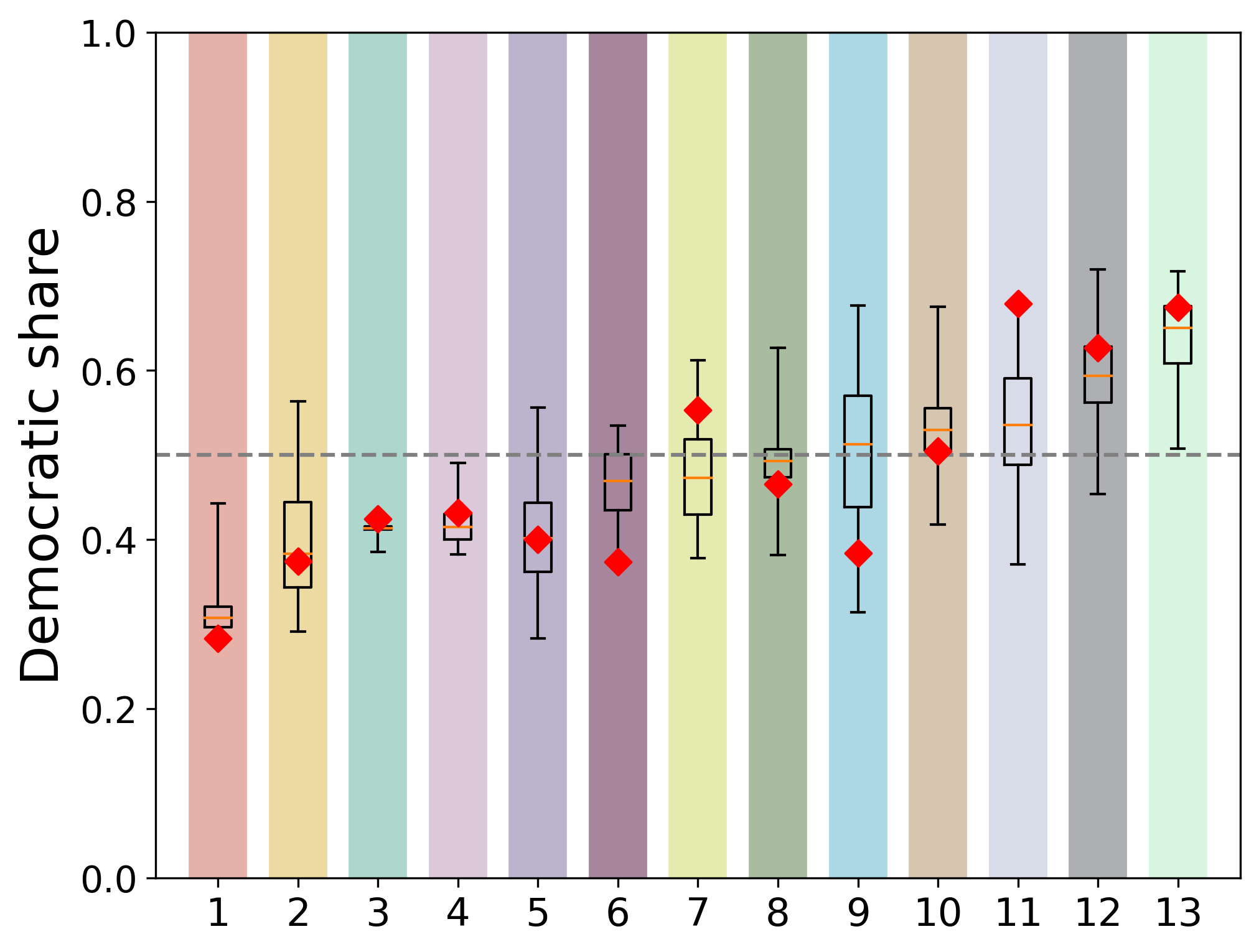}};
    \node at (3.6, 0.8) {Judges};
    \node at (3.6,0) {\includegraphics[width=0.2\textwidth]{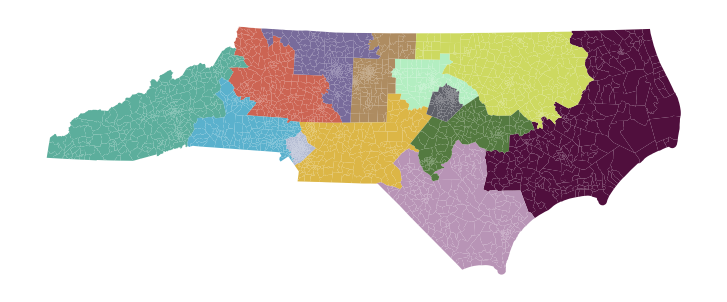}};
    \end{scope}
    
    \begin{scope}[yshift=-3.8cm, xshift=7.4cm]
    \node at (0,0) {\includegraphics[width=0.28\textwidth]{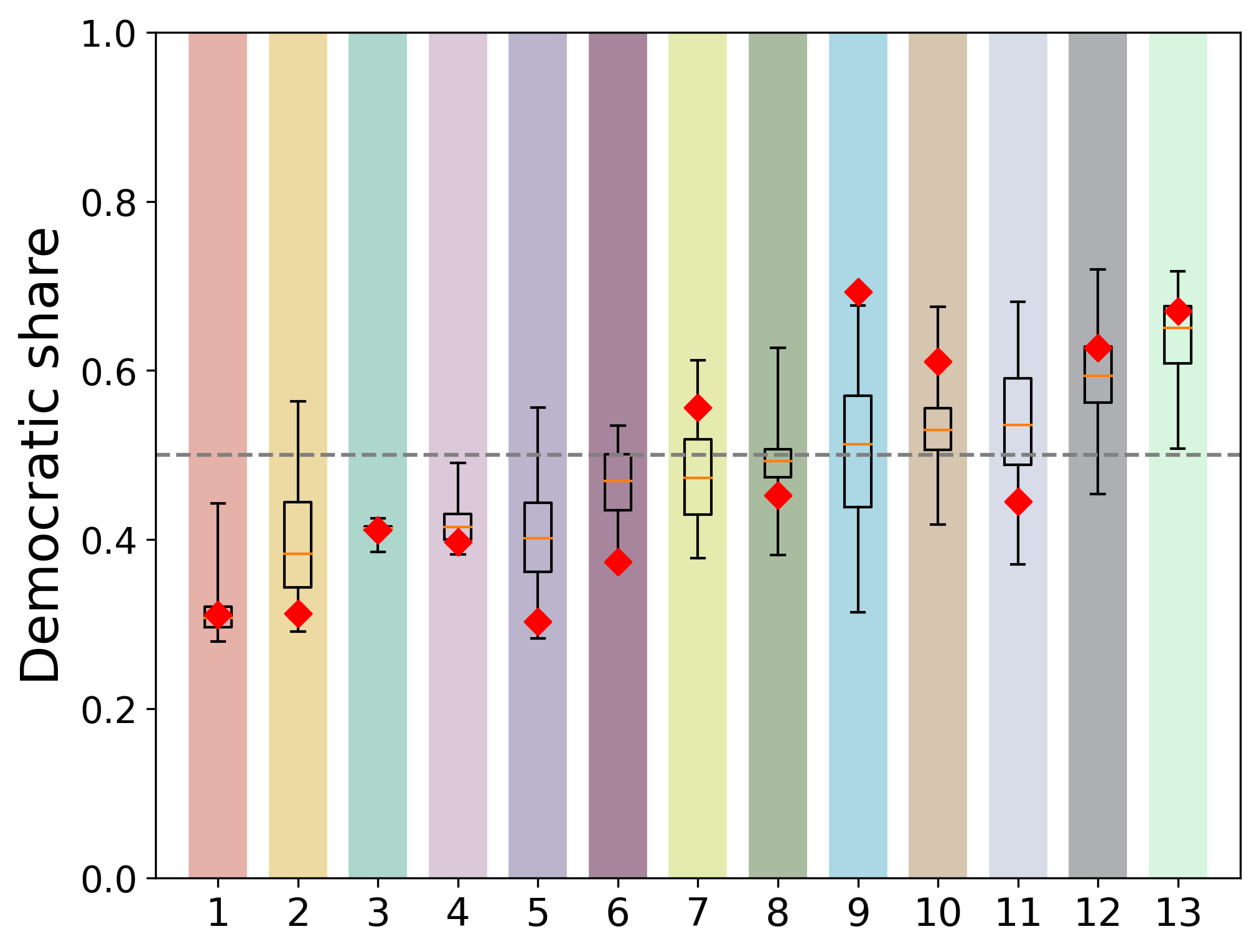}};
    \node at (3.6, 0.8) {2020};
    \node at (3.6,0) {\includegraphics[width=0.2\textwidth]{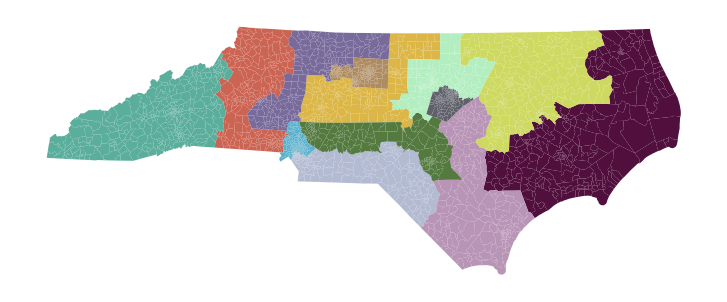}};
    \end{scope}
    
    \end{tikzpicture}
    \caption{Comparison of four enacted or proposed Congressional plans for North Carolina using votes shares from the Presidential 2016 race. Area-weighting was used. The red diamonds in each boxplot indicate the vote shares of the plan being evaluated for the district matched to that component of the barycenter. The maps on the right show the plan in question colored by a best matching to the barycenter.}\label{fig:enactedbyarea}
\end{figure}

\begin{figure}
    \centering
    \foreach \s in {AZ,CO,GA,IA,LA,MA,MD,MI,MN,NC,NE,NM,OH,OK,OR,PA,UT,VA,WI}{
        \begin{subfigure}{0.195\textwidth}
            \s 
            
            \includegraphics[width=\textwidth]{otherstates/\s_barycenter.png}
        \end{subfigure}
    }
    \caption{Area-weighted ensemble barycenters for a selection of states using precinct data from \cite{mggg_states}. Congressional apportionment from the 2020 Census was used.}
    \label{fig:otherstates_area}
\end{figure}

\begin{figure}
    \centering
    \foreach \s in {AZ,CO,GA,IA,LA,MA,MD,MI,MN,NC,NE,NM,OH,OK,OR,PA,UT,VA,WI}{
        \begin{subfigure}{0.195\textwidth}
            \s 
            
            \includegraphics[width=\textwidth]{otherstates/\s_barycenterbypop.png}
        \end{subfigure}
    }
    \caption{Population-weighted ensemble barycenters for a selection of states using precinct data from \cite{mggg_states}. Congressional apportionment from the 2020 Census was used.}
    \label{fig:otherstates_pop}
\end{figure}

\end{document}